\documentclass[10pt]{article}

\usepackage[margin=1in]{geometry}
\usepackage[T1]{fontenc}
\usepackage[english]{babel}
\usepackage[round]{natbib} 
\usepackage[dvipsnames]{xcolor}
\usepackage{amsmath, amssymb, amsthm}
\usepackage{bbm}
\usepackage{blkarray}
\usepackage{cancel}
\usepackage{enumitem}
\usepackage{float}
\usepackage{bm}
\usepackage{graphicx}
\usepackage{mathrsfs}
\usepackage{microtype}
\usepackage{ragged2e}
\usepackage{thmtools}
\usepackage{tikz}
\usepackage[Symbolsmallscale]{upgreek}
\usepackage{caption}
\usepackage{subcaption}
\usepackage[pagebackref]{hyperref}
\hypersetup{citecolor=Blue, colorlinks=true, linkcolor=purple}
\usepackage{fancyhdr}
\usepackage{algorithm}
\usepackage{algpseudocode}
\usepackage{comment}
\usepackage{subfiles}
\usepackage{mathtools} 
\mathtoolsset{showonlyrefs}

\usepackage{silence}
\def\spacingset#1{\renewcommand{\baselinestretch}%
    {#1}\small\normalsize} \spacingset{1}

\DeclareUnicodeCharacter{2161}{II} 

\newtheorem{theorem}{Theorem}[section]
\newtheorem{proposition}[theorem]{Proposition}

\newtheorem{lemma}[theorem]{Lemma}
\newtheorem{corollary}[theorem]{Corollary}

\newtheorem{assumption}[theorem]{Assumption}

\newtheoremstyle{recall}
{3pt}{3pt}{}{}{\bfseries}{.}{.5em}
{\thmname{#1}\thmnote{ #3}}

\title{ \Large \textbf{Recursive Maximum Likelihood Estimation for Interacting Particle Systems using Virtual Particles}}
\author{
Louis Sharrock\thanks{Department of Statistical Science, University College London. \texttt{l.sharrock@ucl.ac.uk}} \and 
Nikolas Kantas\thanks{Department of Mathematics, Imperial College London. \texttt{n.kantas@imperial.ac.uk}} \and 
Grigorios A. Pavliotis\thanks{Department of Mathematics, Imperial College London. \texttt{g.pavliotis@imperial.ac.uk}}
}

\date{}

\begin{document}
\maketitle

\spacingset{1.0}

\begin{abstract}
We study recursive maximum likelihood estimation for stochastic interacting particle systems based on continuous observation of a single particle. In this regime, consistent estimation of the finite-particle log-likelihood is not possible, even in the limit as the number of particles $N\rightarrow\infty$ and the time horizon $t\rightarrow\infty$. We thus seek to optimise the stationary log-likelihood of the limiting mean-field system. We achieve this via a form of stochastic gradient estimate in continuous time, with stochastic gradient estimates computed using the continuous trajectory of the single observed particle, alongside a virtual interacting particle system and a virtual tangent interacting particle system, which are integrated with the online parameter estimate. For fixed numbers of real and virtual particles, we show that the resulting algorithms drive the gradient of a finite-particle surrogate objective to zero as $t\to\infty$. We then prove that, in the iterated limit $t\to\infty$ followed by $N,M\to\infty$, these surrogate gradients converge uniformly to the gradient of the stationary log-likelihood of the limiting mean-field system, yielding convergence to its stationary points. We illustrate the method on several numerical examples, including a model with quadratic confinement and interaction potentials, a model of interacting FitzHugh--Nagumo neurons, and a stochastic Kuramoto model.
\end{abstract}

\section{Introduction}

Interacting particle systems provide a flexible framework for modelling populations of particles whose dynamics depend on collective behaviour. Such models arise in a broad range of applications, including statistical physics \citep{benedetto1997kinetic}, collective behaviour and multi-agent systems \citep{benachour1998nonlinear,canuto2012eulerian}, mean-field games and stochastic control \citep{buckdahn2017meanfield,cardaliaguet2018mean,carmona2018probabilistic,cardaliaguet2019master}, nonlinear filtering \citep{crisan2010approximate}, neuroscience and other problems in mathematical biology \citep{burger2007aggregation,baladron2012meanfield}, opinion dynamics and social interaction models \citep{chazelle2017wellposedness,goddard2022noisy}, quantitative finance \citep{giesecke2020inference}, Bayesian computation \citep{liu2016stein}, and the mean-field analysis of neural networks \citep{mei2018mean,sirignano2020mean,hu2020meanfield,rotskoff2022trainability}.

Since the seminal work of \citet{mckean1966class}, together with other early contributions such as \cite{vlasov1968vibrational,oelschlager1984martingale,sznitman1991topics,meleard1996asymptotic}, the theory of interacting particle systems and their mean-field limits has grown into a substantial modern literature. Existing results cover, amongst other topics, well-posedness \citep[e.g.,][]{huang2019distribution,chaudruderaynal2020strong}, strong existence and uniqueness \citep[e.g.,][]{jourdain2008nonlinear,bauer2018strong,mishura2020existence}, long-time behaviour \citep[e.g.,][]{carrillo2006contractions,cattiaux2008probabilistic,bolley2013uniform,eberle2019quantitative,bashiri2020longtime}, and large-particle limits \citep[e.g.,][]{malrieu2001logarithmic,malrieu2003convergence,durmus2020elementary,lacker2023sharp}. 

In parallel, there has been significant progress regarding statistical inference for this class of processes \citep[e.g.,][]{kasonga1990maximum,bishwal2011estimation,giesecke2020inference,chen2021maximum,sharrock2022parameterestimation,sharrock2023online,amorino2023parameter,dellamaestra2023lan,jasra2025bayesian,nickl2025bayesian}. In particular, there is now a substantial body of work devoted to likelihood-based parameter estimation for interacting particle systems and their mean-field limits, based on continuous or discrete observation of the entire particle system, or repeated observation of mean-field trajectories \citep[e.g.,][]{kasonga1990maximum,bishwal2011estimation,giesecke2020inference,chen2021maximum,sharrock2022parameterestimation,sharrock2023online,amorino2023parameter,dellamaestra2023lan,iguchi2025parameter,sharrock2026efficient}. A growing number of papers also consider statistical inference in the sparse observation regime, i.e., in the case where it is only possible to observe a small number of particles from the interacting particle system, or a small number of trajectories from the (stationary) mean-field system \citep[e.g.,][]{genoncatalot2022inference,genoncatalot2023parametric,pavliotis2022eigenfunction,pavliotis2024method,pavliotis2026fourierbased}. This regime is well motivated in practice. For example, it may be computationally very expensive to perform estimation based on a large number of particles, or else very expensive to accurately measure more than a small number of particles from the system due to financial or physical constraints.

In this paper, we focus on the second of these settings, namely the case in which the cost of acquiring accurate measurements dominates the computational cost of implementation. In such situations, an estimator that requires observation of only a \emph{single} particle trajectory is preferable to one that requires even a small number of observed particles (e.g., two or three), even if this comes at the price of additional computation. We are also particularly interested in {online} or {recursive} methods, which update the parameter estimate in real time as observations arrive. These methods contrast with {offline} or {batch} approaches, which typically require optimisation of a function, such as the log-likelihood, over the entire observed data path, and may become impractically slow when data are collected over long time horizons or when the model itself is costly to evaluate. Online methods have attracted considerable recent interest \citep[e.g.,][]{sirignano2017stochastic,sirignano2020stochastic,bhudisaksang2021online,bourguin2026quantitative}, including in recent work by the present authors on online parameter estimation for interacting particle systems \citep{sharrock2023online,sharrock2026efficient}. 

Of particular relevance to the current work is \citet{sharrock2026efficient}, which introduced an online estimator requiring observation of three particles from the data-generating interacting particle system: one primary particle and two auxiliary particles. Under mild assumptions, this estimator was shown to converge to a stationary point of the asymptotic log-likelihood of the interacting particle system and, under additional conditions, to the true parameter value. Crucially, however, this approach relied on access to two auxiliary particles to obtain a consistent estimate of the gradient of the asymptotic log-likelihood and hence a consistent estimator of the true parameter. Thus, the question of whether it is possible to obtain a comparable online estimator which only requires observation of a \emph{single} particle trajectory remains open. In this paper, we provide a positive answer to this question.

\subsection{Contributions}
Our main contributions are summarised below.
 
\begin{itemize}
\item We introduce two new algorithms for online parameter estimation in interacting particle systems that require observation of only a single particle trajectory. In addition to the observed trajectory, both estimators use a collection of virtual particles, which are integrated alongside the current parameter estimate.
\item For fixed numbers of real and virtual particles, we show that both estimators drive the gradient of a finite-particle surrogate objective to zero as $t\to\infty$. We then prove that, in the iterated limit $t\to\infty$ followed by $N,M\to\infty$, these surrogate gradients converge to the gradient of the stationary log-likelihood of the limiting mean-field process.
\item We illustrate the performance of the estimators in several numerical experiments, including a model with quadratic confinement and quadratic interaction, a model of interacting FitzHugh--Nagumo neurons, and a stochastic Kuramoto model.
\end{itemize}

\subsection{Related Work}
\label{sec:related-work}

There is now a substantial body of work on recursive parameter estimation for continuous-time stochastic processes. Early contributions include \citet{gerencser1984continuoustime,levanony1994recursive,gerencser2009recursive}. More recently, \citet{sirignano2017stochastic,sirignano2020stochastic} introduced \emph{stochastic gradient descent in continuous time}, an efficient method for statistical inference in fully observed diffusion processes, and established almost sure convergence and asymptotic normality; see also \citet{bourguin2026quantitative} for a more recent analysis based on Malliavin calculus. Related ideas have since been developed for partially observed diffusions \citep{surace2019online,sharrock2022theory,sharrock2022twotimescale, sharrock2022joint,sharrock2023twotimescale}, jump diffusions \citep{bhudisaksang2021online}, nonlinear diffusions and interacting particle systems \citep{sharrock2023online,sharrock2026efficient}, and models driven by coloured noise \citep{pavliotis2025filtered}. Closely related continuous-time gradient methods for objectives defined through the stationary law of a diffusion have also been proposed and analysed by \citet{wang2022forward,wang2024continuous}.

In this paper, our focus is on the application of such techniques to statistical inference for interacting particle systems and their mean-field limits. In recent years, there has been growing interest in this problem, building on early contributions by \citet{kasonga1990maximum,bishwal2011estimation,giesecke2020inference}. Recent contributions include new results on maximum likelihood estimation \citep{chen2021maximum,sharrock2022parameterestimation,dellamaestra2023lan}, online estimation procedures \citep{sharrock2023online}, local asymptotic normality \citep{dellamaestra2023lan,heidari2025local}, joint estimation of drift and diffusion coefficients \citep{amorino2023parameter}, and estimation for weakly interacting hypoelliptic diffusions \citep{iguchi2025parameter}. Whilst less directly related to our work, we note in passing that a number of authors have also considered nonparametric and semiparametric approaches to statistical inference in mean-field processes. See, e.g., \cite{lu2019nonparametric,lu2021learning,dellamaestra2022nonparametric,yao2022meanfield,lang2021identifiability,comte2022nonparametric,amorino2025polynomial,belomestny2024nonparametric,comte2024nonparametric,nickl2025bayesian,pavliotis2026fourierbased}.

Many of the aforementioned works operate in a dense observation regime, whereby it is assumed possible to observe the entire interacting particle system, or else multiple i.i.d.\ trajectories of the limiting mean-field dynamics. On the other hand, several papers have considered sparse observation regimes. These include \citet{genoncatalot2022inference,genoncatalot2023parametric}, who developed parametric procedures for statistical inference in ergodic McKean-Vlasov processes based on observation of a single trajectory, and \citet{comte2024nonparametric}, who studied nonparametric estimation in a similar framework. From the interacting particle system perspective, \citet{pavliotis2022eigenfunction,pavliotis2024method,pavliotis2026fourierbased} proposed estimators based on eigenfunction martingales, moment identities, and Fourier expansions, each using only a single observed particle. Most recently, \citet{sharrock2026efficient} introduced an online parameter-estimation method based on the observation of three continuous trajectories from the interacting particle system (see below for a more detailed comparison). Partial observations have also been considered in complementary settings, including kinetic interacting particle systems with incomplete discrete data \citep{amorino2024kinetic} and partially observed nonlinear diffusions estimated via likelihood-based or Bayesian multilevel methods \citep{jasra2025parameter, jasra2025bayesian}.

The closest methodological neighbour to the present work is the recent paper by the same authors \citep{sharrock2026efficient}. In that paper, we also study online maximum likelihood estimation for interacting particle systems in the sparse-observation regime. There, however, the estimator is defined with respect to the asymptotic, jointly in time and in the number of particles, log-likelihood of the \emph{interacting particle system} itself. This yields an algorithm which, in order to obtain consistent parameter estimates, requires observation of \emph{three} particles. By contrast, in the present paper the estimator is defined with respect to the asymptotic, now only in time, log-likelihood of the limiting \emph{mean-field system}. This yields an estimator that requires observation of only a \emph{single} particle from the interacting particle system, together with a collection of virtual particles propagated at the current parameter estimate.

\subsection{Paper Organisation}
The remainder of this paper is organised as follows. In Section~\ref{sec:prelims}, we define notation, the model, and the likelihood. In Section~\ref{sec:methodology}, we present our main methodological contributions. In Section~\ref{sec:theory}, we state our assumptions and our main theoretical results. In Section~\ref{sec:numerics}, we present several numerical examples illustrating our proposed methodology. Finally, in Section~\ref{sec:conclusion}, we provide some concluding remarks. 

\section{Preliminaries}
\label{sec:prelims} 

\subsection{Notation}
\paragraph{Norms and Inner Products.} We use $\langle\cdot,\cdot\rangle$ and $\|\cdot\|$ to denote, respectively, the Euclidean inner product and the Euclidean norm on $\mathbb{R}^d$. For a symmetric positive definite matrix $A$, we write $\|z\|^2_A :=\langle z, A^{-1}z \rangle$. For matrices and higher order tensors, we use $\|\cdot\|$ to denote the Frobenius norm. Finally, we write $\|\cdot\|_q$ to denote the $\ell^q$ norm. 

\paragraph{Probability Measures.}
We write $\mathcal{P}(\mathbb{R}^d)$ for the collection of all probability measures on $\mathbb{R}^d$. In addition, for $q\geq 1$, we write
\begin{equation}
\mathcal{P}_q(\mathbb{R}^d):=\{\mu\in\mathcal{P}(\mathbb{R}^d):\int_{\mathbb{R}^d}\|x\|^q\,\mu(\mathrm{d}x)<\infty\}.
\end{equation}
for the set of all probability measures with finite $q^{\text{th}}$ moment. For $\mu\in\mathcal{P}_q(\mathbb{R}^d)$, we occasionally write $\mu(\|\cdot\|^q)=\int_{\mathbb{R}^d}\|x\|^q\,\mu(\mathrm{d}x)$ for the $q$th moment of $\mu$.

\paragraph{Signed Measures.} We write $\mathcal{M}(\mathbb{R}^d)$ for the collection of all finite signed Borel measures on $\mathbb{R}^d$. For $\eta\in\mathcal{M}(\mathbb{R}^d)$, we write $|\eta|$ for its total variation measure. In addition, for $r\geq 0$, we define
\begin{equation}
\mathcal{M}_r(\mathbb{R}^d)
:=
\left\{
\eta\in\mathcal{M}(\mathbb{R}^d)
:
\int_{\mathbb{R}^d}(1+\|x\|^r)\,|\eta|(\mathrm{d}x)<\infty
\right\}.
\end{equation}
For $\eta\in\mathcal{M}_r(\mathbb{R}^d)$, we define the polynomially weighted total variation distance by
\begin{equation}
\|\eta\|_{\mathrm{TV},r}
:=
\int_{\mathbb{R}^d}(1+\|x\|^r)\,|\eta|(\mathrm{d}x).
\end{equation}
We will also consider $\mathbb{R}^p$-valued finite signed measures of the form $\eta=(\eta_1,\dots,\eta_p)$, where $\eta_\ell\in\mathcal M(\mathbb R^d)$ for $\ell=1,\dots,p$. In this case, we write $\eta\in\mathcal M_r(\mathbb R^d;\mathbb R^p)$ if each component belongs to $\mathcal M_r(\mathbb R^d)$. We also define
$
\|\eta\|_{\mathrm{TV},r}
:=
\sum_{\ell=1}^p \|\eta_\ell\|_{\mathrm{TV},r}.
$
Finally, for a measurable map $\varphi:\mathbb R^d\to\mathbb R^m$,  we write $\langle \eta,\varphi\rangle\in\mathbb R^{p\times m}$ for the matrix whose $\ell$-th row is given by $\int_{\mathbb R^d}\varphi(x)\,\eta_\ell(\mathrm d x)$, whenever these integrals are finite.

\paragraph{The Wasserstein Distance.} For $q\geq 1$ and $\mu,\nu\in\mathcal{P}_q(\mathbb{R}^d)$, we write $\mathsf{W}_{q}(\mu,\nu)$ for the Wasserstein distance of order $q$, namely,
\begin{equation}
\mathsf{W}_{q}(\mu,\nu)
:=
\inf_{\pi\in\Pi(\mu,\nu)}
\left[\int_{\mathbb{R}^d\times\mathbb{R}^d}\|x-y\|^q\,\pi(\mathrm{d}x,\mathrm{d}y)\right]^{\frac{1}{q}},
\end{equation}
where $\Pi(\mu,\nu)$ denotes the set of all couplings of $\mu$ and $\nu$. That is, the set of all probability measures on $\mathbb{R}^d\times\mathbb{R}^d$ with marginals $\mu$ and $\nu$.
For $r\geq 0$ and $\mu,\nu\in\mathcal P_{r+1}(\mathbb R^d)$, we also define the weighted $1$-Wasserstein discrepancy
\begin{equation}
\mathsf W_{1,r}(\mu,\nu)
:=
\inf_{\Gamma\in\Pi(\mu,\nu)}
\int_{\mathbb R^d\times\mathbb R^d}
\bigl(1+\|u\|^r+\|v\|^r\bigr)\|u-v\|\,\Gamma(\mathrm d u,\mathrm d v).
\label{eq:weighted-wasserstein}
\end{equation}

\subsection{The Model}

\subsubsection{The Interacting Particle System}
We consider a weakly interacting particle system (IPS) on $\mathbb{R}^d$, parameterised by $\theta\in\Theta\subseteq\mathbb{R}^p$, where $\Theta$ is an open set, of the form
\begin{align}
\mathrm{d}x_t^{\theta,i,N} &= 
    \Big[\frac{1}{N}\sum_{j=1}^N b(\theta,x_t^{\theta,i,N},x_t^{\theta,j,N})\Big]\mathrm{d}t+ \sigma\mathrm{d}w_t^{i,N}, \label{eq:IPS} 
\end{align}
where $\sigma\in\mathbb{R}^{d\times d}$ is a constant and invertible matrix, and $w^{i,N}:=(w_t^{i,N})_{t\geq0}$ are a set of independent $\mathbb{R}^d$-valued standard Brownian motions. We assume that $\smash{(x_0^{i,N})_{i=1}^N}$ are a set of i.i.d. $\mathbb{R}^d$-valued random variables, with common law $\mu_0$, independent of $(w_t^{i,N})_{t\geq 0}$. 

Let $\smash{\mu_t^{\theta,N}:=\frac{1}{N}\sum_{j=1}^N\delta_{x_t^{\theta,j,N}}}$ denote the empirical law of the IPS. In addition, suppose that we define the function $B(\theta,x,\mu) := \int_{\mathbb{R}^d} b(\theta,x,y)\mu(\mathrm{d}y)$. Using this notation, the IPS above can be written in the form
\begin{align}
\mathrm{d}x_t^{\theta,i,N} &= 
    B(\theta,x_t^{\theta,i,N},\mu_t^{\theta,N})\mathrm{d}t+ \sigma\mathrm{d}w_t^{i,N}. \label{eq:IPS-rewrite} 
\end{align}
It will sometimes be useful to view this IPS as an SDE on $(\mathbb{R}^d)^N$. In particular, suppose that we write $\boldsymbol{x}_t^{\theta,N} = (x_t^{\theta,1,N},\dots,{x}_t^{\theta,N,N})^{\top}\in(\mathbb{R}^d)^N$. Then this process is the solution of
\begin{equation}
    \label{eq:vector-sde}
    \mathrm{d}\boldsymbol{x}_t^{\theta,N} = {B}^N(\theta,\boldsymbol{x}_t^{\theta,N})\mathrm{d}t + \Sigma_N \mathrm{d}\boldsymbol{w}_t^N,
\end{equation}
where $\Sigma_N = \boldsymbol{I}_N\otimes \sigma$, $\boldsymbol{w}^N = (w^{1,N},\dots,w^{N,N})^{\top}$ is a $(\mathbb{R}^d)^N$-valued standard Brownian motion, and the function $B^N(\theta,\cdot):(\mathbb{R}^d)^N\rightarrow(\mathbb{R}^d)^N$ is defined according to the form $ B^N(\theta,\boldsymbol{x}^{N}) = (B^{1,N}(\theta,\boldsymbol{x}^{N}),\dots,B^{N,N}(\theta,\boldsymbol{x}^{N}))^{\top}$, where, for each $i\in[N]:=\{1,\dots,N\}$, the function 
 $B^{i,N}(\theta,\cdot):(\mathbb{R}^d)^N\rightarrow\mathbb{R}^d$ is defined according to $B^{i,N}(\theta,\boldsymbol{x}^{N}) =B(\theta,x^{i,N},\mu^N)$.
 
We will assume that there exists a true, static parameter $\theta_0\in\Theta$ which generates the observed system $\smash{(x_t^{i,N})_{t\geq 0}:= (x_t^{\theta_0,i,N})_{t\geq 0}}$. Thus, we operate under the exact modelling regime, and in our notation will suppress the dependence of the observed path on the true parameter $\theta_0$.

\subsubsection{The McKean-Vlasov SDE} 
We are interested in the case where the number of particles $N\gg 1$ so that, under appropriate conditions \citep[e.g.,][]{malrieu2001logarithmic,cattiaux2008probabilistic}, any single particle in the IPS can be well approximated by the solution of the limiting McKean-Vlasov SDE (MVSDE), namely, 
\begin{align}
\mathrm{d}x_t^{\theta}&= B(\theta,x_t^{\theta},\mu_t^{\theta})\mathrm{d}t + \sigma\mathrm{d}w_t,\quad \mu_t^{\theta} := \mathrm{Law}(x_t^{\theta}) \label{eq:MVSDE} 
\end{align}
where $w=(w_t)_{t\geq 0}$ is a standard $\mathbb{R}^d$-valued Brownian motion. This phenomenon is known as the \emph{propagation of chaos} \citep{sznitman1991topics,chaintron2022propagation,chaintron2022propagationa}.

\subsection{Model Assumptions}
\label{sec:model-assumptions}
We are now ready to state some initial assumptions on the model. We will impose these assumptions throughout.

\begin{assumption}
\label{assumption:moments}
    The initial law satisfies $\mu_0\in\mathcal{P}_q(\mathbb{R}^d)$ for all $q\in\mathbb{N}$.
\end{assumption}

\begin{assumption}
\label{assumption:model}
    For each $\theta\in\Theta$, the IPS \eqref{eq:IPS} and the MVSDE \eqref{eq:MVSDE} are well posed, admit unique invariant measures $\pi_{\theta}^N$ and $\pi_{\theta}$, and satisfy uniform-in-time moment bounds together with uniform-in-time propagation of chaos.
\end{assumption}

\begin{assumption}
\label{assumption:drift}
    For each $\theta\in\Theta$, the drift is continuous, locally Lipschitz, and of polynomial growth in both spatial variables. The same is also true of $\partial_{\theta} b(\theta,\cdot,\cdot)$, $\partial_x b(\theta,\cdot,\cdot)$, and $\partial_{y}b(\theta,\cdot,\cdot)$.
\end{assumption}

Our assumptions are deliberately stated in some generality. In typical models, they can be verified under standard dissipativity or convexity assumptions on the confinement and interaction potentials; see, e.g., \cite{malrieu2001logarithmic,malrieu2003convergence} for some classical assumptions, and \cite{carrillo2020longtime,delgadino2023phase,lacker2018strong,lacker2023hierarchies,lacker2023sharp} for some more recent results.

\subsection{The Likelihood Function}
\label{sec:log-lik}
We are interested in online inference for the unknown parameter $\theta_0$. We will perform this task based on recursive maximisation of an appropriate likelihood function. 

\subsubsection{The Log-Likelihood of the Interacting Particle System}
\label{sec:log-lik-ips}
Let $\mathbb{P}_t^{\theta,N}$ denote the probability measure induced by the trajectories $\smash{(x_s^{\theta,i,N})_{s\in[0,t]}^{i\in[N]}}$ of the IPS. Then, using Girsanov's Theorem \citep[e.g.,][]{oksendal2003stochastic},
we have a log-likelihood function given (up to an additive constant) by \citep[e.g.,][]{kasonga1990maximum,dellamaestra2023lan}
\begin{align}
\mathcal{L}_t^{N}(\theta)
&=\sum_{i=1}^N\Big[ \int_0^t \big \langle  B(\theta,x_s^{i,N},\mu_s^{N}), (\sigma\sigma^{\top})^{-1}\mathrm{d}x_s^{i,N}\big\rangle - \frac{1}{2} \int_0^t \| B(\theta,x_s^{i,N},\mu_s^{N})\|_{\sigma\sigma^{\top}}^2\mathrm{d}s \Big]. \label{ips:log-likelihood}
\end{align}
The behaviour of this function in the joint limit as $N\rightarrow\infty$ and $t\rightarrow\infty$ is the subject of the following result.

\begin{proposition}
\label{prop:ips-likelihood-n-t-limit}
    Suppose that Assumption \ref{assumption:moments}, Assumption \ref{assumption:model}, and Assumption~\ref{assumption:drift} hold. Then, as $N\rightarrow\infty$ and then $t\rightarrow\infty$, it holds that
    \begin{align}
         \frac{1}{Nt} \left[\mathcal{L}_t^N(\theta) - \mathcal{L}_t^{N}(\theta_0)\right]
         &\stackrel{L^1}{\longrightarrow} -\mathcal{L}(\theta),\qquad \mathcal{L}(\theta):=\frac{1}{2}\int_{\mathbb{R}^d} \|B(\theta,x,\pi_{\theta_0}) - B(\theta_0,x,\pi_{\theta_0})\|_{\sigma\sigma^{\top}}^2\pi_{\theta_0}(\mathrm{d}x) \label{eq:asymptotic-ll-IPS}
    \end{align}
    where $\pi_{\theta_0}\in\mathcal{P}(\mathbb{R}^d)$ denotes the unique invariant measure of the MVSDE evaluated at the true parameter $\theta_0$.
\end{proposition}

\begin{proof}
    See Corollary 9, \citet{sharrock2026efficient}.
\end{proof}

\subsubsection{The Log-Likelihood of the McKean--Vlasov SDE}
\label{sec:log-lik-mvsde}
Let $\mathbb{P}_t^{\theta}$ denote the probability measure induced by the solution $(x_s^{\theta})_{s\in[0,t]}$ of the MVSDE \eqref{eq:MVSDE}. Then, once more appealing to Girsanov's Theorem, we have a log-likelihood function given by \citep[e.g.,][Section 2.3]{dellamaestra2023lan}
\begin{equation}
    \mathcal{L}_t(\theta) = \int_0^t \big\langle B(\theta,x_s,\mu_s^{\theta}), (\sigma\sigma^{\top})^{-1}\mathrm{d}x_s\big\rangle - \frac{1}{2}\int_0^t \big\|B(\theta,x_s,\mu_s^{\theta})\big\|^2_{\sigma\sigma^{\top}}\mathrm{d}s \label{eq:MVSDE-log-likelihood}
\end{equation}
where $(x_s)_{s\geq 0}:=(x_s^{\theta_0})_{s\geq 0}$ denotes the path of the MVSDE at the true parameter $\theta_0$. The asymptotic behaviour as the time horizon $t\rightarrow\infty$ is characterised by the following result. 

\begin{proposition}
\label{prop:mvsde-likelihood-t-limit}
    Suppose that Assumption \ref{assumption:moments}, Assumption \ref{assumption:model}, and Assumption~\ref{assumption:drift} hold.  Then, as $t\rightarrow\infty$, it holds that
    \begin{align}
        \frac{1}{t}\left[\mathcal{L}_t(\theta) - \mathcal{L}_t(\theta_0)\right]
        &\stackrel{L^1}{\longrightarrow} - \mathcal{J}(\theta), \qquad \mathcal{J}(\theta):=\frac{1}{2}\int_{\mathbb{R}^d} \| B(\theta,x,\pi_{\theta}) - B(\theta_0,x,\pi_{\theta_0})\|_{\sigma\sigma^{\top}}^2\pi_{\theta_0}(\mathrm{d}x). \label{eq:asymptotic-ll-mvsde}
    \end{align}
    where $\pi_{\theta}, \pi_{\theta_0}\in\mathcal{P}(\mathbb{R}^d)$ denote the unique invariant measures of the MVSDE, evaluated at the parameter $\theta$ and the true parameter $\theta_0$, respectively.
\end{proposition}

\begin{proof}
See Proposition 10, \citet{sharrock2026efficient}.
\end{proof}

It is worth noting that the asymptotic log-likelihood of the IPS in \eqref{eq:asymptotic-ll-IPS} differs from the asymptotic log-likelihood of the MVSDE in \eqref{eq:asymptotic-ll-mvsde}. This difference arises because the model drift in $\mathcal{L}$ is evaluated at the true invariant law $\pi_{\theta_0}$, while the model drift in $\mathcal{J}$ is evaluated at the parameter dependent invariant law $\pi_{\theta}$. Nonetheless, under standard identifiability assumptions, both functions are non-negative and uniquely minimised at the true parameter $\theta_0$ \citep[e.g.,][]{genoncatalot2023parametric}.


\section{Methodology}
\label{sec:methodology}

Our objective is to estimate the true parameter $\theta_{0}$ in an online fashion, based on the continuous stream of observations of a single particle $\smash{(x_t^{i,N})_{t\geq 0}}$ from the IPS. To achieve this task, we will seek to recursively minimise an appropriately chosen objective.

\subsection{The Objective Function}
 We are interested in the case where the number of particles $N\gg 1$, and thus any single particle in the IPS resembles a solution of the MVSDE. In this regime, there are two natural choices for the objective function.
 \begin{itemize}
     \item[(i)] The first is the asymptotic -- both in time and in the number of particles -- {negative} log-likelihood of the IPS, as defined in Proposition~\ref{prop:ips-likelihood-n-t-limit}.    
     \item[(ii)] The second is the asymptotic -- now only in time --  negative log-likelihood of the limiting MVSDE, as defined in Proposition~\ref{prop:mvsde-likelihood-t-limit}.
 \end{itemize} 
 In the companion to this paper, we studied algorithms designed with reference to the first of these two objectives. In this paper, we instead consider algorithms designed with reference to the second, which for convenience we recall again here in the form
\begin{align}
    \mathcal{J}(\theta)&= \int_{\mathbb{R}^d} \frac{1}{2}\langle B(\theta,x,\pi_{\theta})-B(\theta_0,x,\pi_{\theta_0}), B(\theta,x,\pi_{\theta})-B(\theta_0,x,\pi_{\theta_0})\rangle_{\sigma\sigma^{\top}} \pi_{\theta_0}(\mathrm{d}x)\\
    &:=\int_{\mathbb{R}^d}J(\theta,x,\pi_{\theta},\pi_{\theta},\pi_{\theta_0})\pi_{\theta_0}(\mathrm{d}x). \label{eq:obj-func-2}
\end{align}
By expanding and simplifying the integrand, this objective can also be written in a slightly more explicit form, namely, 
\begin{align}
\mathcal{J}(\theta)
&=
\int_{(\mathbb{R}^d)^3}
\frac{1}{2}
\big\langle
b(\theta,x,y)-B(\theta_0,x,\pi_{\theta_0}),
\,b(\theta,x,z)-B(\theta_0,x,\pi_{\theta_0})
\big\rangle_{\sigma\sigma^{\top}}
\,\pi_{\theta_0}(\mathrm{d}x)\pi_{\theta}(\mathrm{d}y)\pi_{\theta}(\mathrm{d}z) \\
&:= \int_{(\mathbb{R}^d)^3} j(\theta,x,y,z,\pi_{\theta_0})\,\pi_{\theta_0}(\mathrm{d}x)\pi_{\theta}(\mathrm{d}y)\pi_{\theta}(\mathrm{d}z).
\end{align}
As we will see below, the estimators derived with respect to this objective will have rather different properties from those obtained in \citet{sharrock2026efficient}.

\subsection{The Gradient of the Objective Function}

We would like to use a (stochastic) gradient based approach to optimise the objective function. Our first task is thus to characterise its gradient. Let $\nu_{\theta}:=\partial_{\theta}\pi_{\theta}\in\mathcal M(\mathbb R^d;\mathbb R^p)$ denote the weak derivative of the invariant law with respect to the parameter. That is,
\begin{equation}
\partial_{\theta_\ell}\int_{\mathbb R^d}\varphi(x)\,\pi_\theta(\mathrm d x)
=
\int_{\mathbb R^d}\varphi(x)\,\nu_{\theta,\ell}(\mathrm d x),
\qquad \ell=1,\dots,p,
\end{equation}
for every sufficiently regular test function $\varphi\in C^1(\mathbb R^d)$ for which the derivative exists. We will require the following integrability assumption on $(\pi_{\theta})_{\theta\in\Theta}$ and $(\nu_{\theta})_{\theta\in\Theta}$.

\begin{assumption}
\label{ass:appendix-meanfield-c1}

For all $q\geq 1$, the families $(\pi_\theta)_{\theta\in\Theta}$ and $(\nu_\theta)_{\theta\in\Theta}$ have uniformly bounded moments of order $q$, where for $\nu_\theta$ this is understood in total variation.
\end{assumption}

We then have the following result.

\begin{proposition}
\label{prop:ips-likelihood-n-t-limit-grad}
    Suppose that Assumption \ref{assumption:moments} - \ref{assumption:drift} and Assumption~\ref{ass:appendix-meanfield-c1} hold. Then the gradient of $\mathcal{J}$ with respect to $\theta$ is given by
    \begin{align}
     \partial_{\theta}\mathcal{J}(\theta) &= \int_{\mathbb{R}^d} G(\theta,x,\pi_{\theta},\nu_{\theta}) (\sigma\sigma^{\top})^{-1}\left(B(\theta,x,\pi_{\theta}) - B(\theta_0,x,\pi_{\theta_0})\right)  \pi_{\theta_0}(\mathrm{d}x) \\
     &:=\int_{\mathbb{R}^d} H(\theta,x,\pi_{\theta},\nu_{\theta},\pi_{\theta},\pi_{\theta_0})\pi_{\theta_0}(\mathrm{d}x)
        \label{eq:asymptotic-obj-grad}
    \end{align}
    where $G:\Theta\times\mathbb{R}^d\times\mathcal{P}(\mathbb{R}^d)\times\mathcal{M}(\mathbb{R}^d;\mathbb{R}^p)\rightarrow\mathbb{R}^{p\times d}$ is defined according to $G(\theta,x,\mu,\eta) := \int_{\mathbb{R}^d} 
    \partial_{\theta} b(\theta,x,y)\mu(\mathrm{d}y) + \int_{\mathbb{R}^d} b(\theta,x,z)\eta(\mathrm{d}z)$.
\end{proposition}

\begin{proof}
See Appendix~\ref{appendix:proofs-section3}.
\end{proof}

We can also obtain an alternative, more explicit representation for the gradient of the objective function.

\begin{proposition}
\label{prop:ips-likelihood-n-t-limit-grad-explicit}
     Suppose that Assumption \ref{assumption:moments} - \ref{assumption:drift} and Assumption~\ref{ass:appendix-meanfield-c1} hold. Then the gradient of $\mathcal{J}$ with respect to $\theta$ is given by
   \begin{align}
    \partial_{\theta} \mathcal{J}(\theta)
    &=\int_{(\mathbb{R}^d)^3} \big[g(\theta,x,y,\nu_{\theta}) \big]  (\sigma\sigma^{\top})^{-1} \big[b(\theta,x,z) - B(\theta_0,x,\pi_{\theta_0})\big] \pi_{\theta_0}(\mathrm{d}x) \pi_{\theta}(\mathrm{d}y)\pi_{\theta}(\mathrm{d}z)  \\
    &:=\int_{(\mathbb{R}^d)^3} h(\theta,x,y,\nu_{\theta},z,\pi_{\theta_0})\pi_{\theta_0}(\mathrm{d}x) \pi_{\theta}(\mathrm{d}y)\pi_{\theta}(\mathrm{d}z)
\end{align}
    where $g:\Theta\times\mathbb{R}^d\times \mathbb{R}^d\times\mathcal{M}(\mathbb{R}^d;\mathbb{R}^p)\rightarrow\mathbb{R}^{p\times d}$ is defined according to $g(\theta,x,y,\eta) := \partial_{\theta}b(\theta,x,y) + \int b(\theta,x,z)\eta(\mathrm{d}z)$.
\end{proposition}

\begin{proof}
See Appendix~\ref{appendix:proofs-section3}.
\end{proof}

\subsection{The Stochastic Gradient of the Objective Function}
\label{sec:sgdct}
The formulae above are exact but not implementable, since both the stationary law $\pi_{\theta}$ and its derivative $\nu_{\theta}=\partial_{\theta}\pi_{\theta}$ are unknown. In order to proceed, we thus seek stochastic estimates.

\subsubsection{The Averaged Virtual Particle Estimate}
\label{sec:stochastic-estimate-averaged}
Our first estimate is derived based on our first expression for the gradient of the objective (cf. Proposition~\ref{prop:ips-likelihood-n-t-limit-grad}). Let $(x_s)_{s\geq 0}$ denote a solution of the MVSDE, with $\mu_s=\mathrm{Law}(x_s)$. Let $(\mu_s^{\theta})_{s\geq 0}$ denote the law of a solution $(x_s^{\theta})_{s\geq 0}$ of the MVSDE evaluated at $\theta\in\Theta$, and $(\eta_s^{\theta})_{s\geq 0}:=(\partial_{\theta}\mu_s^{\theta})_{s\geq 0}$ denote the derivative of this law with respect to the parameter.  

We begin with the observation that, assuming ergodicity and convergence of $(\mu_s^{\theta})_{\theta\in\Theta}\rightarrow(\pi_{\theta})_{\theta\in\Theta}$ and $(\eta_s^{\theta})_{\theta\in\Theta}\rightarrow(\nu_{\theta})_{\theta\in\Theta}$ as $s\rightarrow\infty$, we have
\begin{align}
\partial_{\theta} \mathcal{J}(\theta)& \approx 
\displaystyle \lim_{t\rightarrow\infty}\frac{1}{t} \left[ \int_{0}^t G(\theta,x_s,\mu_s^{\theta},\eta_s^{\theta}) 
(\sigma\sigma^{\top})^{-1} \left(B(\theta,x_s,\mu_s^{\theta}) - B(\theta_0,x_s,\mu_s)\right)\mathrm{d}s\right].
\label{IPS_estimator_4-a} 
\end{align}
Substituting the true dynamics for $(x_s)_{s\geq 0}$, and noting that the additional martingale term converges to zero under our conditions, it follows that
\begin{align}
\partial_{\theta}\mathcal{J}(\theta)&\approx \lim_{t\rightarrow\infty} \frac{1}{t}\left[ \int_0^t G(\theta,x_s,\mu_s^{\theta},\eta_s^{\theta}) (\sigma\sigma^{\top})^{-1} \big[B(\theta,x_s,\mu_s^{\theta})\mathrm{d}s - \mathrm{d}x_s \big]  \right].
\end{align}
Finally, assuming uniform-in-time propagation of chaos for the IPS and the tangent IPS, uniformly in $\theta\in\Theta$, 
we arrive at
\begin{align}
 \nabla_{\theta}\mathcal{J}(\theta)&\approx \lim_{t\rightarrow\infty} \lim_{N,M\rightarrow\infty} \frac{1}{t} \left[ \int_0^t G(\theta,x_s^{i,N},{\mu}_s^{\theta,M},{\eta}_s^{\theta,M}) (\sigma\sigma^{\top})^{-1} \Big( B(\theta,x_s^{i,N},\bar{\mu}_s^{\theta,M})\mathrm{d}s - \mathrm{d}x_s^{i,N}\Big)  \right]. 
\end{align}
where \(\vphantom{({x}_s^{\theta,j,M})_{s\geq 0}^{j\in[M]}}{\mu}_s^{\theta,M} = \frac{1}{M}\sum_{j=1}^M \delta_{{x}_s^{\theta,j,M}}\) and \(\bar{\mu}_s^{\theta,M} = \frac{1}{M}\sum_{k=1}^M \delta_{\bar{x}_s^{\theta,k,M}}\) denote the empirical laws of two independent solutions $\vphantom{{\mu}_s^{\theta,M} = \frac{1}{M}\sum_{j=1}^M \delta_{{x}_s^{\theta,j,M}}}({x}_s^{\theta,j,M})_{s\geq 0}^{j\in[M]}$ and $\vphantom{{\mu}_s^{\theta,M} = \frac{1}{M}\sum_{j=1}^M \delta_{{x}_s^{\theta,j,M}}}(\bar{x}_s^{\theta,k,M})_{s\geq 0}^{k\in[M]}$ of the IPS; and where $\eta_s^{\theta,M}=\frac{1}{M}\sum_{j=1}^M y_s^{\theta,j,M}\delta'_{x_s^{\theta,j,M}}$ is the formal derivative of the empirical measure \(\vphantom{({x}_s^{\theta,j,M})_{s\geq 0}^{j\in[M]}}{\mu}_s^{\theta,M} = \frac{1}{M}\sum_{j=1}^M \delta_{{x}_s^{\theta,j,M}}\) with respect to the parameter, with $(y_s^{\theta,j,M})_{s\ge0}^{j\in[M]}:=(\partial_{\theta} x_s^{\theta,j,M})_{s\ge0}^{j\in[M]}$ denoting the tangent IPS associated with $({x}_s^{\theta,j,M})_{s\geq 0}^{j\in[M]}$. This expression suggests that, for $N,M\gg 1$, a natural stochastic estimate for $\nabla_{\theta}\mathcal{J}(\theta_t)\mathrm{d}t \vphantom{{\mu}_s^{\theta,M} = \frac{1}{M}\sum_{j=1}^M \delta_{{x}_s^{\theta,j,M}}}$ is given by $\vphantom{{\mu}_s^{\theta,M} = \frac{1}{M}\sum_{j=1}^M \delta_{{x}_s^{\theta,j,M}}}$
\begin{align}
\partial_{\theta}\mathcal{J}(\theta_t)\mathrm{d}t &\approx G(\theta_t,x_t^{i,N},\hat{\mu}_t^M, \hat{\eta}_t^M)(\sigma\sigma^{\top})^{-1}\Big( B(\theta_t,x_t^{i,N},\tilde{\mu}_{t}^{M})\mathrm{d}t - \mathrm{d}x_t^{i,N}\Big), 
\label{eq:average-stochastic-estimate}
\end{align}
where $\hat{\mu}_t^M:= \frac{1}{M}\sum_{j=1}^M \delta_{\hat{{x}}_t^{j,M}}$ and $\tilde{\mu}_t^M:= \frac{1}{M}\sum_{k=1}^M \delta_{\tilde{x}_t^{k,M}}$ are the empirical measures of two independent solutions $(\hat{{x}}_t^{j,M})_{t\geq 0}^{j\in[M]}$ and $(\tilde{x}_t^{k,M})_{t\geq 0}^{k\in[M]}$
of the IPS, both integrated with the online parameter estimate $(\theta_t)_{t\geq 0} \vphantom{{\mu}_s^{\theta,M} = \frac{1}{M}\sum_{j=1}^M \delta_{{x}_s^{\theta,j,M}}}$, viz 
\begin{align}
\mathrm{d}\hat{{x}}_t^{j,M}&=  \bigg[\frac{1}{M}\sum_{\ell=1}^Mb(\theta_t,\hat{{x}}_t^{j,M}, \hat{{x}}_t^{\ell,M}) \bigg]\mathrm{d}t + \sigma\mathrm{d}\hat{w}_t^{j,M}, \quad t\geq 0,  \quad j\in[M], \label{eq:IPS-hat} \\
\mathrm{d}\tilde{x}_t^{k,M}&=  \bigg[\frac{1}{M}\sum_{\ell=1}^Mb(\theta_t,\tilde{x}_t^{k,M}, \tilde{x}_t^{\ell,M}) \bigg]\mathrm{d}t + \sigma\mathrm{d}\tilde{w}_t^{k,M}, \quad t\geq 0, \quad k\in[M], \label{eq:IPS-tilde} 
\end{align}
and where $\hat{\eta}_t^M=\frac{1}{M} \sum_{j=1}^M \hat{y}_t^{j,M} \delta'_{\hat{{x}}_t^{j,M}}$ is the formal derivative of $\hat{\mu}_t^M$ with respect to the parameter, also integrated with the online parameter estimate, $(\hat{y}_t^{j,M})^{j\in[M]}_{t\geq 0}$ denoting the tangent IPS associated with $(\hat{{x}}_{t}^{j,M})_{t\geq 0}^{j\in[M]}\vphantom{\hat{\eta}_t^M(\mathrm{d}x)=\frac{1}{M} \sum_{i=1}^M \hat{y}_t^{i,M} \delta'_{\hat{\boldsymbol{x}}_t^{i,M}}(\mathrm{d}x)}$, namely, 
\begin{align}
\mathrm{d}\hat{y}_{t}^{i,M} 
&= \frac{1}{M}\sum_{j=1}^M \Big[
\partial_{\theta} b(\theta_t, \hat{x}_t^{i,M}, \hat{x}_t^{j,M})
+ \hat{y}_{t}^{i,M}\nabla_x b(\theta_t, \hat{x}_t^{i,M}, \hat{x}_t^{j,M})
+ \hat{y}_t^{j,M} \nabla_y b(\theta_t,\hat{x}_{t}^{i,M}, \hat{x}_t^{j,M})
\Big]\mathrm{d}t.
\label{eq:IPS-hat-tangent}
\end{align}
We will refer to $(\hat{x}_t^{j,M})^{j\in[M]}_{t\geq 0}$ and $(\tilde{x}_t^{k,M})^{k\in[M]}_{t\geq 0}$ as \emph{virtual particles}, to distinguish them from the \emph{real particles} $(x_t^{i,N})^{i\in[N]}_{t\geq 0}$ defining the data-generating process.

\subsubsection{The Particlewise Virtual Particle Estimate}
\label{sec:stochastic-estimate-particlewise}
We can also obtain a different stochastic estimate, based on the alternative representation for the gradient of the objective (cf. Proposition~\ref{prop:ips-likelihood-n-t-limit-grad-explicit}). Let $(x_s)_{s\geq 0}$ denote a solution of the MVSDE; $(x_s^{\theta})_{s\geq 0}$ and $(\bar{x}_s^{\theta})_{s\geq 0}$ denote two independent solutions of the MVSDE, both evaluated at $\theta\in\Theta$; and $(y_s^{\theta})_{s\geq 0}:=(\partial_{\theta}x_s^{\theta})_{s\geq 0}$ denote the tangent solution of the MVSDE associated with $(x_s^{\theta})_{s\geq 0}$. 

We begin by noting that, under the assumption that $(x_s)_{s\geq 0}$ is ergodic, $(x_s^{\theta})_{s\geq 0}$ and $(y_s^{\theta})_{s\geq 0}$ are jointly ergodic, and $(\bar{x}_s^{\theta})_{s\geq 0}$ is ergodic, we have that
\begin{align}
    \partial_{\theta} \mathcal{J}(\theta)
&\approx\lim_{t\rightarrow\infty}\frac{1}{t}\bigg[\int_0^t \Big(b_{\theta}(\theta,x_s,x_s^{\theta}) + y_s^{\theta} b_{x_s^{\theta}}(\theta,x_s,x_s^{\theta}) \Big)(\sigma\sigma^{\top})^{-1} \Big(b(\theta,x_s,\bar{x}_s^{\theta}) - B(\theta_0,x_s,\mu_s)\Big)\mathrm{d}s\bigg], \nonumber
\end{align}
Substituting the true dynamics for $(x_s)_{s\geq 0}$, and using the fact that the additional martingale term converges to zero, we have that
\begin{align}
\partial_{\theta} \mathcal{J}(\theta)
&\approx \lim_{t\rightarrow\infty}\frac{1}{t}\bigg[\int_0^t \Big(b_{\theta}(\theta,x_s,x_s^{\theta}) + y_s^{\theta}  b_{x_s^{\theta}}(\theta,x_s,x_s^{\theta}) \Big)(\sigma\sigma^{\top})^{-1} \Big(b(\theta,x_s,\bar{x}_s^{\theta}) \mathrm{d}s - \mathrm{d}x_s\Big)\bigg]. 
\end{align}
Finally, under the assumption of uniform-in-time propagation of chaos, it follows from the previous display that
\begin{align}
\partial_{\theta} \mathcal{J}(\theta)
&\approx \lim_{t\rightarrow\infty}\lim_{M,N\rightarrow\infty}\frac{1}{t}\bigg[\int_0^t \Big(g(\theta,x_s^{i,N}, x_s^{\theta,j,M}, \eta_s^{\theta,j,M}) \Big)(\sigma\sigma^{\top})^{-1} \Big(b(\theta,x_s^{i,N},\bar{x}_s^{\theta,k,M})\mathrm{d}s -\mathrm{d}x_s^{i,N}\Big) \bigg] 
\end{align}
 where, similar to before, $\vphantom{\eta_s^{\theta,j,M}:= y_s^{\theta,j,M}\delta'_{x_s^{\theta,j,M}}}({x}_s^{\theta,j,M})_{s\geq 0}^{j\in[M]}$ and $(\bar{x}_s^{\theta,k,M})_{s\geq 0}^{k\in[M]}$ denote two independent solutions  of the IPS, and where $\vphantom{({x}_s^{\theta,j,M})_{s\geq 0}^{j\in[M]}}\eta_s^{\theta,j,M}:= y_s^{\theta,j,M}\delta'_{x_s^{\theta,j,M}}$ denotes the formal derivative of the single particle empirical measure $\mu_s^{\theta,j,M}:=\delta_{x_s^{\theta,j,M}}$ with respect to the parameter, with $\vphantom{\eta_s^{\theta,j,M}:= y_s^{\theta,j,M}\delta'_{x_s^{\theta,j,M}}}(y_s^{\theta,j,M})^{j\in[M]}_{s\geq 0}:=(\partial_{\theta}x_s^{\theta,j,M})^{j\in[M]}_{s\geq 0}$ denoting the solution of the tangent IPS associated with $({x}_s^{\theta,j,M})_{s\geq 0}^{j\in[M]}$. This expression suggests that, for $N,M\gg 1$, a natural stochastic estimate for $\nabla_{\theta}\mathcal{J}(\theta_t)\mathrm{d}t\vphantom{(y_s^{\theta,j,M})^{j\in[M]}_{s\geq 0}}$ is given by $\vphantom{({x}_s^{\theta,j,M})_{s\geq 0}^{j\in[M]}}$
\begin{align}
\partial_{\theta}\mathcal{J}(\theta_t)\,\mathrm{d}t
&\approx \Big(g(\theta_t,x_t^{i,N}, \hat{x}_t^{j,M},\hat{\eta}_t^{j,M})\Big)(\sigma\sigma^{\top})^{-1} \Big(b(\theta_t,x_t^{i,N},\tilde{x}_t^{k,M})\,\mathrm{d}t -\mathrm{d}x_t^{i,N}\Big).
\label{eq:particlewise-stochastic-estimate}
\end{align}
where $(\hat{{x}}_t^{j,M})_{t\geq 0}^{j\in[M]}$ and $(\tilde{x}_t^{k,M})_{t\geq 0}^{k\in[M]}$ are independent solutions of the IPS, both integrated with the online parameter estimate, and $\hat{\eta}_t^{j,M}:=\hat{y}_t^{j,M}\delta'_{\hat{x}_t^{j,M}}$ is the formal derivative of $\hat{\mu}_t^{j,M} :=\delta_{\hat{x}_t^{j,M}}$ with respect to the parameter, also integrated with the online parameter estimate, $(\hat{y}_{t}^{j,M})_{t\geq 0}^{j\in[M]}$ once more denoting the solution of the tangent IPS associated with $(\hat{{x}}_{t}^{j,M})_{t\geq 0}^{j\in[M]}\vphantom{\hat{\eta}_t^M(\mathrm{d}x)=\frac{1}{M} \sum_{i=1}^M \hat{y}_t^{j,M} \delta'_{\hat{\boldsymbol{x}}_t^{j,M}}(\mathrm{d}x)}$.

\subsection{Stochastic Gradient Descent in Continuous Time}

In order to optimise $\mathcal{J}$, a natural approach is to consider a gradient descent algorithm, in our case in continuous time. In particular, we would like to simulate
\begin{align}
\mathrm{d}{\theta}_t&= -\gamma_t \partial_{\theta}\mathcal{J}(\theta_t)\mathrm{d}t, \label{eq:continuous-time-grad-descent}
\end{align}
where $\gamma_t:\mathbb{R}_{+}\rightarrow \mathbb{R}_{+}$ is a deterministic, positive, non-increasing function known as the \emph{learning rate}. In practice, we will replace exact gradients with the stochastic estimates derived in Sections~\ref{sec:stochastic-estimate-averaged} and~\ref{sec:stochastic-estimate-particlewise}. Following the taxonomy introduced in \citet{sirignano2017stochastic}, we will refer to these algorithms as \emph{stochastic gradient descent in continuous time} (SGDCT).

\subsubsection{The Averaged Virtual Particle Estimator}
\label{sec:algorithm-1}
Substituting the stochastic estimate in \eqref{eq:average-stochastic-estimate} into \eqref{eq:continuous-time-grad-descent}, our first SGDCT algorithm is given by
\begin{align}
\mathrm{d}\theta_t &= -\gamma_t \Big(G(\theta_t,x_t^{i,N},\hat{\mu}_t^M,\hat{\eta}_t^{M})\Big)(\sigma\sigma^{\top})^{-1}\Big( B(\theta_t,x_t^{i,N},\tilde{\mu}_t^M)\mathrm{d}t - \mathrm{d}x_t^{i,N}\Big). 
\label{eq:algorithm-1} 
\end{align}
\normalsize

\subsubsection{The Particlewise Virtual Particle Estimator}
\label{sec:algorithm-2}
Substituting the stochastic estimate in \eqref{eq:particlewise-stochastic-estimate} into \eqref{eq:continuous-time-grad-descent}, our second SGDCT algorithm is given by
\begin{align}
\mathrm{d}\vartheta_t &= -\gamma_t \Big(g(\vartheta_t,x_t^{i,N},\hat{x}_t^{j,M},\hat{\eta}_t^{j,M})\Big)(\sigma\sigma^{\top})^{-1} \Big(b(\vartheta_t,x_t^{i,N},\tilde{x}_t^{k,M})\mathrm{d}t -\mathrm{d}x_t^{i,N}\Big).
\label{eq:algorithm-2} 
\end{align}
\normalsize

\subsubsection{Discussion}
\label{sec:method-discussion}
It is instructive to rewrite these update equations in a slightly different form. In particular, after substituting the particle dynamics $\mathrm{d}x_t^{i,N} = B(\theta_0,x_t^{i,N},\mu_t^N)\mathrm{d}t + \sigma\mathrm{d}w_t^{i,N}$, and reorganising, we have that 
 \begin{align}
\mathrm{d}\theta_t
&= -\underbrace{\gamma_t\nabla_{\theta}\mathcal{J}(\theta_t)\mathrm{d}t}_{\text{true descent term}} 
- \underbrace{\gamma_t \big(H(\theta_t, x_{t}^{i,N}, \hat{\mu}_t^M, \hat{\eta}_t^M, \tilde{\mu}_t^M, \mu_t^N)-\nabla_{\theta}\mathcal{J}(\theta_t)\big)\mathrm{d}t}_{\text{fluctuation term}} \\
&\quad+\underbrace{\gamma_tG(\theta_t,x_t^{i,N},\hat{\mu}_t^M,\hat{\eta}_t^M)\sigma^{-\top} \mathrm{d}w_t^{i,N}}_{\text{noise term}}  \label{eq:algorithm-1-rewrite} \\[1mm]
\mathrm{d}\vartheta_t
&= -\underbrace{\gamma_t\nabla_{\theta}\mathcal{J}(\vartheta_t)\mathrm{d}t}_{\text{true descent term}} 
- \underbrace{\gamma_t \big(h(\vartheta_t, x_{t}^{i,N}, \hat{x}_t^{j,M}, \hat{\eta}_t^{j,M}, \tilde{x}_t^{k,M}, \mu_t^N)-\nabla_{\theta}\mathcal{J}(\vartheta_t)\big)\mathrm{d}t}_{\text{fluctuation term}} \\
&\quad+\underbrace{\gamma_tg(\vartheta_t,x_t^{i,N},\hat{x}_t^{j,M},\hat{\eta}_t^{j,M})\sigma^{-\top} \mathrm{d}w_t^{i,N}}_{\text{noise term}}. \label{eq:algorithm-2-rewrite}
\end{align}
From these expressions, it is clear that the algorithms in \eqref{eq:algorithm-1} - \eqref{eq:algorithm-2} are, indeed, continuous-time stochastic gradient descent algorithms with respect to $\mathcal{J}$. These two estimators, while seemingly only slightly different from those introduced in \citet{sharrock2026efficient}, only require observation of a single particle trajectory. In particular, unlike the analogous estimator in \citet{sharrock2026efficient}, the \emph{averaged} virtual particle estimator $(\theta_t)_{t\geq 0}$ defined in Section~\ref{sec:algorithm-1} does not depend on the empirical law of the observed particle system. Instead, it relies on the empirical law of two independent virtual particle systems, both integrated with the online parameter estimate. In a similar way, the particlewise virtual particle estimator $(\vartheta_t)_{t\geq 0}$ defined in Section~\ref{sec:algorithm-2} does not depend on auxiliary particles from the observed particle system, but instead on auxiliary particles from two virtual particle systems. 

Comparing the two estimators, it is clear that one can view the estimator $(\theta_t)_{t\geq 0}$ in Section~\ref{sec:algorithm-1} as an \emph{averaged} version of the estimator $(\vartheta_t)_{t\geq 0}$ in Section~\ref{sec:algorithm-2}. Indeed, wherever a single virtual particle appears in \eqref{eq:algorithm-2}, an average over all of the particles appears in \eqref{eq:algorithm-1}. More precisely, the averaged update is the conditional expectation of a uniformly resampled particlewise update, i.e., a Rao--Blackwellisation of the particlewise estimator. Thus, if we simulate $M\gg 1$ virtual particles, the estimator defined in Section~\ref{sec:algorithm-2} is (slightly) less costly to implement than the one defined in Section~\ref{sec:algorithm-1}. On the other hand, it also uses less of the available information. In cases where the interaction is defined via low-dimensional empirical moments, and thus the propagation cost is $\mathcal{O}(M)$, the averaged update may therefore be preferable. Conversely, for generic pairwise interactions, where the cost is $\mathcal{O}(M^2)$, the particlewise update may instead be preferable. 

We will later establish convergence results for these estimators, as both the time horizon $t\rightarrow\infty$, the number of particles in the true data-generating process $N\rightarrow\infty$, and the number of virtual particles $M\rightarrow\infty$. This suggests that, in principle, it will be necessary to propagate a large number of virtual particles, potentially at a very large computational cost. In practice, however, we find that the performance of these estimators is very robust to the number of virtual particles. Indeed, in many cases, it is sufficient to use as few as two. Regardless, even if a large number of virtual particles are required, these estimators still offer an advantage in cases where the cost of measuring additional particle trajectories dominates the simulation cost.

In this paper, we are particularly interested in the regime where only a single particle trajectory is observed. In principle, however, the same virtual-particle construction could be immediately extended to other observation regimes. In this case, one would simply average the full update increment over the available observed trajectories. The theoretical analysis in that case is unchanged, except for a reduction in the variance of the resulting estimator; see also Corollary~33 in \citet{sharrock2026efficient}.


\section{Theoretical Results}
\label{sec:theory}

In this section, we present our main results regarding the convergence of the estimators introduced in Sections~\ref{sec:algorithm-1} and \ref{sec:algorithm-2}.

\subsection{Notation}
We will first require some additional notation. Fix $N,M\in\mathbb{N}$. Let $K=dN + dM + dpM + dM $. Let $\boldsymbol{x}^N = (x^{1,N}, \dots, x^{N,N})^{\top}$, $\hat{\boldsymbol{x}}^M = (\hat{x}^{1,M}, \dots, \hat{x}^{M,M})^{\top}$, $\hat{\boldsymbol{y}}^M = (\hat{y}^{1,M}, \dots, \hat{y}^{M,M})^{\top}$, and $\tilde{\boldsymbol{x}}^M = (\tilde{x}^{1,M},\dots,\tilde{x}^{M,M})^{\top}$. We can then define the concatenated process 
    \begin{align}
    \boldsymbol{z}^{N,M} &= \mathcal{C}(\boldsymbol{x}^N, \hat{\boldsymbol{x}}^M, \hat{\boldsymbol{y}}^M, \tilde{\boldsymbol{x}}^M) \in \mathbb{R}^{K},
    \end{align} 
where $\mathcal{C}:(\mathbb{R}^{d})^N \times (\mathbb{R}^{d})^M \times (\mathbb{R}^{p\times d})^M  \times (\mathbb{R}^d)^M \rightarrow\mathbb{R}^{K}$ is the concatenation operator. We also write $\mu^N:=\frac1N\sum_{a=1}^N\delta_{x^{a,N}}$, $\hat\mu^M:=\frac1M\sum_{a=1}^M\delta_{\hat x^{a,M}}$, and $\tilde\mu^M:=\frac1M\sum_{a=1}^M\delta_{\tilde x^{a,M}}$. For $i\in[N]$ and $j,k\in[M]$, we can then define
\begin{align*}
B^{i,N,M}(\theta,\boldsymbol z^{N,M})
&:=
\textstyle \frac1M\sum_{a=1}^M b(\theta,x^{i,N},\tilde{x}^{a,M})
=
B(\theta,x^{i,N},\tilde{\mu}^{M}),\\
G^{i,N,M}(\theta,\boldsymbol  z^{N,M})
&:=
\textstyle \frac1M\sum_{a=1}^M
\big[
\partial_\theta b(\theta,x^{i,N},\hat{x}^{a,M})
+
\hat{y}^{a,M}\nabla_y b(\theta,x^{i,N},\hat{x}^{a,M})
\big],\\
H^{i,N,M}(\theta,\boldsymbol z^{N,M})
&:=
G^{i,N,M}(\theta,\boldsymbol  z^{N,M})(\sigma\sigma^{\top})^{-1}
\big[
B^{i,N,M}(\theta,\boldsymbol  z^{N,M})-B(\theta_0,x^{i,N},\mu^N)
\big],
\intertext{and, similarly,}
b^{i,k,N,M}(\theta, \boldsymbol z^{N,M})
&:= b(\theta,x^{i,N},\tilde{x}^{k,M}),\\
g^{i,j,N,M}(\theta,\boldsymbol z^{N,M})
&:=
\partial_\theta b(\theta,x^{i,N},\hat{x}^{j,M})
+
\hat{y}^{j,M}\nabla_y b(\theta,x^{i,N},\hat{x}^{j,M}),\\
h^{i,j,k,N,M}(\theta, \boldsymbol z^{N,M})
&:=
g^{i,j,N,M}(\theta,\boldsymbol  z^{N,M})(\sigma\sigma^{\top})^{-1}
\big[
b^{i,k,N,M}(\theta,\boldsymbol  z^{N,M})-B(\theta_0,x^{i,N},\mu^N)
\big].
\end{align*}
For fixed $\theta\in\Theta$, let $\boldsymbol{z}_t^{\theta,N,M}=\mathcal{C}(\boldsymbol{x}_t^{\theta_0,N},\hat{\boldsymbol{x}}_t^{\theta,M},\hat{\boldsymbol{y}}_t^{\theta,M},\tilde{\boldsymbol{x}}_t^{\theta,M})$ denote the concatenated process obtained by integrating the virtual systems at parameter value $\theta$. 
Then this process satisfies an SDE of the form
    \begin{equation}
    \mathrm{d}\boldsymbol{z}_t^{\theta,N,M}
    =
    \Phi^{N,M}(\theta,\boldsymbol{z}_t^{\theta,N,M})\,\mathrm{d}t
    +
    \Sigma_{N,M}\,\mathrm{d}\boldsymbol{b}_t^{N,M},
    \end{equation}
for a suitable drift $\Phi^{N,M}$, diffusion matrix $\Sigma_{N,M}$, and Brownian motion $\boldsymbol{b}_t^{N,M}$. Suppose we write $\Pi_\theta^{N,M}$ for the invariant law of this process. By independence of the observed and virtual blocks, this invariant law factorises as  
\begin{equation}
\Pi_\theta^{N,M}=\pi_{\theta_0}^N\otimes\Lambda_\theta^M\otimes\pi_\theta^M,
\end{equation}
where $\Lambda_\theta^M$ denotes the invariant law of the pair $(\hat{\boldsymbol{x}}_t^{\theta,M},\hat{\boldsymbol{y}}_t^{\theta,M})_{t\geq 0}$, and $\pi_\theta^M$ denotes the invariant law of $(\tilde{\boldsymbol{x}}^{\theta,M})_{t\geq 0}$

We will also require some additional terminology. First, we say that a function $F:\Theta\times\mathbb{R}^K\rightarrow\mathbb{R}$ has the averaged polynomial growth property if there exist $q\ge1$ and $C<\infty$ such that, for all $\theta\in\Theta$, and for all $i\in[N]$,
\begin{align}
|F(\theta,\boldsymbol{z}^{N,M})|
\le
C\Big(
1+\|x^{i,N}\|^{q}
+\frac1M\sum_{a=1}^M(\|\hat{x}^{a,M}\|^q+\|\tilde{x}^{a,M}\|^q+\|\hat{y}^{a,M}\|^q)
\Big)
\tag{PGP-a}
\label{eq:IPS-PGP-a}
\end{align}
Meanwhile, we say that $F$ has the particlewise polynomial-growth property if there exist $q\ge1$ and $C<\infty$ such that, for all $\theta\in\Theta$, and for all $i\in[N]$ and $j,k\in[M]$, 
\begin{align}
|F(\theta,\boldsymbol z^{N,M})|
\le
C\Big(
1+\|x^{i,N}\|^{q}
+\|\hat{x}^{j,M}\|^q+\|\tilde{x}^{k,M}\|^q+\|\hat{y}^{j,M}\|^q
\Big)
\tag{PGP-b}
\label{eq:IPS-PGP-b}
\end{align}
Fix some $\alpha\in(0,1]$. We write $\mathbb{G}^{K,N,M}$ for the space of all functions ${F}:\Theta\times\mathbb{R}^K\rightarrow\mathbb{R}$ such that, for each $\theta\in\Theta$, the map $\boldsymbol{z}\mapsto F(\theta,\boldsymbol{z})$ is in ${C}(\mathbb{R}^K)$; for each fixed $\boldsymbol{z}\in\mathbb{R}^K$, the map $\theta\mapsto F(\theta,\boldsymbol{z})$ is in $C^2(\Theta)$; and, again for each fixed $\boldsymbol{z}\in\mathbb{R}^K$, the maps $\theta \mapsto \partial_{\theta} F(\theta,\boldsymbol{z})$ and $\theta\mapsto \partial_{\theta}^2 F(\theta,\boldsymbol{z})$ are H\"older continuous with exponent $\alpha$.  We then define
\begin{equation}
\mathbb G_c^{K,N,M}
:=
\Bigl\{F\in\mathbb G^{K,N,M}: \int_{\mathbb R^K}F(\theta,z)\,\Pi_\theta^{N,M}(\mathrm dz)=0\ \text{for all }\theta\in\Theta\Bigr\}.
\end{equation}
Finally, we use the notation $\bar{\mathbb{G}}^{a,K,N,M}$ and $\bar{\mathbb{G}}^{b,K,N,M}$ to denote the subsets of $\mathbb{G}^{K,N,M}$ consisting of all $F\in\mathbb{G}^{K,N,M}$ such that $F$ and all of its first and second derivatives with respect to $\theta$ satisfy \eqref{eq:IPS-PGP-a} and \eqref{eq:IPS-PGP-b}, respectively. 

\subsection{Assumptions}
We are now ready to define our standing assumptions. We begin with a standard Robbins-Monro type assumption on the learning rate \citep[e.g.,][]{sirignano2017stochastic}.

\begin{assumption}
\label{assumption:learning-rate}
    The learning rate $\gamma_t:\mathbb{R}_{+}\rightarrow\mathbb{R}_{+}$ is a positive, non-increasing function which satisfies $\int_0^{\infty}\gamma_t\mathrm{d}t = \infty$, $\int_0^{\infty}\gamma_t^2\mathrm{d}t<\infty$, $\int_0^{\infty}|\dot{\gamma}_t|\mathrm{d}t<\infty$, and $\lim_{t\rightarrow\infty}\gamma_t t^{\rho}=0$ for some $\rho>0$.
\end{assumption}

We next introduce our assumptions on the concatenated process consisting of the observed IPS, the two virtual IPSs, and the virtual tangent IPS. These conditions are required in order to control the ergodic behaviour of the IPS. In particular, they ensure that fluctuation terms tend to zero sufficiently quickly as $t\rightarrow\infty$. Following the now well established approach in \citet{sirignano2017stochastic}, we control such terms by rewriting them in terms of the solutions of some related Poisson equations. This condition requires that these solutions are unique, and that they grow at most polynomially in a suitable sense. 

\begin{assumption}
\label{assumption:poisson}
The following conditions hold for each $N,M\in\mathbb{N}$.
    \begin{itemize}
        \item[(i)] For every $\theta\in\Theta$, the process $(\boldsymbol{z}_t^{\theta,N,M})_{t\geq 0}$ is ergodic with unique invariant law $\Pi_{\theta}^{N,M}\in\mathcal{P}(\mathbb{R}^{K})$.
        \item[(ii)] For every $q>0$, there exists $K_q<\infty$ such that, for all $i\in[N]$, $j,k\in[M]$, uniformly in $\theta\in\Theta$, 
        \begin{align}
            &\int_{\mathbb{R}^K} \Big(1+||x^{i,N}||^{q} +||\hat{x}^{j,M}||^q+||\tilde{x}^{k,M}||^q+||\hat{y}^{j,M}||^q\Big) \Pi_{\theta}^{N,M}(\mathrm{d}\boldsymbol{z}^{N,M})\leq K_q.
        \end{align}
        \item[(iii)] Let $\Sigma_{\theta,\ell}^{N,M}:=\partial_{\theta_\ell}\Pi_{\theta}^{N,M}$, $\ell=1,\dots p$. For every $q>0$, there exists $K_q<\infty$ such that, for each $\ell=1,\dots,p$, all $i\in[N]$, $j,k\in[M]$, and uniformly in $\theta\in\Theta$,
        \begin{align}
        \int_{\mathbb R^K}
        \Bigl(
        1+\|x^{i,N}\|^{q}+\|\hat{x}^{j,M}\|^q+\|\tilde{x}^{k,M}\|^q+\|\hat{y}^{j,M}\|^q
        \Bigr)
        |\Sigma_{\theta,\ell}^{N,M}|(\mathrm{d}\boldsymbol{z}^{N,M})
        \le K_q,
        \end{align}
        \item[(iv)] Let $\mathcal A_{\theta}^{N,M}$ denote the infinitesimal generator of $(\boldsymbol z_t^{\theta,N,M})_{t\ge0}$. For every $F\in\mathbb G_c^{K,N,M}$, the Poisson equation
        \begin{equation}
        \mathcal A_{\theta}^{N,M}v(\theta,\boldsymbol z^{N,M})=F(\theta,\boldsymbol z^{N,M})
        \end{equation}
        admits a unique solution $v\in\mathbb G^{K,N,M}$ such that $\boldsymbol{z}\mapsto v(\theta,\boldsymbol{z})$ belongs to $C^2(\mathbb R^K)$. Moreover, if $F\in\bar{\mathbb G}^{a,K,N,M}$ or $F\in\bar{\mathbb G}^{b,K,N,M}$, then $v$ belongs to the same class, and $\partial_{\boldsymbol{z}}\partial_\theta v$ has the corresponding polynomial growth.
        \item[(v)] For all $i\in[N]$, $j,k\in[M]$, and for all  $q>0$, $\mathbb{E}[||{{x}}^{i,N}_t||^q]<\infty$, $\mathbb{E}[||\hat{{x}}^{j,M}_t||^q]<\infty$, $\mathbb{E}[||\tilde{{x}}^{k,M}_t||^q]<\infty$, and $\mathbb{E}[||\hat{{y}}^{j,M}_t||^q]<\infty$. In addition, there exists $K_q>0$ such that, for all $i\in[N]$, $j,k\in[M]$, and for sufficiently large $t$, 
        \begin{align}
        \mathbb{E}\Big[\sup_{s\le t}\|x_s^{i,N}\|^q\Big]
        +
        \mathbb{E}\Big[\sup_{s\le t}\|\hat{x}_s^{j,M}\|^q\Big]
        +
        \mathbb{E}\Big[\sup_{s\le t}\|\tilde{x}_s^{k,M}\|^q\Big]
        +
        \mathbb{E}\Big[\sup_{s\le t}\|\hat{y}_s^{j,M}\|^q\Big]
        \le K_q\sqrt{t}
        \end{align}
        In addition, for all $j\in[M]$, and uniformly in $\theta\in\Theta$,
        \begin{align}
        \mathbb{E}\Big[\sup_{s\leq t}||x_s^{\theta,j,M}||^q\Big] +  \mathbb{E}\Big[\sup_{s\leq t}||y_s^{\theta,j,M}||^q\Big] \leq K_q\sqrt{t}.
        \end{align}
        \end{itemize}
\end{assumption}

\begin{assumption}
\label{assumption:pgp}
The following conditions hold for each $N,M\in\mathbb{N}$.
\begin{itemize}
    \item[(i.a)] For all $i\in[N]$, $H^{i,N,M}$ belongs to $\bar{\mathbb{G}}^{a,K,N,M}$, component-wise
    \item[(ii.a)]  For all $i\in[N]$, $B^{i,N,M}$ belongs to $\mathbb{G}^{K,N,M}$, and satisfies \eqref{eq:IPS-PGP-a}, component-wise.
    \item[(iii.a)]  For all $i\in[N]$, $G^{i,N,M}$ satisfies \eqref{eq:IPS-PGP-a}, component-wise.
\end{itemize}
and
\begin{itemize}
    \item[(i.b)] For all $i\in[N]$, $j,k\in[M]$, $h^{i,j,k,N,M}$ belongs to $\bar{\mathbb{G}}^{b,K,N,M}$, component-wise
    \item[(ii.b)] For all $i\in[N]$, $k\in[M]$, $b^{i,k,N,M}$ belongs to $\mathbb{G}^{K,N,M}$, and satisfies \eqref{eq:IPS-PGP-b}, component-wise.
    \item[(iii.b)] For all $i\in[N]$, $j\in[M]$, $g^{i,j,N,M}$ satisfies \eqref{eq:IPS-PGP-b}, component-wise.
\end{itemize}
\end{assumption}

Finally, in order to study limits as the number of real particles $N\rightarrow\infty$ and the number of virtual particles $M\rightarrow\infty$, we will require an additional uniform-in-time propagation-of-chaos assumption. In particular, the following assumption  requires that the stationary first marginal and its parameter derivative converge uniformly to their mean-field limits, with convergence holding in the polynomially weighted total variation norm.

\begin{assumption}
\label{ass:appendix-derivative-poc}
For all $r\geq 1$, there exists a deterministic sequence $\delta_M\downarrow 0$ such that
\begin{equation}
\sup_{\theta\in\Theta}
\|\pi_\theta^{1,M}-\pi_\theta\|_{\mathrm{TV},r}
\leq \delta_M,
\qquad
\sup_{\theta\in\Theta}
\|\nu_\theta^{1,M}-\nu_\theta\|_{\mathrm{TV},r}
\leq \delta_M.
\end{equation}
\end{assumption}

\subsection{Preliminary Results}
\label{sec:prelim-results}

We begin by introducing two finite particle surrogates for our original objective function (see Section~\ref{sec:methodology}). In particular, we will consider
\small
\begin{align}
\mathcal{J}^{i,N,M}(\theta)
&:= \int_{\mathbb{R}^K} J(\theta,x^{i,N},\hat{\mu}^M,\tilde{\mu}^M,\mu^N)\,\Pi_{\theta}^{N,M}(\mathrm{d}\boldsymbol{z}^{N,M}) =: \int_{\mathbb{R}^K} J^{i,N,M}(\theta,\boldsymbol{z}^{N,M})\,\Pi_{\theta}^{N,M}(\mathrm{d}\boldsymbol{z}^{N,M}), \\
\mathcal{J}^{i,j,k,N,M}(\theta)
&:= \int_{\mathbb{R}^K} j(\theta,x^{i,N},\hat{x}^{j,M},\tilde{x}^{k,M},\mu^N)\,\Pi_{\theta}^{N,M}(\mathrm{d}\boldsymbol{z}^{N,M}) =: \int_{\mathbb{R}^K} j^{i,j,k,N,M}(\theta,\boldsymbol{z}^{N,M})\,\Pi_{\theta}^{N,M}(\mathrm{d}\boldsymbol{z}^{N,M}).
\end{align}
\normalsize
These functions can be viewed as the objectives targeted by our two algorithms when the numbers of particles, both real and virtual, are fixed and finite. We can also identify the gradients of these surrogate finite-particle objectives.

\begin{proposition}
\label{prop:asymptotic-partial-log-lik-grad}
Suppose that Assumption \ref{assumption:moments} - \ref{assumption:drift} and Assumption \ref{assumption:poisson}(i)-(iii) hold. Then, for every $i\in[N]$ and $j,k\in[M]$,
\begin{align}
\partial_{\theta}\mathcal{J}^{i,N,M}(\theta)
&=
\int_{\mathbb{R}^{K}}
H^{i,N,M}(\theta,\boldsymbol z^{N,M})\,
\Pi_{\theta}^{N,M}(\mathrm{d}\boldsymbol z^{N,M}),
\\
\partial_{\theta}\mathcal{J}^{i,j,k,N,M}(\theta)
&=
\int_{\mathbb{R}^{K}}
h^{i,j,k,N,M}(\theta,\boldsymbol z^{N,M})\,
\Pi_{\theta}^{N,M}(\mathrm{d}\boldsymbol z^{N,M}).
\end{align}
\end{proposition}

\begin{proof}
See Appendix~\ref{appendix:proofs-section43}.
\end{proof}

Interestingly, the two finite-particle surrogates (and therefore also their gradients) actually coincide. This provides some intuition as to why the two algorithms behave so similarly in practice (see Section~\ref{sec:numerics}). In particular, while the averaged and the particle-wise estimators use different path-wise stochastic approximations of the gradient, these stochastic approximations are ultimately both targeting the same gradient, even for fixed $N$ and $M$. This is the subject of the following propopsition.

\begin{proposition}
\label{prop:j-vp}
Suppose that Assumption~\ref{assumption:model} and Assumption~\ref{assumption:poisson}(i) hold. Then the two finite-particle surrogate objectives coincide, and are equal to
\begin{align}
    \mathcal{J}_{\mathrm{vp}}^{i,N,M}(\theta)&:=\int_{(\mathbb{R}^d)^N} J(\theta,x^{i,N},\pi_{\theta}^{1,M},\pi_{\theta}^{1,M},\mu^N)\pi_{\theta_0}^N(\mathrm{d}x^N),
\end{align}
where $\pi_{\theta}^{1,M}$ denotes the first marginal of $\pi_{\theta}^M$, the stationary distribution of the (virtual) interacting particle system evaluated at $\theta\in\Theta$. Suppose, in addition, the assumptions of Proposition~\ref{prop:asymptotic-partial-log-lik-grad} hold. Then the gradients of the two finite-particle surrogate objective coincide, and are equal to
\begin{align}
    \partial_{\theta}\mathcal{J}_{\mathrm{vp}}^{i,N,M}(\theta)&=\int_{(\mathbb{R}^d)^N} H(\theta,x^{i,N},\pi_{\theta}^{1,M},\nu_{\theta}^{1,M},\pi_{\theta}^{1,M},\mu^N)\pi_{\theta_0}^N(\mathrm{d}x^N).
\end{align}
where, similar to above, $\nu_\theta^{1,M}:=\partial_\theta \pi_\theta^{1,M}$ denotes the weak derivative of the first marginal of $\pi_{\theta}^M$.
\end{proposition}

\begin{proof}
See Appendix~\ref{appendix:proofs-section43}.
\end{proof} 

In general, the function $\mathcal{J}_{\mathrm{vp}}^{i,N,M}$ and its gradient $\partial_{\theta}\mathcal{J}_{\mathrm{vp}}^{i,N,M}$ can be viewed as finite-particle approximations to the original mean-field objective function (cf. Proposition~\ref{prop:ips-likelihood-n-t-limit-grad}) and its gradient (cf. Proposition~\ref{prop:ips-likelihood-n-t-limit-grad-explicit}). We formalise this idea in the following proposition, which shows that the finite-particle approximation of the gradient converges (uniformly) to the gradient of the original objective as $N,M\rightarrow\infty$.

\begin{proposition}
\label{prop:inf-n-convergence-1}
    Suppose that the assumptions of Proposition~\ref{prop:asymptotic-partial-log-lik-grad} hold. In addition, suppose that Assumptions~\ref{ass:appendix-meanfield-c1} and~\ref{ass:appendix-derivative-poc} hold. Then there exist a constant $C<\infty$, independent of $N$ and $M$, and deterministic sequences $\varepsilon_N\downarrow 0$ and $\delta_M\downarrow 0$ such that, for all $N,M\in\mathbb N$ and all $i\in[N]$,
\begin{align}
\sup_{\theta\in\Theta}
\big\|
\partial_{\theta}\mathcal{J}_{\mathrm{vp}}^{i,N,M}(\theta)
-
\partial_{\theta}\mathcal{J}(\theta)
\big\|
\leq
C\left(\varepsilon_N + \delta_M\right).
\label{eq:j-convergence}
\end{align}
\end{proposition}

\begin{proof}
See Appendix~\ref{appendix:proofs-section43}.
\end{proof}

The previous proposition establishes a combined error in both $N$ and $M$ for the finite-particle  approximate gradient. In fact, it is also useful to separate these two effects. To do so, we introduce the intermediate finite-particle objective
\begin{equation}
\mathcal{J}_N^i(\theta)
:=
\int_{(\mathbb{R}^d)^N}
J\!\left(\theta, x^{i,N}, \pi_\theta, \pi_\theta, \mu^N\right)
\,\pi_{\theta_0}^N(dx^N).
\end{equation}
This is the objective obtained by keeping the number of particles in the observed system fixed, while replacing only the virtual particle system by the mean-field stationary law. A useful interpretation of $\smash{\mathcal{J}_N^i}$ is in terms of a one-particle Markovian projection of the finite-$N$ drift. Let $\smash{\pi_{\theta_0}^{i,N}}$ denote the $i$-th marginal of $\pi_{\theta_0}^N$, and let $\smash{\pi_{\theta_0}^{-i,N}(\,\cdot\,\mid x)}$ denote a regular conditional law of the remaining particles, given $x^{i,N} = x$, under the stationary law $\pi_{\theta_0}^N$. We can then define
\begin{equation}
B_{\theta_0}^{i,N}(x)
:=
\int_{(\mathbb{R}^d)^{N-1}}
B\big(
\theta_0,
x,
\tfrac{1}{N}\delta_x + \tfrac{1}{N}\textstyle\sum_{j\neq i}\delta_{y^j}
\big)
\,\pi_{\theta_0}^{-i,N}(dy^{-i}\mid x).
\end{equation}
Thus, $B_{\theta_0}^{i,N}$ is the effective one-particle drift obtained by conditioning the finite-$N$ drift on the single observed particle. We can then decompose the finite-particle pseudo-objective as follows.

\begin{proposition}
\label{prop:add:markovian-projection}
For every $\theta \in \Theta$, there exists a constant $C_N^i$ independent of $\theta$ such that the finite-particle objective $\mathcal{J}_N^{i}$ can be written as
\begin{equation}
\mathcal{J}_N^i(\theta)
=
C_N^i
+
\frac{1}{2}
\int_{\mathbb{R}^d}
\left\|
B(\theta,x,\pi_\theta) - B_{\theta_0}^{i,N}(x)
\right\|_{\sigma\sigma^\top}^2
\,\pi_{\theta_0}^{i,N}(dx).
\end{equation}
\end{proposition}

\begin{proof}
See Appendix~\ref{appendix:proofs-section43}.
\end{proof}

This result shows that the minimiser of the finite-particle objective does not fit the full finite-particle drift directly. Instead, it fits the corresponding one-particle Markovian projection in the stationary weighted $\mathsf L^2$-geometry induced by $\sigma\sigma^\top$. This is the natural target in the regime where only a single particle trajectory is observed. Thus, $\mathcal{J}_N^i$ provides the natural intermediate objective for separating the error due to the finite observed system from the error due to the virtual-particle approximation. In particular, we have the following result.

\begin{proposition}
\label{prop:add:separate-N-M}
Suppose that the assumptions of Proposition~\ref{prop:inf-n-convergence-1} hold. Then $\mathcal{J}_N^i$ is differentiable and
\begin{equation}
\partial_\theta \mathcal{J}_N^i(\theta)
=
\int_{(\mathbb{R}^d)^N}
H(\theta,x^{i,N},\pi_\theta,\nu_\theta,\pi_\theta,\mu^N)\,\pi_{\theta_0}^N(\mathrm d x^N).
\end{equation}
Moreover, there exists a constant $C<\infty$, independent of $N$ and $M$, such that for all $N,M\in\mathbb{N}$ and all $i\in[N]$,
\begin{equation}
\sup_{\theta\in\Theta}\big\|
\partial_\theta \mathcal{J}_{\mathrm{vp}}^{i,N,M}(\theta)-\partial_\theta \mathcal{J}_N^i(\theta)
\big\|
\le C\,\delta_M,
\qquad
\sup_{\theta\in\Theta}\big\|
\partial_\theta \mathcal{J}_N^i(\theta)-\partial_\theta \mathcal{J}(\theta)
\big\|
\le C\,\varepsilon_N.
\end{equation}
\end{proposition}

\begin{proof}
See Appendix~\ref{appendix:proofs-section43}.
\end{proof}

The bias decomposition is particularly informative at the true parameter $\theta_0$.
In particular, using the fact that $\partial_\theta \mathcal J(\theta_0)=0$, we have
\begin{equation}
\partial_\theta \mathcal J_{\mathrm{vp}}^{i,N,M}(\theta_0)
=
\bigl(\partial_\theta \mathcal J_{\mathrm{vp}}^{i,N,M}(\theta_0)-\partial_\theta \mathcal J_N^i(\theta_0)\bigr)
+\partial_\theta \mathcal J_N^i(\theta_0).
\end{equation}
Thus, the bias at the true parameter decomposes into (i) a virtual particle approximation error of order $\delta_M$; and (ii) a real particle bias of order $\varepsilon_N$. The former vanishes in the limit as $M\rightarrow\infty$, suggesting that the bias which remains is a finite-$N$ effect, rather than a limitation of the virtual approximation.

\subsection{Main Results}
\label{sec:main-results}
We are now ready to state our main results. We begin by characterising the asymptotic behaviour of the estimators in the limit as the time horizon $t\rightarrow\infty$, given fixed and finite numbers of real and virtual particles. For clarity, we will now write $(\theta_t^{N,M})_{t\ge0}$ and $(\vartheta_t^{N,M})_{t\ge0}$ for the estimators corresponding to the given values of $N$ and $M$.

\begin{proposition}
\label{prop:finite-n-m-convergence}
Suppose that Assumptions~\ref{assumption:learning-rate} - \ref{assumption:pgp} hold. Let $N,M\in\mathbb{N}$, $i\in[N]$, and $j,k\in[M]$. Suppose that $\mathbb{P}(\theta_t\in\Theta~\forall t\geq 0) = \mathbb{P}(\vartheta_t\in\Theta~\forall t\geq 0)=1$. Then, almost surely, it holds that 
\begin{align}
    \lim_{t\rightarrow\infty} \left\| \partial_{\theta}\mathcal{J}_{\mathrm{vp}}^{i,N,M}(\theta_t^{N,M}) \right\| = 0, \qquad \lim_{t\rightarrow\infty} \left\| \partial_{\theta}\mathcal{J}_{\mathrm{vp}}^{i,N,M}(\vartheta_t^{N,M}) \right\| = 0.
\end{align}
\end{proposition}

\begin{proof}
See Appendix~\ref{appendix:proofs-section44}.
\end{proof}

We next consider joint asymptotics as both the time horizon and the number of particles go to infinity. In particular, our second main result establishes that $(\theta_t)_{t\geq 0}$ and $(\vartheta_t)_{t\geq 0}$ converge to the stationary points of the mean-field objective $\mathcal{J}$ in the iterated limit as first $t\rightarrow\infty$, and then $N,M\rightarrow\infty$.

\begin{proposition}
\label{prop:inf-n-m-convergence}
Suppose that the assumptions of Proposition~\ref{prop:inf-n-convergence-1} and Proposition~\ref{prop:finite-n-m-convergence} hold. Then, almost surely,
\begin{equation}
\lim_{N,M\to\infty}\limsup_{t\to\infty}\bigl\|\partial_{\theta}\mathcal{J}(\theta_t^{N,M})\bigr\| = 0,
\qquad
\lim_{N,M\to\infty}\limsup_{t\to\infty}\bigl\|\partial_{\theta}\mathcal{J}(\vartheta_t^{N,M})\bigr\| = 0.
\end{equation}
\end{proposition}

\begin{proof}
See Appendix~\ref{appendix:proofs-section44}.
\end{proof}

\begin{corollary}
\label{cor:optional-iterated-parameter-consistency}
Suppose that the assumptions of Proposition~\ref{prop:inf-n-m-convergence} hold. Suppose in addition that, for every $\varepsilon>0$, $\smash{c_\varepsilon
:=
\inf\{
\|\partial_\theta \mathcal J(\theta)\|
:
\theta\in\Theta,\ \|\theta-\theta_0\|\ge \varepsilon
\}
>0
}
$. 
Then, almost surely,
\begin{equation}
\lim_{N,M\to\infty}\limsup_{t\to\infty}\|\theta_t^{N,M}-\theta_0\|=0,
\qquad
\lim_{N,M\to\infty}\limsup_{t\to\infty}\|\vartheta_t^{N,M}-\theta_0\|=0.
\end{equation}
\end{corollary}

\begin{proof}
    See Appendix~\ref{appendix:proofs-section44}.
\end{proof}

\subsection{Discussion} 
\label{sec:main-discussion}

We conclude this section with a brief discussion of the theoretical results just established, as well as the assumptions we have imposed in order to obtain these results.

\subsubsection{The Main Results} Propositions~\ref{prop:j-vp} and~\ref{prop:finite-n-m-convergence} relate to the behaviour of the two algorithms for fixed and finite values of $N$ and $M$. In particular, these results demonstrate that both estimators target the same surrogate objective $\mathcal{J}_{\mathrm{vp}}^{i,N,M}$, even though they utilise different pathwise stochastic gradient estimates. In practice, the difference between them is thus primarily a variance--cost trade-off. As discussed in Section~\ref{sec:method-discussion}, the averaged update uses more of the available virtual-particle information and thus in general can be expected to have a smaller conditional variance. On the other hand, the particlewise update is cheaper to evaluate and may be preferable when the virtual particle system is expensive to propagate.

Propositions~\ref{prop:inf-n-convergence-1}, \ref{prop:add:markovian-projection}, and~\ref{prop:add:separate-N-M} further clarify the respective roles of the real and virtual particle numbers. The virtual particle approximation error is of order $\delta_M$, while the discrepancy caused by observing only a single particle from a finite system is of order $\varepsilon_N$. In particular, the decomposition through the intermediate objective $\mathcal{J}_N^i$ shows that, even if the number of virtual particles is taken arbitrarily large, one should still expect a residual finite-$N$ bias. Proposition~\ref{prop:add:markovian-projection} identifies this bias precisely: in the single-particle observation regime, the relevant finite-$N$ target is not the full finite-particle drift itself, but rather its one-particle Markovian projection $B_{\theta_0}^{i,N}$. This is natural, since a single observed trajectory cannot recover the full configuration of the surrounding particle system. The remaining bias at fixed $N$ is therefore intrinsic to the observation regime, rather than an artefact of the virtual particle approximation.

Proposition~\ref{prop:inf-n-m-convergence} and Corollary~\ref{cor:optional-iterated-parameter-consistency} can be understood as iterated-limit results which separate optimisation error from approximation error. First, for fixed $N$ and $M$, the algorithms drive the gradient of the surrogate objective to zero. Second, as $N,M\to\infty$, the surrogate gradient converges uniformly to the gradient of the mean-field objective $\mathcal{J}$. Thus, under the additional separation condition in Corollary~\ref{cor:optional-iterated-parameter-consistency}, which is a natural identifiability assumption excluding spurious stationary points away from $\theta_0$, the parameter estimates converge to the true parameter in the iterated limit.

\subsubsection{The Main Assumptions} The most restrictive assumptions in our analysis are Assumptions~\ref{assumption:poisson}, \ref{assumption:pgp}, and~\ref{ass:appendix-derivative-poc}. These assumptions are very standard in the continuous-time stochastic approximation literature \citep[e.g.,][]{sirignano2017stochastic,surace2019online}, where Poisson equation methods are used to control fluctuation terms and replace time averages by ergodic averages. In the present setting, however, verifying such assumptions is somewhat delicate, as the relevant state variable is the concatenated process consisting of the observed IPS, two virtual IPSs, and the virtual tangent IPS. Thus, one must control not only the ergodicity and moments of the virtual IPS, but also the corresponding derivative process, whose dynamics involve the spatial derivatives of the interaction drift and may grow more quickly. We do not attempt here to give minimal or easily verifiable primitive conditions under which Assumptions~\ref{assumption:poisson} and~\ref{assumption:pgp} hold. Nonetheless, establishing such conditions (e.g., under dissipativity, convexity, or contractivity assumptions on the confinement and interaction potentials), remains an important open problem.

Another issue relates to propagation of chaos for the extended system. For the approximation result in Proposition~\ref{prop:inf-n-convergence-1}, we assumed directly that the stationary first marginal $\smash{\pi_\theta^{1,M}}$ and its parameter derivative $\smash{\nu_\theta^{1,M}}$ converge to their mean-field limits, uniformly over $\theta\in\Theta$. While this is precisely the required assumption, we do not derive it from primitive assumptions on the model. That is, we do not establish a full propagation-of-chaos result for the augmented process including the tangent variables. Doing so would require a quantitative analysis of the linearised McKean--Vlasov dynamics associated with the parameter derivative, and would likely lead to explicit rates for $\delta_M$. Such results would no doubt be of independent interest, but are beyond the scope of the current paper.

It is also worth noting that Propositions~\ref{prop:finite-n-m-convergence},~\ref{prop:inf-n-m-convergence}, and~\ref{cor:optional-iterated-parameter-consistency} all assume that the parameter processes remain in $\Theta$ for all time. This is a fairly standard assumption in continuous-time stochastic approximation, and in concrete implementations can be enforced by projected or truncated variants of the algorithms, or proved directly using model-specific Lyapunov arguments. We do not pursue a theoretical analysis of such modifications here.

\subsubsection{Related Methodology} 
Finally, it is instructive to compare our results with those obtained in the companion paper \citep{sharrock2026efficient}. The estimator studied there targets the asymptotic likelihood of the IPS itself, and requires observation of three real particles from the data-generating system. By contrast, the estimators introduced here target the mean-field objective, and require observation of only a single real particle, with the missing interaction terms replaced by virtual particles integrated at the current parameter estimate. The benefit of the present construction is therefore a substantially weaker observation requirement, which may be advantageous when additional measurements are expensive but simulation is comparatively cheap. The price paid is that the theory now depends on the behaviour of the virtual particle and tangent systems, and the natural finite-$N$ target becomes the projected objective $\mathcal{J}_N^i$ rather than the full finite-particle likelihood. In this sense, the two approaches are complementary: the estimator in \citet{sharrock2026efficient} is closer to the finite-$N$ interacting system when several trajectories can be observed, while the virtual-particle estimators developed here are tailored to the genuinely sparse regime in which only a single trajectory is available.


\section{Numerical Results}
\label{sec:numerics}
In this section, we present numerical experiments to illustrate the performance of the proposed estimators. In all cases, unless otherwise specified, we discretise the SDEs using a standard Euler--Maruyama scheme, with constant time-step $\Delta t = 0.1$.  We perform all experiments using a MacBook Pro 16'' (2021) laptop with Apple M1 Pro chip and 16GB of RAM.

\subsection{Quadratic Confinement, Quadratic Interaction}
\label{sec:numerics-quadratic}
We begin by considering a one-dimensional IPS with quadratic confinement potential and quadratic interaction potential, parametrised by $\theta=(\theta_1,\theta_2)^{\top}\in\mathbb{R}^2$, namely
\begin{equation}    
\label{eq:IPS_linear_model}
    \mathrm{d}x_t^{\theta,i,N} = \Big[-\theta_1 x_t^{\theta,i,N} - \frac{\theta_2}{N} \sum_{j=1}^N \left(x_t^{\theta,i,N} - x_t^{\theta,j,N}\right)\Big] \mathrm{d}t + \sigma\mathrm{d}w_t^{i,N},
\end{equation}
where $\sigma>0$ is a (known) diffusion coefficient, and $w^{i,N}=(w_t^{i,N})_{t\geq 0}$ are a set of independent standard Brownian motions.  In this model, $\theta_1$ is a \emph{confinement parameter}, which controls the rate at which each particle is driven towards zero, while $\theta_2$ is an \emph{interaction parameter}, which determines the strength of interaction between the particles. 

The two online parameter estimators are obtained by substituting the relevant model-specific quantities into \eqref{eq:algorithm-1} and \eqref{eq:algorithm-2}. This yields
\begin{align}
\mathrm{d}\begin{bmatrix} \theta_{t,1} \\ \theta_{t,2} \end{bmatrix} &= 
\gamma_t \begin{bmatrix} -x_t^{i,N} + \theta_{t,2}\bar{\hat{y}}_{t,1}^M \\[1mm] -(x_t^{i,N} - \bar{\hat{x}}_t^M ) + \theta_{t,2} \bar{\hat{y}}_{t,2}^M \end{bmatrix}(\sigma\sigma^\top)^{-1}\Bigg[\mathrm{d}x_t^{i,N} - \left(-\theta_{t,1} x_t^{i,N} - \theta_{t,2}(x_t^{i,N} - \bar{\tilde{x}}_t^M )\right)\mathrm{d}t \Bigg] \label{eq:algorithm-1-linear}
\intertext{or}
\mathrm{d}\begin{bmatrix} \vartheta_{t,1} \\ \vartheta_{t,2} \end{bmatrix} &= 
\gamma_t \begin{bmatrix} -x_t^{i,N} + \vartheta_{t,2} \hat{y}_{t,1}^{j,M} \\[1mm] -(x_t^{i,N} - \hat{x}_t^{j,M} ) + \vartheta_{t,2} {\hat{y}}_{t,2}^{j,M} \end{bmatrix} (\sigma\sigma^\top)^{-1}\Bigg[\mathrm{d}x_t^{i,N} - \left(-\vartheta_{t,1} x_t^{i,N} - \vartheta_{t,2}(x_t^{i,N} - \tilde{x}_t^{k,M} )\right)\mathrm{d}t \Bigg], \label{eq:algorithm-2-linear}
\end{align}
where $\bar{(\cdot)}$ denotes the empirical mean of a vector of particles, and all other terms are as  defined previously.
In this example, it is in fact possible to compute several relevant quantities in closed form (cf. Appendix~\ref{appendix:quadratic-model-proofs}). Let $\alpha:=\theta_1+\theta_2$ and $\alpha_0:=\theta_{0,1}+\theta_{0,2}$. The mean-field objective is then given by 
\begin{equation}
\mathcal J(\theta)=\frac{(\alpha-\alpha_0)^2}{4\alpha_0}.
\end{equation} 
Thus, this objective only identifies the sum $\alpha:=\theta_1+\theta_2$. In addition, for every finite $N$ and every $M\in\mathbb{N}$, we have that 
\begin{equation}
\mathcal J_{\mathrm{vp}}^{i,N,M}(\theta)=J_N^i(\theta) =
\frac{N\theta_{0,1}+\theta_{0,2}}{4N\alpha_0\theta_{0,1}}
\big(\alpha-\alpha_N^\star\big)^2
+
\frac{(N-1)\theta_{0,2}^2}{4N\big(N\theta_{0,1}+\theta_{0,2}\big)},
\label{eq:explicit-JNi-quadratic}
\end{equation}
where the minimising value of $\alpha$ is given by
\begin{equation}
\alpha_N^\star
=
\frac{N\theta_{0,1}\alpha_0}{N\theta_{0,1}+\theta_{0,2}}
=
\alpha_0-\frac{\alpha_0\theta_{0,2}}{N\theta_{0,1}+\theta_{0,2}}.
\end{equation}
This implies that the finite-particle pseudo-target is exactly independent of the number of virtual particles, and the finite-$N$ bias is of order $N^{-1}$. Similarly, if only one of the parameters is to be estimated, with the other parameter assumed known, then the pseudo-true values are given respectively by
\begin{equation}
\theta_{1,N}^\star
=
\frac{N\theta_{0,1}^2-\theta_{0,2}^2}{N\theta_{0,1}+\theta_{0,2}} = 
\theta_{0,1}-
\frac{\alpha_0\theta_{0,2}}{N\theta_{0,1}+\theta_{0,2}}, \qquad \theta_{2,N}^\star
=
\frac{(N-1)\theta_{0,1}\theta_{0,2}}{N\theta_{0,1}+\theta_{0,2}}=\theta_{0,2}-\frac{\theta_{0,2}(\theta_{0,1}+\theta_{0,2})}{N\theta_{0,1}+\theta_{0,2}}.
\end{equation}

To begin, we assume that the true parameters are given by $\smash{\theta_0 = (1.2,0.5)^{\top}}$. Meanwhile, the initial parameter estimates are given by $\smash{\theta_{\mathrm{init},1}\sim \mathcal{U}(1.5,2.5)}$ and $\smash{\theta_{\mathrm{init},2}\sim\mathcal{U}(1.0,1.5)}$, respectively. We simulate trajectories from the IPS with $N=100$ particles and for $T=2000$ iterations, with initial condition $\smash{x_0^{i,N}\sim\mathcal{N}(0,1)}$. For both estimators, we use the learning rate $\gamma_t = \frac{1}{(1+t)^{0.55}}$. Finally, we use $M=20$ virtual particles.

\begin{figure}[b!]
  \centering
  \begin{subfigure}{0.325\linewidth}
    \includegraphics[width=\linewidth]{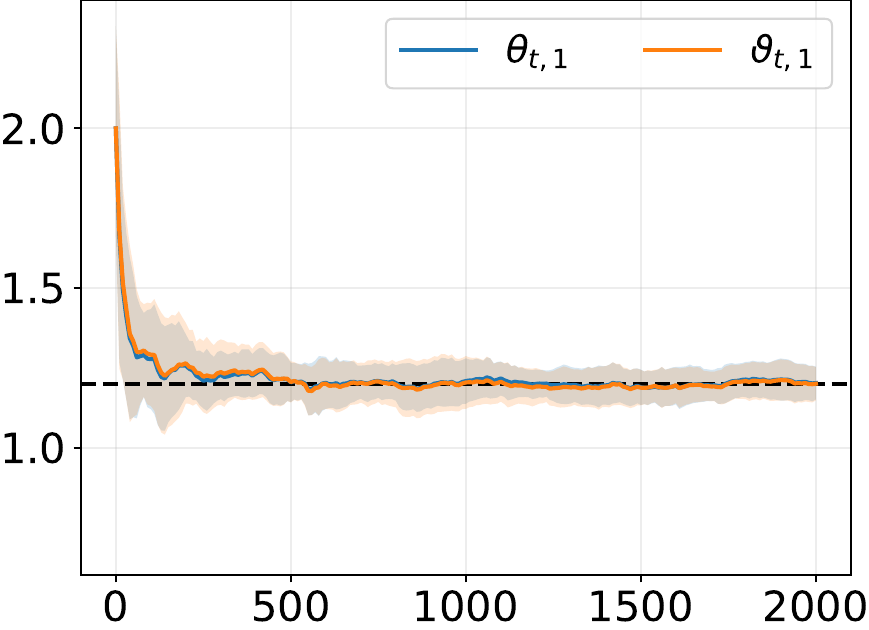}
    \caption{$\theta_1$.}
    \label{fig:1a}
  \end{subfigure}\hfill
  \begin{subfigure}{0.325\linewidth}
    \includegraphics[width=\linewidth]{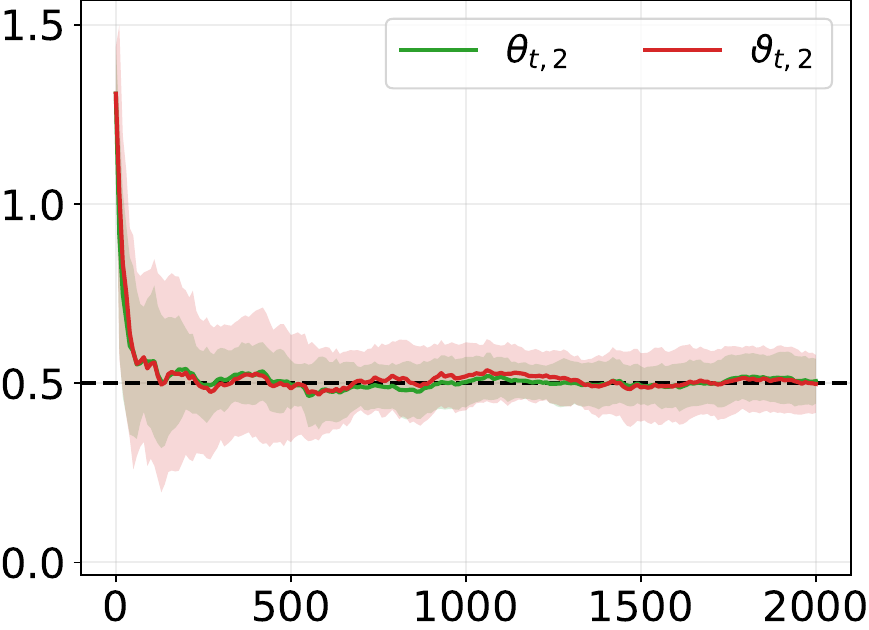}
    \caption{$\theta_2$.}
    \label{fig:1b}
  \end{subfigure}
  \begin{subfigure}{0.325\linewidth}
    \includegraphics[width=\linewidth]{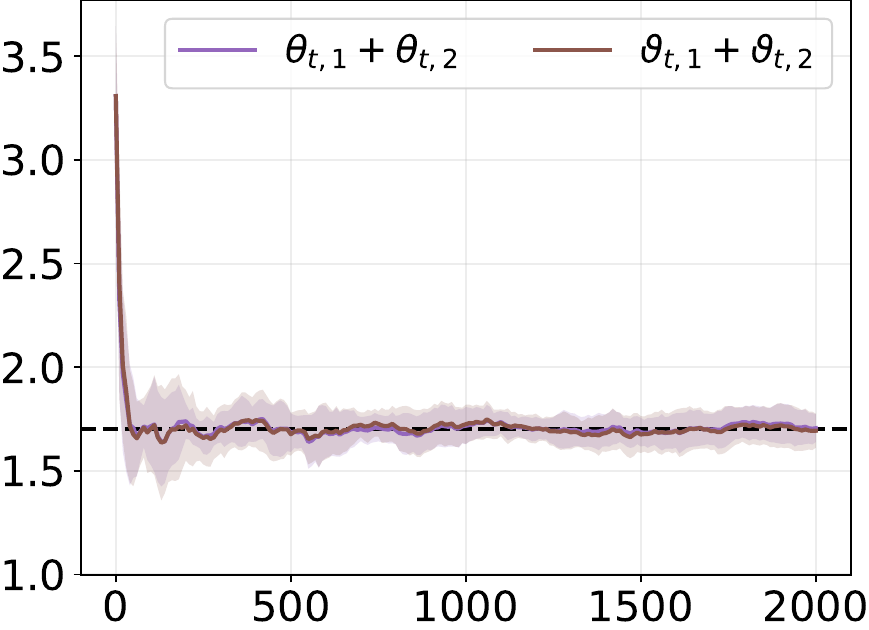}
    \caption{$\theta_1+\theta_2$.}
    \label{fig:1c}
  \end{subfigure}
  \captionsetup{width=\textwidth}
  \caption{\textbf{Online parameter estimation for a model with quadratic confinement potential and quadratic interaction potential}. We plot the sequence of online parameter estimates $\smash{(\theta_t)_{t\geq 0}}$ and $\smash{(\vartheta_t)_{t\geq 0}}$, as defined by the update equations in \eqref{eq:algorithm-1-linear} and \eqref{eq:algorithm-2-linear}. The true parameters are given by $\smash{\theta_0 = (1.2,0.5)^{\top}}$. The initial parameter estimates are given by $\smash{\theta_{\mathrm{init},1}\sim \mathcal{U}(1.5,2.5)}$ and $\smash{\theta_{\mathrm{init},2}\sim\mathcal{U}(1.0,1.5)}$.}
  \label{fig:1}
\end{figure}

The performance of the two estimators is illustrated in Figure~\ref{fig:1}. In this experiment, we report results for three cases: only $\theta_1$ is estimated (Fig.~\ref{fig:1a}); only $\theta_2$ is estimated (Fig.~\ref{fig:1b}); and both $\theta_1$ and $\theta_2$ are estimated (Fig.~\ref{fig:1c}). In the final case, since $(\theta_1,\theta_2)^{\top}$ are not jointly identifiable, we instead plot $\theta_1+\theta_2$ \citep[e.g.,][]{sharrock2023online}. In all three cases, both online parameter estimates converge to the true parameter(s). This being said, the relative performance of the two estimators does vary  somewhat based on the parameter to be estimated. In particular, in the case where only the confinement parameter is to be estimated (Fig.~\ref{fig:1a}), the evolution of both parameter estimates (blue, orange) is essentially identical. On the other hand, in the case where only the interaction parameter is to be estimated (Fig.~\ref{fig:1b}), the variance of the averaged estimator (green) is somewhat smaller than the variance of the second estimator (red). This difference is also present when both parameters are to be estimated (Fig.~\ref{fig:1c}), although here it is less evident.

\begin{figure}[t!]
  \centering
  \begin{subfigure}{0.245\linewidth}
    \includegraphics[width=\linewidth]{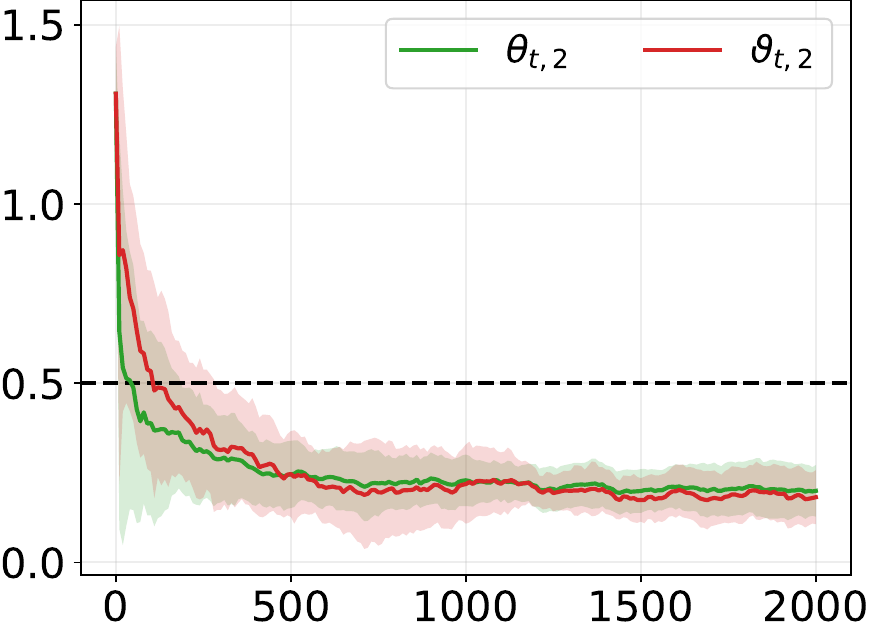}
    \caption{$N=2$.}
    \label{fig:2a}
  \end{subfigure}\hfill
  \begin{subfigure}{0.245\linewidth}
    \includegraphics[width=\linewidth]{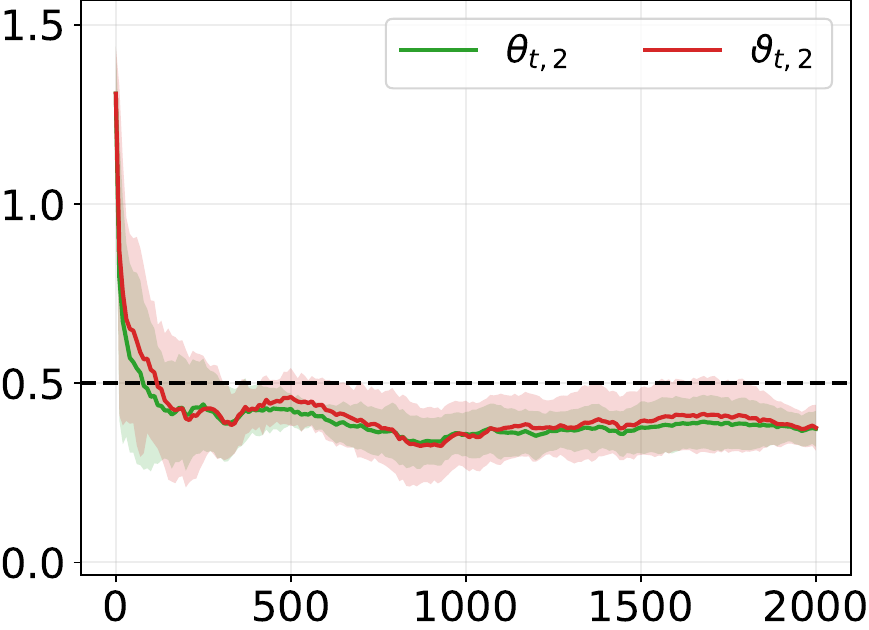}
    \caption{$N=5$.}
    \label{fig:2b}
  \end{subfigure}
  \begin{subfigure}{0.245\linewidth}
    \includegraphics[width=\linewidth]{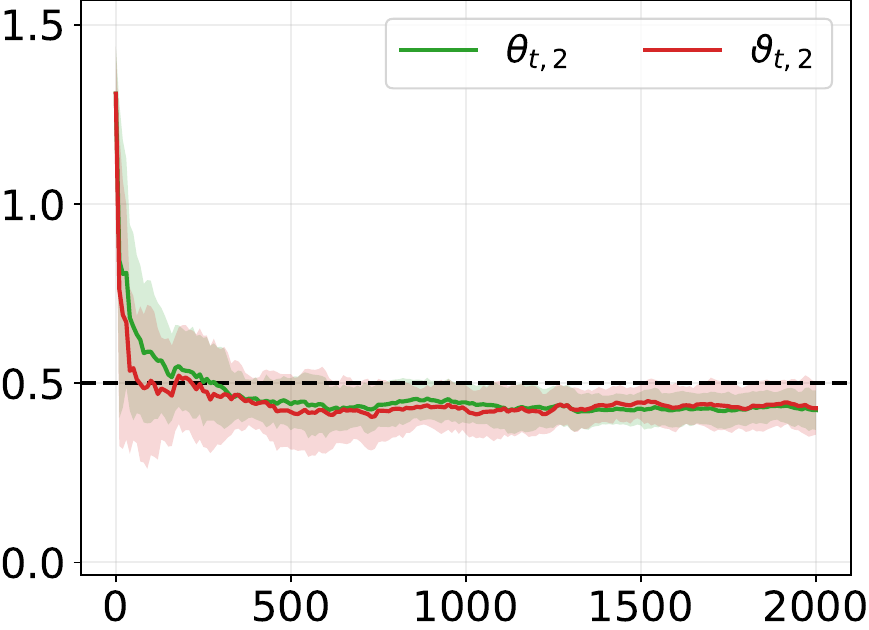}
    \caption{$N=10$.}
    \label{fig:2c}
  \end{subfigure}\hfill
  \begin{subfigure}{0.245\linewidth}
    \includegraphics[width=\linewidth]{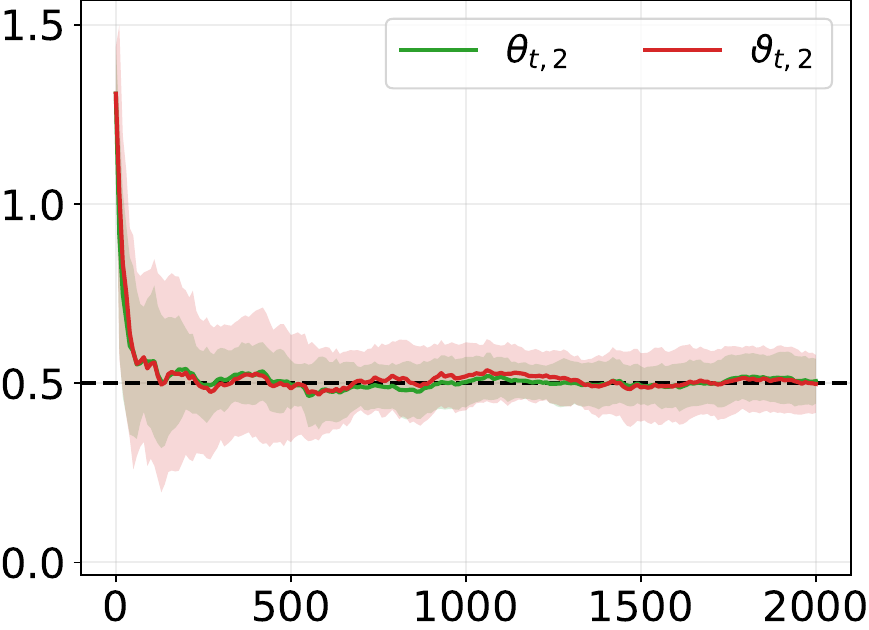}
    \caption{$N=100$.}
    \label{fig:2d}
  \end{subfigure}
  \captionsetup{width=\textwidth}
  \caption{\textbf{Online parameter estimation for the interaction parameter in a model with quadratic confinement potential and quadratic interaction potential}. We plot the sequence of online parameter estimates $\smash{(\theta_{t,2})_{t\geq 0}}$ and $\smash{(\vartheta_{t,2})_{t\geq 0}}$, as defined by the update equations in \eqref{eq:algorithm-1-linear} and \eqref{eq:algorithm-2-linear}, for $N\in\{2,5,10,100\}$.  The true parameters are given by $\smash{\theta_0 = (1.2,0.5)^{\top}}$, with the first parameter assumed known. The initial parameter estimates are given by $\smash{\theta_{\mathrm{init},2}\sim\mathcal{U}(1.0,1.5)}$.}
  \label{fig:2}
\end{figure}

In Figures~\ref{fig:2} and~\ref{fig:3}, we continue to investigate the performance of the two online parameter estimators, now as a function of the number of particles in the data-generating IPS. It is worth emphasising that this is \emph{not} the same as the number of \emph{observed} particles, which remains 1 for both estimators. Our results indicate a clear dependence on this value: as the number of particles increases, the final parameter estimates become progressively closer to the true parameter (Fig.~\ref{fig:2}), and the corresponding \(\mathsf L^2\) error decreases (Fig.~\ref{fig:3}). This behaviour is consistent with our theoretical results. In particular, our estimators are only expected to converge in the limit as $N\rightarrow\infty$ (cf. Proposition~\ref{prop:inf-n-m-convergence}). Meanwhile, for fixed and finite numbers of particles, the two estimators will instead converge to the stationary points (e.g., the minimiser) of the finite-particle surrogate objective (cf. Proposition~\ref{prop:finite-n-m-convergence}). For small numbers of particles, this can differ appreciably from the minimiser of the true target (i.e., the true parameter), leading to a significant asymptotic bias (Fig.~\ref{fig:2a}, Fig.~\ref{fig:2b}).

\begin{figure}[t!]
  \centering
  \begin{subfigure}{0.45\linewidth}
    \includegraphics[trim=0 10mm 0 0, clip, width=\linewidth]{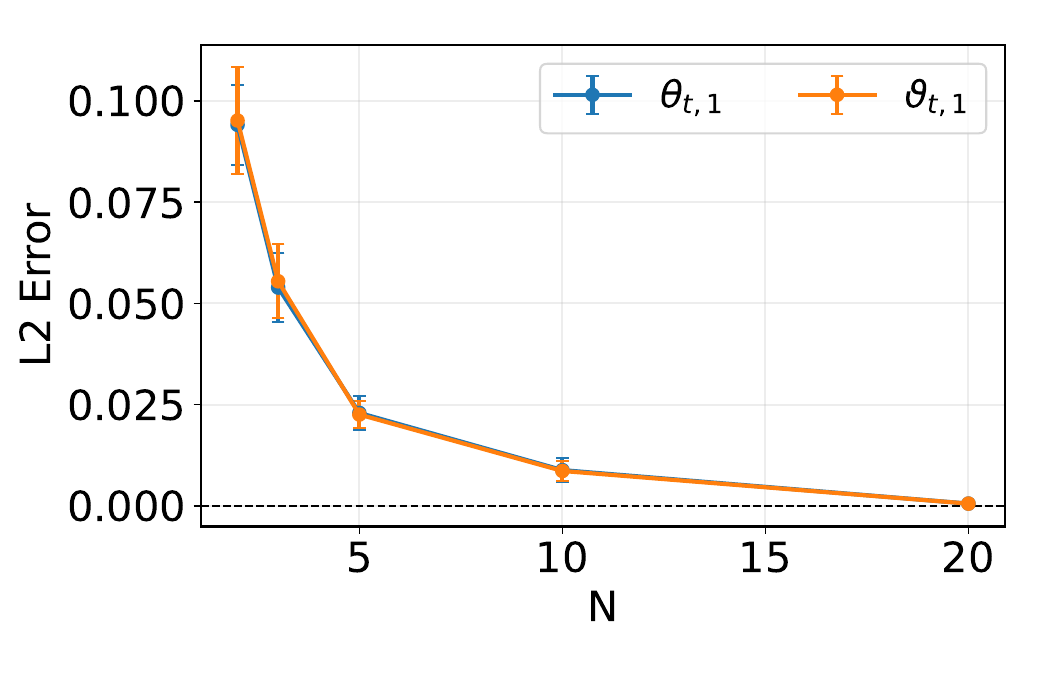}
    \caption{$\theta_1$.}
    \label{fig:3a}
  \end{subfigure}
  \begin{subfigure}{0.45\linewidth}
    \includegraphics[trim=0 10mm 0 0, clip, width=\linewidth]{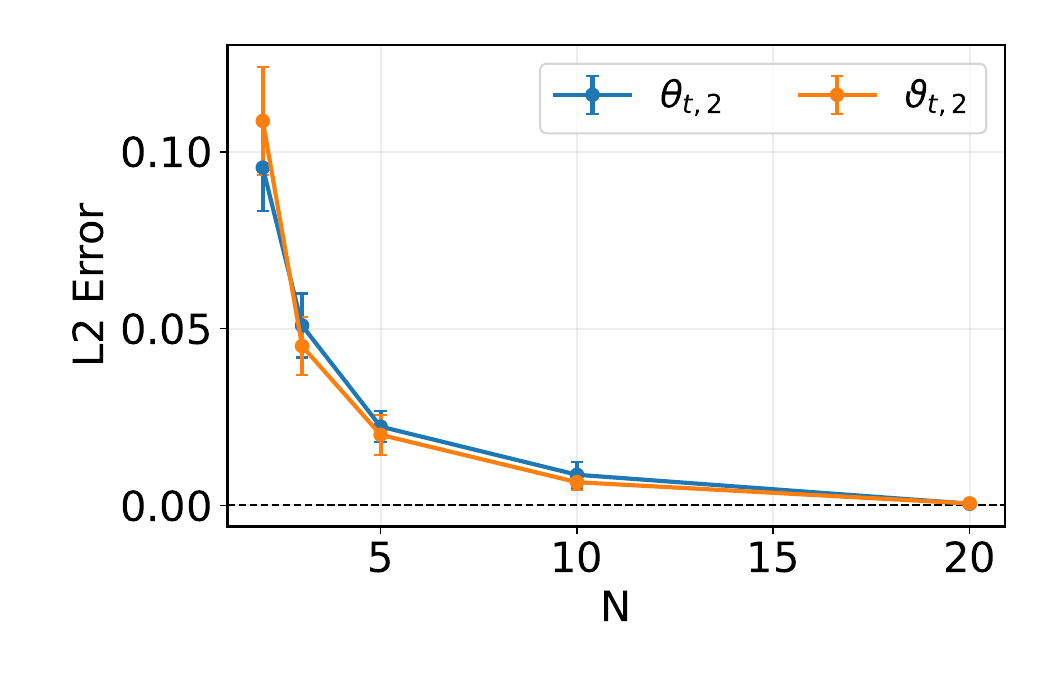}
    \caption{$\theta_2$.}
    \label{fig:3b}
  \end{subfigure}
  \captionsetup{width=\textwidth}
  \caption{\textbf{The $\mathrm{L}^2$ error of the averaged and the non-averaged estimators, for a model with quadratic confinement potential and quadratic interaction potential.} We plot the $\mathrm{L}^2$ error for both estimators after $T=50,000$ iterations, for $N\in\{3,5,10,25,50\}$ particles.}
  \label{fig:3}
\end{figure}

Finally, in Figure~\ref{fig:4}, we consider the performance of the estimators as a function of $M$, the number of virtual particles used in each algorithm. The results are once again consistent with our theoretical results (cf.\ Proposition~\ref{prop:add:separate-N-M} and Appendix~\ref{appendix:quadratic-model-proofs}). In particular, the asymptotic $\mathsf L^2$ error of both parameter estimators is essentially invariant to the value of $M$. Thus, even for small or modest values of $M$ both parameter estimates converge to the true parameter for sufficiently large values of $N$. Conversely, even when the value of $M$ is large, both parameter estimates exhibit a non-negligible asymptotic bias for small values of $N$.

\begin{figure}[b!]
  \centering
  \begin{subfigure}{0.45\linewidth}
    \includegraphics[trim=0 10mm 0 0, clip, width=\linewidth]{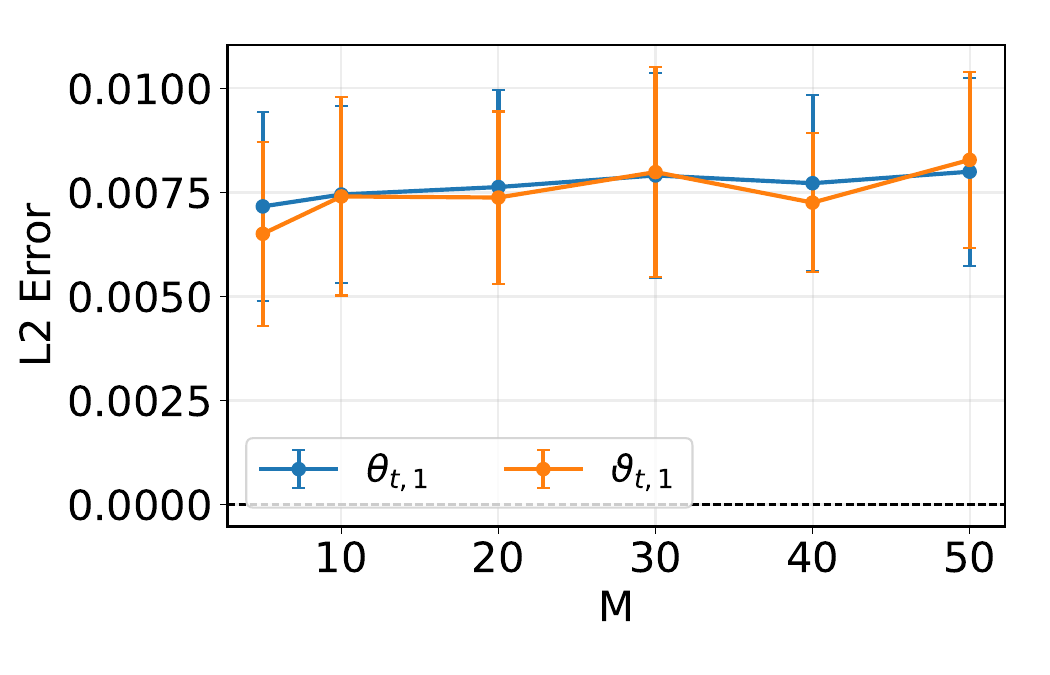}
    \caption{$\theta_1$.}
    \label{fig:4a}
  \end{subfigure}
  \begin{subfigure}{0.45\linewidth}
    \includegraphics[trim=0 10mm 0 0, clip, width=\linewidth]{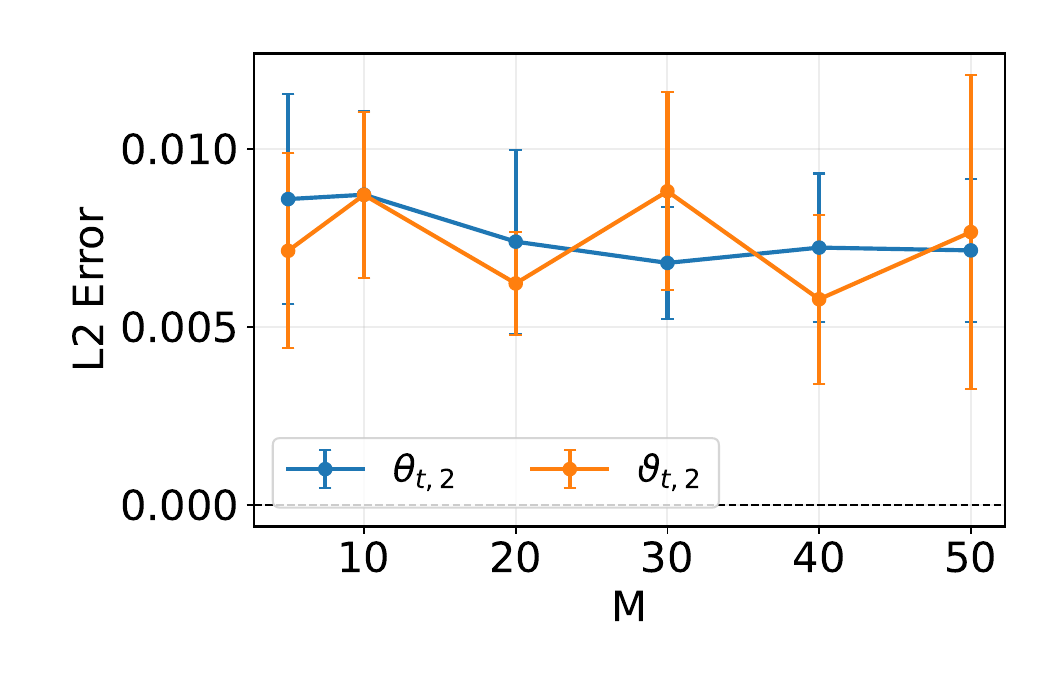}
    \caption{$\theta_2$.}
    \label{fig:4b}
  \end{subfigure}
  \captionsetup{width=\textwidth}
  \caption{\textbf{The $\mathrm{L}^2$ error of the averaged and the non-averaged estimators, for a model with quadratic confinement potential and quadratic interaction potential.} We plot the $\mathrm{L}^2$ error for both estimators after $T=50,000$ iterations, for $M\in\{5,10,20,30,40,50\}$ virtual particles.}
  \label{fig:4}
  \vspace{-5mm}
\end{figure}

\subsection{Stochastic FitzHugh--Nagumo Model}
\label{sec:fitzhugh}
We next consider a stochastic FitzHugh--Nagumo model \citep[e.g.,][]{baladron2012meanfield,luccon2021periodicity}, parametrised by $\theta = (\theta_1,\theta_2,\theta_3, \theta_4)^{\top}\in\mathbb{R}^4$, and defined by
\begin{align}
\mathrm{d}x_t^{\theta,i,N} &= \Big[ \theta_1 \Big(x_t^{\theta,i,N} - \frac{1}{3}(x_t^{\theta,i,N})^3 - y_t^{\theta,i,N}\Big) - \frac{\theta_2}{N}\sum_{j=1}^N (x_t^{\theta,i,N} - x_t^{\theta,j,N})\Big] \mathrm{d}t + \sigma\mathrm{d}w_t^{i,N} \label{eq:fitzhugh_nagumo1_IPS} \\
\mathrm{d}y_t^{\theta,i,N}& = \Big[ x_t^{\theta,i,N} + \theta_3 - \theta_4 y_{t}^{\theta,i,N} \Big] \mathrm{d}t. \label{eq:fitzhugh_nagumo2_IPS} 
\end{align}
This model originates in neuroscience, modelling the evolution of a collection of neurons of FitzHugh--Nagumo type, each being represented by its voltage $x_t^i$ and recovery variable $y_t^i$, and coupled through a linear mean-field interaction which corresponds to a coupling via electrical synapses. In this case, the model is degenerate, and thus we cannot use Girsanov's theorem to obtain a likelihood function. We thus use a minor modification of the original objective (and the resulting algorithms), in which the inner product is no longer weighted by the inverse of the diffusion coefficient; see also \citet{sharrock2026efficient}. 

We report illustrative results for our two estimators in the case that the first three parameters are to be (jointly) estimated, and the final parameter is known and fixed equal to the ground truth. We assume that the true parameter $\theta_0 = (\theta_{0,1},\theta_{0,2},\theta_{0,3},\theta_{0,4})^{\top}=(0.9,0.4,0.1,1.0)^{\top}$. Meanwhile, the initial parameter estimates are given by $\smash{\theta_{\mathrm{init},1}\sim \mathcal{U}(0.0,0.5)}$,  $\smash{\theta_{\mathrm{init},2}\sim \mathcal{U}(0.5,1.0)}$, and $\smash{\theta_{\mathrm{init},3}\sim \mathcal{U}(0.0,0.5)}$. We simulate trajectories from the IPS with $N\in\{3,5,10,20,50,100\}$ particles and for $T=5000$ iterations, with initial condition $x_0^{i,N}\sim \mathcal{N}(0,1)$ and $y_0^{i,N}\sim \mathcal{N}(0,1)$. For both estimators, we use the learning rate $\gamma_t = \frac{0.5}{(1+t)^{0.55}}$. Finally, we use $M=20$ virtual particles.

\begin{figure}[b!]
  \centering
  \begin{subfigure}{0.32\linewidth}
    \includegraphics[width=\linewidth]{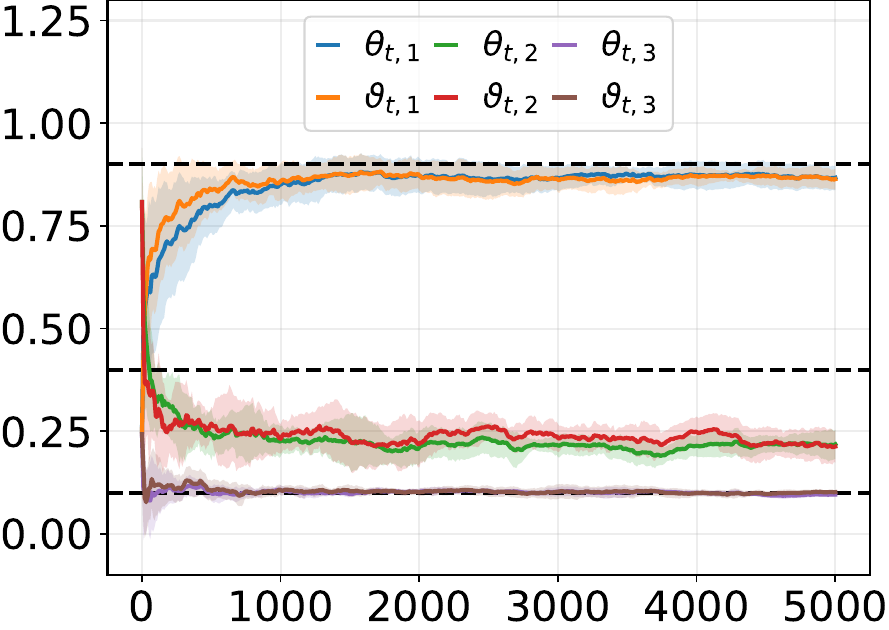}
    \caption{$N=3$.}
    \label{fig:5a}
  \end{subfigure}\hfill
  \begin{subfigure}{0.32\linewidth}
    \includegraphics[width=\linewidth]{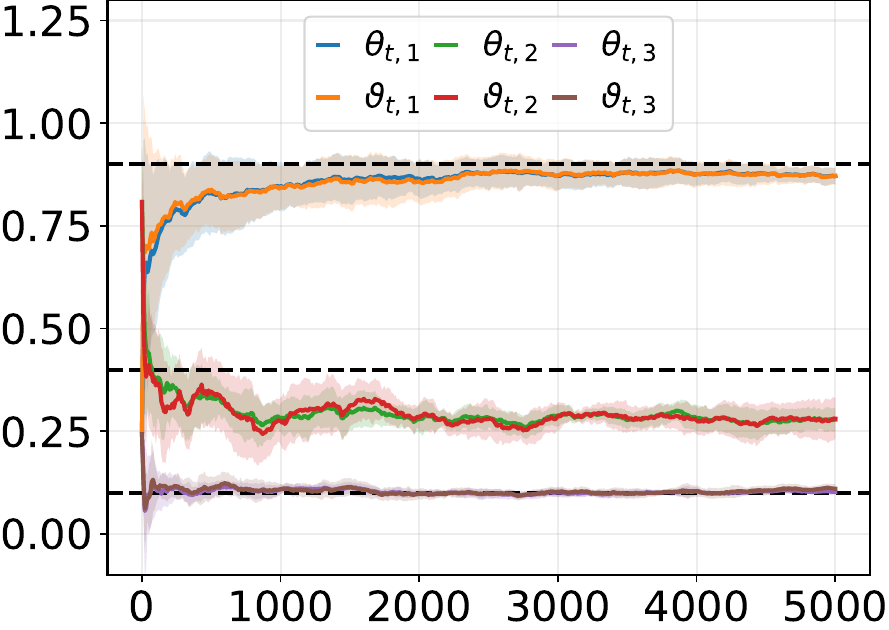}
    \caption{$N=5$.}
    \label{fig:5b}
  \end{subfigure}
  \begin{subfigure}{0.32\linewidth}
    \includegraphics[width=\linewidth]{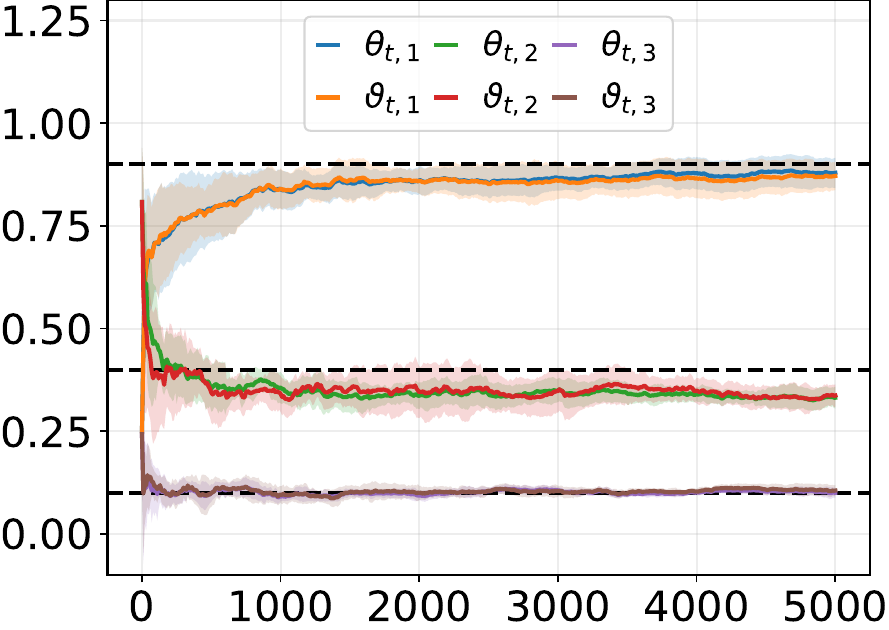}
    \caption{$N=10$.}
    \label{fig:5c}
  \end{subfigure} \\[4mm]
  \begin{subfigure}{0.325\linewidth}
    \includegraphics[width=\linewidth]{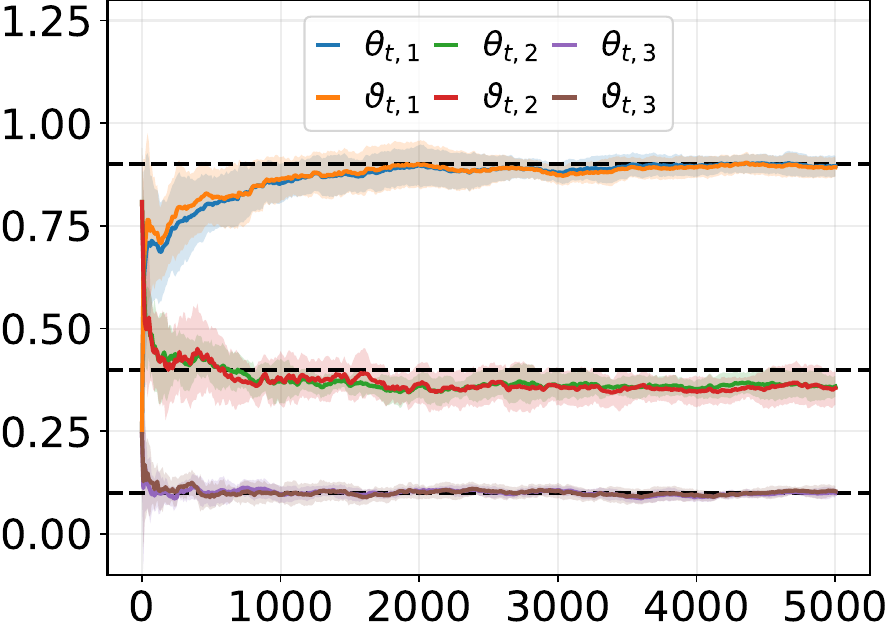}
    \caption{$N=20$.}
    \label{fig:5d}
  \end{subfigure}\hfill
  \begin{subfigure}{0.32\linewidth}
    \includegraphics[width=\linewidth]{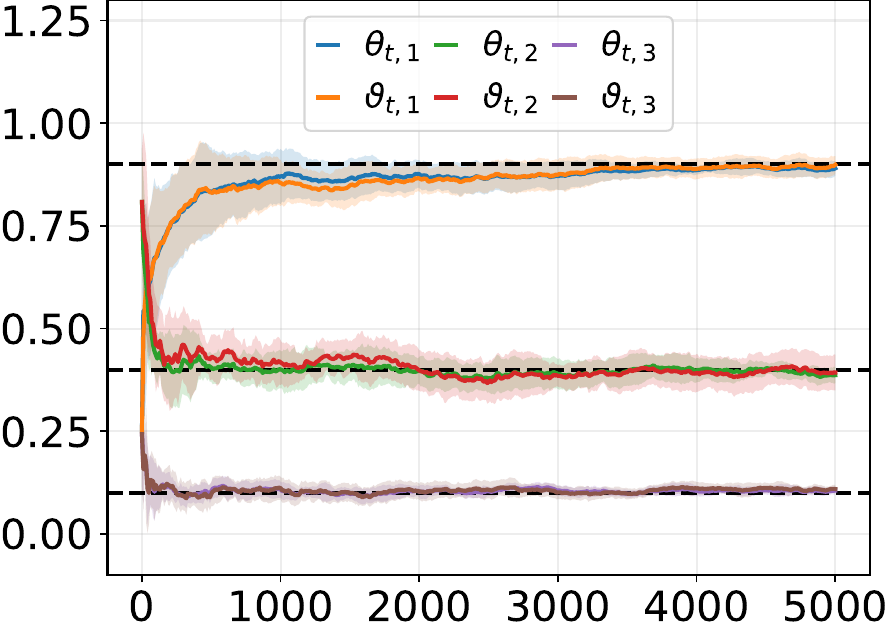}
    \caption{$N=50$.}
    \label{fig:5e}
  \end{subfigure}
  \begin{subfigure}{0.325\linewidth}
    \includegraphics[width=\linewidth]{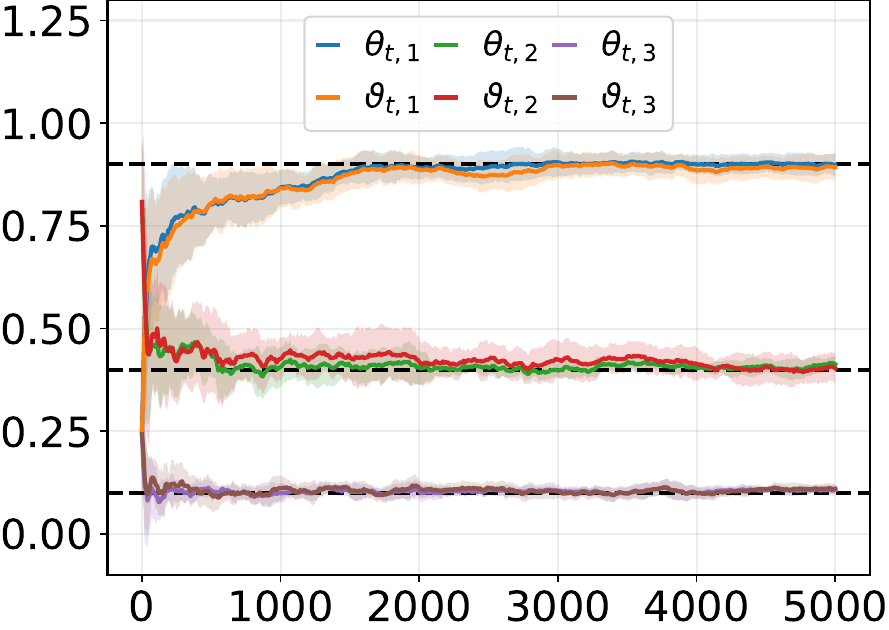}
    \caption{$N=100$.}
    \label{fig:5f}
  \end{subfigure}
  \captionsetup{width=\textwidth}
  \caption{\textbf{Online parameter estimation for the stochastic FitzHugh--Nagumo model}. We plot the sequence of online parameter estimates $\smash{(\theta_t)_{t\geq 0}}$ and $\smash{(\vartheta_t)_{t\geq 0}}$, for $N\in \{3,5,10,20,50,100\}$. The true parameters are given by $\theta_{0,1}=0.9$, $\theta_{0,2}=0.4$ and $\theta_{0,3}=0.1$, and $\theta_{0,4}=1.0$, with the final parameter assumed known. The initial parameter estimates are given by $\smash{\theta_{\mathrm{init},1}\sim \mathcal{U}(0.0,0.5)}$,  $\smash{\theta_{\mathrm{init},2}\sim \mathcal{U}(0.5,1.0)}$, and $\smash{\theta_{\mathrm{init},3}\sim \mathcal{U}(0.0,0.5)}$.}
  \label{fig:5}
\end{figure}

The results, shown in Figure~\ref{fig:5}, are largely consistent with the previous experiment. First, provided that the value of $N$ is sufficiently large, both of our estimators converge to the true parameter values (Fig.~\ref{fig:5e}, Fig.~\ref{fig:5f}). Second, for at least some of the parameters, the (asymptotic) bias of both estimators decreases (monotonically) as a function of $N$. Thus, in particular, the estimators exhibit a non-negligible asymptotic bias for both the confinement parameter $\theta_1$ (orange, blue) and the interaction parameter $\theta_2$ (green, red) when the value of $N$ is small (Fig.~\ref{fig:5a}, Fig.~\ref{fig:5b}, Fig.~\ref{fig:5c}). Interestingly, in this case, both estimators appear to converge to the correct value of $\theta_3$ regardless of the number of particles. Finally, for the interaction parameter $\theta_2$, the variance of the averaged estimator (green) is smaller than the variance of the non-averaged estimator (red), with little difference for the other two parameters.

\subsection{Stochastic Kuramoto Model}
Finally, we consider the stochastic Kuramoto model on the one-dimensional torus \(\mathbb T^{1}\) \citep[e.g.,][]{kuramoto1981rhythms,sakaguchi1988phase,acebron2005kuramoto,bertini2010dynamical}, viz
\begin{equation}
    \mathrm{d}x_t^{\theta,i,N} = -\frac{\theta}{N}\sum_{j=1}^N \sin(x_t^{\theta,i,N}-x_t^{\theta,j,N})\,\mathrm{d}t + \sigma\mathrm{d}w_t^{i,N}.
\end{equation}
where $\theta\in\mathbb{R}$ is the coupling strength. This system of interacting particles models the synchronisation of noisy oscillators interacting through their phases, and finds application in various fields including physics, chemistry, and biology; see, e.g., \cite{acebron2005kuramoto} and references therein. One interesting feature of this model is that its mean-field limit  exhibits a phase transition \citep[e.g.,][]{bertini2010dynamical}. In particular, when $\sigma>\sigma_c$, for some critical noise strength $\sigma_c$, the noise dominates and there is a unique invariant distribution (i.e., the uniform distribution). On the other hand, when $\sigma<\sigma_c$, there exists a family of non-trivial coherent equilibria, and the population tends to synchronise. Equivalently, given a fixed value of $\sigma>0$, there is a unique invariant distribution when $\theta<\theta_c$, and multiple invariant distributions when $\theta>\theta_c$, for some critical coupling strength $\theta_c:=\sigma^2$. 

We illustrate the performance of our estimators in Figure~\ref{fig:6}. In this case, we simulate trajectories from the IPS with $N\in\{3,10,50\}$ particles, and for $T=10000$ iterations. We use $M=20$ virtual particles, and a constant step size of $\gamma=0.02$. We also now consider a \emph{time-varying} specification of the true parameter:
\begin{align}
    \theta_{0,t} &= \begin{cases}
    \theta_{0,1}, & t\in[0,5000),\\
    \theta_{0,2}, & t\in[5000,10000],
    \end{cases}
\label{eq:true-param-kuramoto}
\end{align}
where $\theta_{0,1} = 1.5$ and $\theta_{0,2} = 0.2$. We also assume that $\sigma=1.0 \implies \theta_c = \sigma^2 = 1.0$. Thus, the true coupling strength is initially above its critical value (since $\theta_{0,1}>\theta_{c}$), and then subsequently below its critical value (since $\theta_{0,2}<\theta_{c}$). While, strictly speaking, this scenario is outside the scope of our theoretical results, it demonstrates another advantage of our online estimation procedure in comparison to a batch or offline approach. In particular, our estimators are able to accurately track changes in the true parameter in real time.

\begin{figure}[b!]
  \centering
  \begin{subfigure}{0.325\linewidth}
    \includegraphics[width=\linewidth]{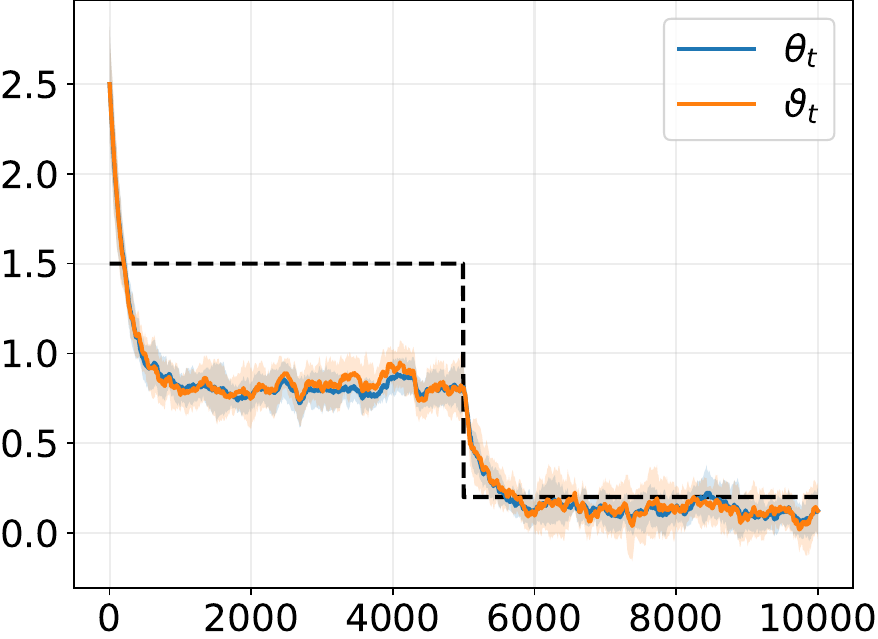}
    \caption{$N=3$.}
    \label{fig:6a}
  \end{subfigure}
  \hfill
  \begin{subfigure}{0.325\linewidth}
    \includegraphics[width=\linewidth]{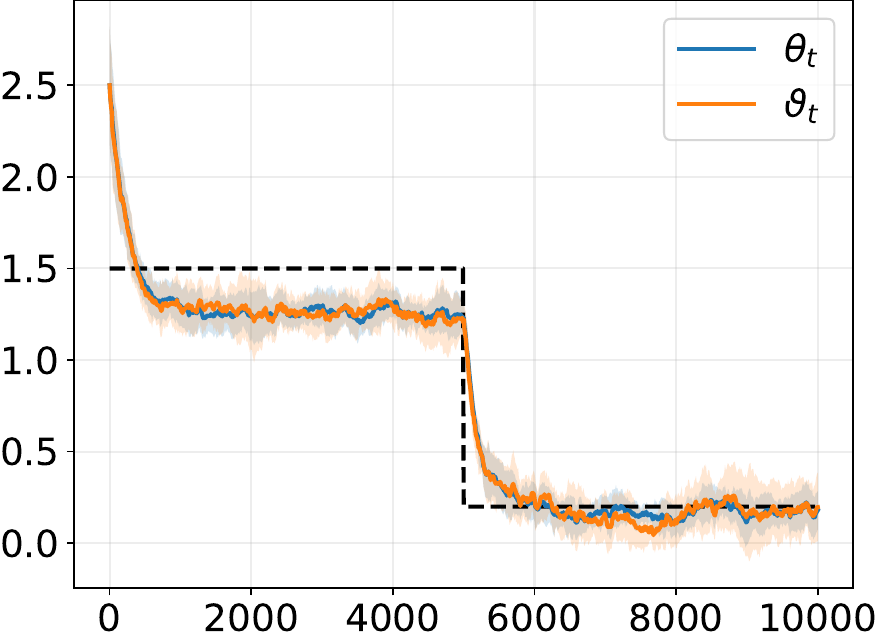}
    \caption{$N=10$.}
    \label{fig:6b}
  \end{subfigure}
  \hfill
  \begin{subfigure}{0.325\linewidth}
    \includegraphics[width=\linewidth]{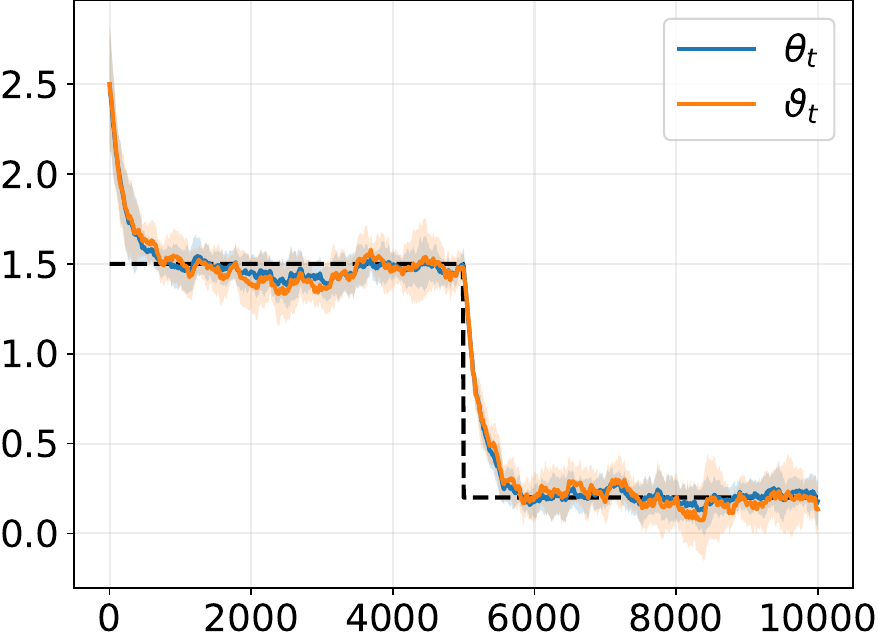}
    \caption{$N=50$.}
    \label{fig:6c}
  \end{subfigure}
  \captionsetup{width=\textwidth}
  \caption{\textbf{Online parameter estimation for the stochastic Kuramoto model}. We plot the sequence of online parameter estimates $\smash{(\theta_t)_{t\geq 0}}$ and $\smash{(\vartheta_t)_{t\geq 0}}$, for $N\in\{3,10,50\}$. The true time-varying parameter is given in \eqref{eq:true-param-kuramoto}. Meanwhile, the initial parameter estimate is given by $\smash{\theta_{\mathrm{init}}\sim \mathcal{U}[2,3]}$.}
  \label{fig:6}
\end{figure}


\section{Conclusion}
\label{sec:conclusion}

In this paper, we introduced two new algorithms for online parameter estimation in interacting particle systems based on continuous observation of a single particle. In both cases, our estimators made use of a collection of auxiliary virtual particles, which enabled a tractable approximation of law-dependent interaction terms that would otherwise be unavailable under such a limited observation regime. Under suitable assumptions, we established convergence of the proposed estimators to the stationary points of finite-particle surrogate objectives as \(t \to \infty\). In the iterated limit \(t\to\infty\) followed by \(N,M\to\infty\), we then showed convergence to the stationary points of the asymptotic log-likelihood of the limiting mean-field process. Finally, we illustrated the effectiveness of our approach in several examples of practical interest, including a model with quadratic confinement and quadratic interaction, a stochastic FitzHugh--Nagumo model of interacting neurons, and a stochastic Kuramoto model.

There are several natural directions for future work. On the theoretical side, the main unresolved question is whether it is possible to obtain easier-to-verify sufficient conditions under which the joint process consisting of the observed particle system, the virtual particle system, and the virtual tangent particle system is ergodic and satisfies uniform-in-time propagation of chaos, and for which the relevant Poisson equations admit unique solutions which satisfy suitable polynomial growth conditions. In this direction, the techniques recently developed in \citet{wang2022forward} are likely to be useful. From a methodological perspective, it would also be interesting to develop adaptive or variance-reduced versions of the algorithm, for example by dynamically tuning the number of virtual particles or by incorporating coupling or control-variate ideas. Finally, it would be interesting to investigate whether the same ideas can be adapted to richer classes of interacting systems, for example models with common noise, heterogeneous particles, or non-exchangeable interaction structure.

 \section*{Acknowledgements} G.A.P. is partially supported by an ERC-EPSRC Frontier Research Guarantee through Grant No. EP/X038645, ERC Advanced Grant No. 247031, and a Leverhulme Trust Senior Research Fellowship, SRF\textbackslash{}R1\textbackslash{}241055.

\bibliography{references}

\begin{thebibliography}{}

\bibitem[Acebr{\'o}n et~al., 2005]{acebron2005kuramoto}
Acebr{\'o}n, J.~A., Bonilla, L.~L., P{\'e}rez~Vicente, C.~J., Ritort, F., and
  Spigler, R. (2005).
\newblock The {Kuramoto} model: A simple paradigm for synchronization
  phenomena.
\newblock {\em Reviews of Modern Physics}, 77(1):137--185.

\bibitem[Amorino et~al., 2025]{amorino2025polynomial}
Amorino, C., Belomestny, D., Pilipauskait{\.e}, V., Podolskij, M., and Zhou,
  S.-Y. (2025).
\newblock Polynomial rates via deconvolution for nonparametric estimation in
  {McKean--Vlasov} {SDEs}.
\newblock {\em Probability Theory and Related Fields}, 193:539--584.

\bibitem[Amorino et~al., 2023]{amorino2023parameter}
Amorino, C., Heidari, A., Pilipauskait{\.e}, V., and Podolskij, M. (2023).
\newblock Parameter estimation of discretely observed interacting particle
  systems.
\newblock {\em Stochastic Processes and their Applications}, 163:350--386.

\bibitem[Amorino and Pilipauskait{\.e}, 2024]{amorino2024kinetic}
Amorino, C. and Pilipauskait{\.e}, V. (2024).
\newblock Kinetic interacting particle system: parameter estimation from
  complete and partial discrete observations.
\newblock {\em arXiv preprint arXiv:2410.10226}.

\bibitem[Baladron et~al., 2012]{baladron2012meanfield}
Baladron, J., Fasoli, D., Faugeras, O., and Touboul, J. (2012).
\newblock Mean-field description and propagation of chaos in networks of
  {Hodgkin--Huxley} and {FitzHugh--Nagumo} neurons.
\newblock {\em Journal of Mathematical Neuroscience}, 2(1):10.

\bibitem[Bashiri, 2020]{bashiri2020longtime}
Bashiri, K. (2020).
\newblock On the long-time behaviour of {McKean--Vlasov} paths.
\newblock {\em Electronic Communications in Probability}, 25:1--14.

\bibitem[Bauer et~al., 2018]{bauer2018strong}
Bauer, M., Meyer-Brandis, T., and Proske, F. (2018).
\newblock Strong solutions of mean-field stochastic differential equations with
  irregular drift.
\newblock {\em Electronic Journal of Probability}, 23:1--35.

\bibitem[Belomestny et~al., 2024]{belomestny2024nonparametric}
Belomestny, D., Podolskij, M., and Zhou, S.-Y. (2024).
\newblock On nonparametric estimation of the interaction function in particle
  system models.
\newblock {\em arXiv preprint arXiv:2402.14419}.

\bibitem[Benachour et~al., 1998]{benachour1998nonlinear}
Benachour, S., Roynette, B., Talay, D., and Vallois, P. (1998).
\newblock Nonlinear self-stabilizing processes {I}: Existence, invariant
  probability, propagation of chaos.
\newblock {\em Stochastic Processes and their Applications}, 75(2):173--201.

\bibitem[Benedetto et~al., 1997]{benedetto1997kinetic}
Benedetto, D., Caglioti, E., and Pulvirenti, M. (1997).
\newblock A kinetic equation for granular media.
\newblock {\em Mathematical Modelling and Numerical Analysis}, 31(5):615--641.

\bibitem[Bertini et~al., 2010]{bertini2010dynamical}
Bertini, L., Giacomin, G., and Pakdaman, K. (2010).
\newblock Dynamical aspects of mean field plane rotators and the {Kuramoto}
  model.
\newblock {\em Journal of Statistical Physics}, 138(1--3):270--290.

\bibitem[Bhudisaksang and Cartea, 2021]{bhudisaksang2021online}
Bhudisaksang, T. and Cartea, {\'A}. (2021).
\newblock Online drift estimation for jump-diffusion processes.
\newblock {\em Bernoulli}, 27(4):2494--2518.

\bibitem[Bishwal, 2011]{bishwal2011estimation}
Bishwal, J. P.~N. (2011).
\newblock Estimation in interacting diffusions: Continuous and discrete
  sampling.
\newblock {\em Applied Mathematics}, 2(9):1154--1158.

\bibitem[Bolley et~al., 2013]{bolley2013uniform}
Bolley, F., Gentil, I., and Guillin, A. (2013).
\newblock Uniform convergence to equilibrium for granular media.
\newblock {\em Archive for Rational Mechanics and Analysis}, 208(2):429--445.

\bibitem[Bourguin et~al., 2026]{bourguin2026quantitative}
Bourguin, S., Dhama, S.~S., and Spiliopoulos, K. (2026).
\newblock Quantitative fluctuation analysis for continuous-time stochastic
  gradient descent via {Malliavin} calculus.
\newblock {\em arXiv preprint arXiv:2603.07149}.

\bibitem[Buckdahn et~al., 2017]{buckdahn2017meanfield}
Buckdahn, R., Li, J., and Ma, J. (2017).
\newblock A mean-field stochastic control problem with partial observations.
\newblock {\em The Annals of Applied Probability}, 27(5):3201--3245.

\bibitem[Burger et~al., 2007]{burger2007aggregation}
Burger, M., Capasso, V., and Morale, D. (2007).
\newblock On an aggregation model with long and short range interactions.
\newblock {\em Nonlinear Analysis: Real World Applications}, 8(3):939--958.

\bibitem[Canuto et~al., 2012]{canuto2012eulerian}
Canuto, C., Fagnani, F., and Tilli, P. (2012).
\newblock An {Eulerian} approach to the analysis of {Krause}'s consensus
  models.
\newblock {\em SIAM Journal on Control and Optimization}, 50(1):243--265.

\bibitem[Cardaliaguet et~al., 2019]{cardaliaguet2019master}
Cardaliaguet, P., Delarue, F., Lasry, J.-M., and Lions, P.-L. (2019).
\newblock {\em The Master Equation and the Convergence Problem in Mean Field
  Games}, volume 201 of {\em Annals of Mathematics Studies}.
\newblock Princeton University Press, Princeton, NJ.

\bibitem[Cardaliaguet and Lehalle, 2018]{cardaliaguet2018mean}
Cardaliaguet, P. and Lehalle, C.-A. (2018).
\newblock Mean field game of controls and an application to trade crowding.
\newblock {\em Mathematics and Financial Economics}, 12(3):335--363.

\bibitem[Carmona and Delarue, 2018]{carmona2018probabilistic}
Carmona, R. and Delarue, F. (2018).
\newblock {\em Probabilistic Theory of Mean Field Games with Applications I}.
\newblock Springer-Verlag, Cham, Switzerland.

\bibitem[Carrillo et~al., 2020]{carrillo2020longtime}
Carrillo, J.~A., Gvalani, R.~S., Pavliotis, G.~A., and Schlichting, A. (2020).
\newblock Long-time behaviour and phase transitions for the {McKean--Vlasov}
  equation on the torus.
\newblock {\em Archive for Rational Mechanics and Analysis}, 235(1):635--690.

\bibitem[Carrillo et~al., 2006]{carrillo2006contractions}
Carrillo, J.~A., McCann, R.~J., and Villani, C. (2006).
\newblock Contractions in the 2-{Wasserstein} length space and thermalization
  of granular media.
\newblock {\em Archive for Rational Mechanics and Analysis}, 179(2):217--263.

\bibitem[Cattiaux et~al., 2008]{cattiaux2008probabilistic}
Cattiaux, P., Guillin, A., and Malrieu, F. (2008).
\newblock Probabilistic approach for granular media equations in the
  non-uniformly convex case.
\newblock {\em Probability Theory and Related Fields}, 140(1--2):19--40.

\bibitem[Chaintron and Diez, 2022a]{chaintron2022propagation}
Chaintron, L.-P. and Diez, A. (2022a).
\newblock Propagation of chaos: A review of models, methods and applications.
  {I}. models and methods.
\newblock {\em Kinetic and Related Models}, 15(6):895--1015.

\bibitem[Chaintron and Diez, 2022b]{chaintron2022propagationa}
Chaintron, L.-P. and Diez, A. (2022b).
\newblock Propagation of chaos: A review of models, methods and applications.
  {II}. applications.
\newblock {\em Kinetic and Related Models}, 15(6):1017--1173.

\bibitem[Chaudru~de Raynal, 2020]{chaudruderaynal2020strong}
Chaudru~de Raynal, P.-E. (2020).
\newblock Strong well posedness of {McKean--Vlasov} stochastic differential
  equations with {H{\"o}lder} drift.
\newblock {\em Stochastic Processes and their Applications}, 130(1):79--107.

\bibitem[Chazelle et~al., 2017]{chazelle2017wellposedness}
Chazelle, B., Jiu, Q., Li, Q., and Wang, C. (2017).
\newblock Well-posedness of the limiting equation of a noisy consensus model in
  opinion dynamics.
\newblock {\em Journal of Differential Equations}, 263(1):365--397.

\bibitem[Chen, 2021]{chen2021maximum}
Chen, X. (2021).
\newblock Maximum likelihood estimation of potential energy in interacting
  particle systems from single-trajectory data.
\newblock {\em Electronic Communications in Probability}, 26:1--13.

\bibitem[Comte and Genon-Catalot, 2023]{comte2022nonparametric}
Comte, F. and Genon-Catalot, V. (2023).
\newblock Nonparametric adaptive estimation for interacting particle systems.
\newblock {\em Scandinavian Journal of Statistics}, 50(4):1716--1755.

\bibitem[Comte et~al., 2025]{comte2024nonparametric}
Comte, F., Genon-Catalot, V., and Lar{\'e}do, C. (2025).
\newblock Nonparametric moment method for scalar {McKean--Vlasov} stochastic
  differential equations.
\newblock {\em ESAIM: Probability and Statistics}, 29:400--449.

\bibitem[Crisan and Xiong, 2010]{crisan2010approximate}
Crisan, D. and Xiong, J. (2010).
\newblock Approximate {McKean--Vlasov} representations for a class of {SPDEs}.
\newblock {\em Stochastics}, 82(1):53--68.

\bibitem[Delgadino et~al., 2023]{delgadino2023phase}
Delgadino, M.~G., Gvalani, R.~S., Pavliotis, G.~A., and Smith, S.~A. (2023).
\newblock Phase transitions, logarithmic {Sobolev} inequalities, and
  uniform-in-time propagation of chaos for weakly interacting diffusions.
\newblock {\em Communications in Mathematical Physics}, 401:275--323.

\bibitem[Della~Maestra and Hoffmann, 2022]{dellamaestra2022nonparametric}
Della~Maestra, L. and Hoffmann, M. (2022).
\newblock Nonparametric estimation for interacting particle systems:
  {McKean--Vlasov} models.
\newblock {\em Probability Theory and Related Fields}, 182(1):551--613.

\bibitem[Della~Maestra and Hoffmann, 2023]{dellamaestra2023lan}
Della~Maestra, L. and Hoffmann, M. (2023).
\newblock The {LAN} property for {McKean--Vlasov} models in a mean-field
  regime.
\newblock {\em Stochastic Processes and their Applications}, 155:109--146.

\bibitem[Durmus et~al., 2020]{durmus2020elementary}
Durmus, A., Eberle, A., Guillin, A., and Zimmer, R. (2020).
\newblock An elementary approach to uniform in time propagation of chaos.
\newblock {\em Proceedings of the American Mathematical Society},
  148(12):5387--5398.

\bibitem[Eberle et~al., 2019]{eberle2019quantitative}
Eberle, A., Guillin, A., and Zimmer, R. (2019).
\newblock Quantitative {Harris}-type theorems for diffusions and
  {McKean--Vlasov} processes.
\newblock {\em Transactions of the American Mathematical Society},
  371(10):7135--7173.

\bibitem[Genon-Catalot and Lar{\'e}do, 2024a]{genoncatalot2022inference}
Genon-Catalot, V. and Lar{\'e}do, C. (2024a).
\newblock Inference for ergodic {McKean--Vlasov} stochastic differential
  equations with polynomial interactions.
\newblock {\em Annales de l'Institut Henri Poincar{\'e} (B) Probabilit{\'e}s et
  Statistiques}, 60(4):2668--2693.

\bibitem[Genon-Catalot and Lar{\'e}do, 2024b]{genoncatalot2023parametric}
Genon-Catalot, V. and Lar{\'e}do, C. (2024b).
\newblock Parametric inference for ergodic {McKean--Vlasov} stochastic
  differential equations.
\newblock {\em Bernoulli}, 30(3):1971--1997.

\bibitem[Gerencs{\'e}r et~al., 1984]{gerencser1984continuoustime}
Gerencs{\'e}r, L., Gy{\"o}ngy, I., and Michaletzky, G. (1984).
\newblock Continuous-time recursive maximum likelihood method: a new approach
  to {Ljung}'s scheme.
\newblock {\em IFAC Proceedings Volumes}, 17(2):683--686.

\bibitem[Gerencs{\'e}r and Prokaj, 2009]{gerencser2009recursive}
Gerencs{\'e}r, L. and Prokaj, V. (2009).
\newblock Recursive identification of continuous-time linear stochastic
  systems: convergence w.p.\ 1 and in {$L^q$}.
\newblock In {\em Proceedings of the 2009 European Control Conference (ECC)},
  pages 1209--1214.

\bibitem[Giesecke et~al., 2020]{giesecke2020inference}
Giesecke, K., Schwenkler, G., and Sirignano, J.~A. (2020).
\newblock Inference for large financial systems.
\newblock {\em Mathematical Finance}, 30(1):3--46.

\bibitem[Goddard et~al., 2022]{goddard2022noisy}
Goddard, B.~D., Gooding, B., Short, H., and Pavliotis, G.~A. (2022).
\newblock Noisy bounded confidence models for opinion dynamics: the effect of
  boundary conditions on phase transitions.
\newblock {\em IMA Journal of Applied Mathematics}, 87(1):80--110.

\bibitem[Heidari and Podolskij, 2025]{heidari2025local}
Heidari, A. and Podolskij, M. (2025).
\newblock Local asymptotic normality for discretely observed mckean-vlasov
  diffusions.
\newblock {\em arXiv preprint arXiv:2511.13366}.

\bibitem[Hu et~al., 2021]{hu2020meanfield}
Hu, K., Ren, Z., Šiška, D., and Szpruch, {\L{}}. (2021).
\newblock {Mean-field Langevin dynamics and energy landscape of neural
  networks}.
\newblock {\em Annales de l'Institut Henri Poincaré, Probabilités et
  Statistiques}, 57(4):2043 -- 2065.

\bibitem[Huang and Wang, 2019]{huang2019distribution}
Huang, X. and Wang, F.-Y. (2019).
\newblock Distribution dependent {SDEs} with singular coefficients.
\newblock {\em Stochastic Processes and their Applications},
  129(11):4747--4770.

\bibitem[Iguchi et~al., 2025]{iguchi2025parameter}
Iguchi, Y., Beskos, A., and Pavliotis, G.~A. (2025).
\newblock Parameter estimation for weakly interacting hypoelliptic diffusions.
\newblock {\em arXiv preprint arXiv:2508.04287}.

\bibitem[Jasra et~al., 2025]{jasra2025parameter}
Jasra, A., Maama, M., and Tempone, R.~F. (2025).
\newblock Parameter estimation for partially observed {McKean--Vlasov}
  diffusions.
\newblock {\em Royal Society Open Science}, 12(12):251918.

\bibitem[Jasra and Wu, 2025]{jasra2025bayesian}
Jasra, A. and Wu, A. (2025).
\newblock Bayesian parameter estimation for partially observed {McKean--Vlasov}
  diffusions using multilevel {Markov} chain {Monte Carlo}.
\newblock {\em Statistics and Computing}, 35(6):210.

\bibitem[Jourdain et~al., 2008]{jourdain2008nonlinear}
Jourdain, B., M{\'e}l{\'e}ard, S., and Woyczynski, W.~A. (2008).
\newblock Nonlinear {SDEs} driven by {L{\'e}vy} processes and related {PDEs}.
\newblock {\em ALEA: Latin American Journal of Probability and Mathematical
  Statistics}, 4:1--29.

\bibitem[Kasonga, 1990]{kasonga1990maximum}
Kasonga, R.~A. (1990).
\newblock Maximum likelihood theory for large interacting systems.
\newblock {\em SIAM Journal on Applied Mathematics}, 50(3):865--875.

\bibitem[Kuramoto, 1981]{kuramoto1981rhythms}
Kuramoto, Y. (1981).
\newblock Rhythms and turbulence in populations of chemical oscillators.
\newblock {\em Physica A: Statistical Mechanics and its Applications},
  106(1--2):128--143.

\bibitem[Lacker, 2018]{lacker2018strong}
Lacker, D. (2018).
\newblock On a strong form of propagation of chaos for {McKean--Vlasov}
  equations.
\newblock {\em Electronic Communications in Probability}, 23:1--11.

\bibitem[Lacker, 2023]{lacker2023hierarchies}
Lacker, D. (2023).
\newblock Hierarchies, entropy, and quantitative propagation of chaos for mean
  field diffusions.
\newblock {\em Probability and Mathematical Physics}, 4(2):377--432.

\bibitem[Lacker and Le~Flem, 2023]{lacker2023sharp}
Lacker, D. and Le~Flem, L. (2023).
\newblock Sharp uniform-in-time propagation of chaos.
\newblock {\em Probability Theory and Related Fields}, 187(1--2):443--480.

\bibitem[Lang and Lu, 2023]{lang2021identifiability}
Lang, Q. and Lu, F. (2023).
\newblock Identifiability of interaction kernels in mean-field equations of
  interacting particles.
\newblock {\em Foundations of Data Science}, 5(4):480--502.

\bibitem[Levanony et~al., 1994]{levanony1994recursive}
Levanony, D., Shwartz, A., and Zeitouni, O. (1994).
\newblock Recursive identification in continuous-time stochastic processes.
\newblock {\em Stochastic Processes and their Applications}, 49(2):245--275.

\bibitem[Liu and Wang, 2016]{liu2016stein}
Liu, Q. and Wang, D. (2016).
\newblock Stein variational gradient descent: A general purpose {Bayesian}
  inference algorithm.
\newblock In {\em Proceedings of the 30th Annual Conference on Neural
  Information Processing Systems (NeurIPS 2016)}.

\bibitem[Lu et~al., 2021]{lu2021learning}
Lu, F., Maggioni, M., and Tang, S. (2021).
\newblock Learning interaction kernels in heterogeneous systems of agents from
  multiple trajectories.
\newblock {\em Journal of Machine Learning Research}, 22(32):1--67.

\bibitem[Lu et~al., 2019]{lu2019nonparametric}
Lu, F., Zhong, M., Tang, S., and Maggioni, M. (2019).
\newblock Nonparametric inference of interaction laws in systems of agents from
  trajectory data.
\newblock {\em Proceedings of the National Academy of Sciences},
  116(29):14424--14433.

\bibitem[Lu{\c c}on and Poquet, 2021]{luccon2021periodicity}
Lu{\c c}on, E. and Poquet, C. (2021).
\newblock Periodicity induced by noise and interaction in the kinetic
  mean-field {FitzHugh--Nagumo} model.
\newblock {\em The Annals of Applied Probability}, 31(2):561--593.

\bibitem[Malrieu, 2001]{malrieu2001logarithmic}
Malrieu, F. (2001).
\newblock Logarithmic {Sobolev} inequalities for some nonlinear {PDE}'s.
\newblock {\em Stochastic Processes and their Applications}, 95(1):109--132.

\bibitem[Malrieu, 2003]{malrieu2003convergence}
Malrieu, F. (2003).
\newblock Convergence to equilibrium for granular media equations and their
  {Euler} schemes.
\newblock {\em The Annals of Applied Probability}, 13(2):540--560.

\bibitem[McKean, 1966]{mckean1966class}
McKean, Jr., H.~P. (1966).
\newblock A class of {Markov} processes associated with nonlinear parabolic
  equations.
\newblock {\em Proceedings of the National Academy of Sciences of the United
  States of America}, 56(6):1907--1911.

\bibitem[Mei et~al., 2018]{mei2018mean}
Mei, S., Montanari, A., and Nguyen, P.-M. (2018).
\newblock A mean field view of the landscape of two-layer neural networks.
\newblock {\em Proceedings of the National Academy of Sciences},
  115(33):E7665--E7671.

\bibitem[M{\'e}l{\'e}ard, 1996]{meleard1996asymptotic}
M{\'e}l{\'e}ard, S. (1996).
\newblock Asymptotic behaviour of some interacting particle systems;
  {McKean--Vlasov} and {Boltzmann} models.
\newblock In Talay, D. and Tubaro, L., editors, {\em Probabilistic Models for
  Nonlinear Partial Differential Equations}, volume 1627 of {\em Lecture Notes
  in Mathematics}, pages 42--95. Springer, Berlin, Heidelberg.

\bibitem[Mishura and Veretennikov, 2020]{mishura2020existence}
Mishura, Y.~S. and Veretennikov, A.~Y. (2020).
\newblock Existence and uniqueness theorems for solutions of {McKean--Vlasov}
  stochastic equations.
\newblock {\em Theory of Probability and Mathematical Statistics}, 103:59--101.

\bibitem[Nickl et~al., 2025]{nickl2025bayesian}
Nickl, R., Pavliotis, G.~A., and Ray, K. (2025).
\newblock Bayesian nonparametric inference in {McKean--Vlasov} models.
\newblock {\em The Annals of Statistics}, 53(1):170--193.

\bibitem[Oelschl{\"a}ger, 1984]{oelschlager1984martingale}
Oelschl{\"a}ger, K. (1984).
\newblock A martingale approach to the law of large numbers for weakly
  interacting stochastic processes.
\newblock {\em The Annals of Probability}, 12(2):458--479.

\bibitem[{\O}ksendal, 2003]{oksendal2003stochastic}
{\O}ksendal, B. (2003).
\newblock {\em Stochastic Differential Equations: An Introduction with
  Applications}.
\newblock Springer-Verlag, 6th edition.

\bibitem[Pavliotis et~al., 2025]{pavliotis2025filtered}
Pavliotis, G.~A., Reich, S., and Zanoni, A. (2025).
\newblock Filtered data based estimators for stochastic processes driven by
  colored noise.
\newblock {\em Stochastic Processes and their Applications}, 181:104558.

\bibitem[Pavliotis and Zanoni, 2022]{pavliotis2022eigenfunction}
Pavliotis, G.~A. and Zanoni, A. (2022).
\newblock Eigenfunction martingale estimators for interacting particle systems
  and their mean field limit.
\newblock {\em SIAM Journal on Applied Dynamical Systems}, 21(4):2338--2370.

\bibitem[Pavliotis and Zanoni, 2024]{pavliotis2024method}
Pavliotis, G.~A. and Zanoni, A. (2024).
\newblock A method of moments estimator for interacting particle systems and
  their mean field limit.
\newblock {\em SIAM/ASA Journal on Uncertainty Quantification}, 12(2):262--288.

\bibitem[Pavliotis and Zanoni, 2026]{pavliotis2026fourierbased}
Pavliotis, G.~A. and Zanoni, A. (2026).
\newblock A {Fourier-based} inference method for learning interaction kernels
  in particle systems.
\newblock {\em SIAM Journal on Applied Mathematics}, 86(2):615--643.

\bibitem[Rotskoff and Vanden-Eijnden, 2022]{rotskoff2022trainability}
Rotskoff, G.~M. and Vanden-Eijnden, E. (2022).
\newblock Trainability and accuracy of artificial neural networks: An
  interacting particle system approach.
\newblock {\em Communications on Pure and Applied Mathematics},
  75(9):1889--1935.

\bibitem[Sakaguchi et~al., 1988]{sakaguchi1988phase}
Sakaguchi, H., Shinomoto, S., and Kuramoto, Y. (1988).
\newblock Phase transitions and their bifurcation analysis in a large
  population of active rotators with mean-field coupling.
\newblock {\em Progress of Theoretical Physics}, 79(3):600--607.

\bibitem[Sharrock, 2022a]{sharrock2022theory}
Sharrock, L. (2022a).
\newblock {\em {On the Theory and Applications of Stochastic Gradient Descent
  in Continuous Time}}.
\newblock {PhD} thesis, Imperial College London.

\bibitem[Sharrock, 2022b]{sharrock2022twotimescale}
Sharrock, L. (2022b).
\newblock Two-timescale stochastic approximation for bilevel optimisation
  problems in continuous-time models.
\newblock In {\em Proceedings of the 39th International Conference on Machine
  Learning (ICML 2022): Workshop on Continuous Time Methods for Machine
  Learning}.

\bibitem[Sharrock and Kantas, 2022]{sharrock2022joint}
Sharrock, L. and Kantas, N. (2022).
\newblock Joint online parameter estimation and optimal sensor placement for
  the partially observed stochastic advection diffusion equation.
\newblock {\em SIAM/ASA Journal on Uncertainty Quantification}, 10(1):55--95.

\bibitem[Sharrock and Kantas, 2023]{sharrock2023twotimescale}
Sharrock, L. and Kantas, N. (2023).
\newblock Two-timescale stochastic gradient descent in continuous time with
  applications to joint online parameter estimation and optimal sensor
  placement.
\newblock {\em Bernoulli}, 29(2):1137--1165.

\bibitem[Sharrock et~al., 2022]{sharrock2022parameterestimation}
Sharrock, L., Kantas, N., Parpas, P., and Pavliotis, G.~A. (2022).
\newblock Parameter estimation for the {McKean--Vlasov} stochastic differential
  equation.
\newblock {\em arXiv preprint arXiv:2106.13751v3}.

\bibitem[Sharrock et~al., 2023]{sharrock2023online}
Sharrock, L., Kantas, N., Parpas, P., and Pavliotis, G.~A. (2023).
\newblock Online parameter estimation for the {McKean--Vlasov} stochastic
  differential equation.
\newblock {\em Stochastic Processes and their Applications}, 162:481--546.

\bibitem[Sharrock et~al., 2026]{sharrock2026efficient}
Sharrock, L., Kantas, N., and Pavliotis, G.~A. (2026).
\newblock Efficient online learning in interacting particle systems.
\newblock {\em arXiv preprint arXiv:2602.20875}.

\bibitem[Sirignano and Spiliopoulos, 2017]{sirignano2017stochastic}
Sirignano, J. and Spiliopoulos, K. (2017).
\newblock Stochastic gradient descent in continuous time.
\newblock {\em SIAM Journal on Financial Mathematics}, 8(1):933--961.

\bibitem[Sirignano and Spiliopoulos, 2020a]{sirignano2020mean}
Sirignano, J. and Spiliopoulos, K. (2020a).
\newblock Mean field analysis of neural networks: A law of large numbers.
\newblock {\em SIAM Journal on Applied Mathematics}, 80(2):725--752.

\bibitem[Sirignano and Spiliopoulos, 2020b]{sirignano2020stochastic}
Sirignano, J. and Spiliopoulos, K. (2020b).
\newblock Stochastic gradient descent in continuous time: A central limit
  theorem.
\newblock {\em Stochastic Systems}, 10(2):124--151.

\bibitem[Surace and Pfister, 2019]{surace2019online}
Surace, S.~C. and Pfister, J.-P. (2019).
\newblock Online maximum-likelihood estimation of the parameters of partially
  observed diffusion processes.
\newblock {\em IEEE Transactions on Automatic Control}, 64(7):2814--2829.

\bibitem[Sznitman, 1991]{sznitman1991topics}
Sznitman, A.-S. (1991).
\newblock Topics in propagation of chaos.
\newblock In {\em Ecole d'Et{\'e} de Probabilit{\'e}s de Saint-Flour XIX --
  1989}, volume 1464 of {\em Lecture Notes in Mathematics}, pages 165--251.
  Springer, Berlin, Heidelberg.

\bibitem[Vlasov, 1968]{vlasov1968vibrational}
Vlasov, A.~A. (1968).
\newblock The vibrational properties of an electron gas.
\newblock {\em Soviet Physics Uspekhi}, 10(6):721--733.

\bibitem[Wang and Sirignano, 2022]{wang2022forward}
Wang, Z. and Sirignano, J. (2022).
\newblock A forward propagation algorithm for online optimization of nonlinear
  stochastic differential equations.
\newblock {\em arXiv preprint arXiv:2207.04496}.

\bibitem[Wang and Sirignano, 2024]{wang2024continuous}
Wang, Z. and Sirignano, J. (2024).
\newblock Continuous-time stochastic gradient descent for optimizing over the
  stationary distribution of stochastic differential equations.
\newblock {\em Mathematical Finance}, 34(2):348--424.

\bibitem[Yao et~al., 2022]{yao2022meanfield}
Yao, R., Chen, X., and Yang, Y. (2022).
\newblock Mean-field nonparametric estimation of interacting particle systems.
\newblock In {\em Proceedings of the Thirty Fifth Conference on Learning Theory
  (COLT 2022)}, volume 178 of {\em Proceedings of Machine Learning Research},
  pages 2242--2275. PMLR.

\end{thebibliography}

\appendix 

\section{Proofs for Section~\ref{sec:methodology}}
\label{appendix:proofs-section3}

\begin{lemma}
\label{lem:appendix-diff-under-integral}
Suppose that Assumptions~\ref{assumption:moments}--\ref{assumption:drift} and
Assumption~\ref{ass:appendix-meanfield-c1} hold. Let $\Phi:\Theta\times\mathbb{R}^d\times\mathcal{P}(\mathbb{R}^d)\rightarrow\mathbb{R}^d$ be defined according to
\begin{equation}
\Phi(\theta,x,\mu)
=
\Phi_0(\theta,x)
+
\int_{\mathbb R^d} K(\theta,x,y)\,\mu(\mathrm d y),
\end{equation}
where $\Phi_0(\cdot,x), K(\cdot,x,y)\in C^1(\Theta)$ for all $x,y\in\mathbb{R}^d$, and $\Phi_0(\theta,\cdot), K(\theta,\cdot,\cdot)$ have polynomial growth for all $\theta\in\Theta$. Then
\begin{equation}
\partial_\theta \Phi(\theta,x,\pi_\theta)
=
\partial_\theta \Phi_0(\theta,x)
+
\int_{\mathbb R^d}\partial_\theta K(\theta,x,y)\,\pi_\theta(\mathrm d y)
+
\int_{\mathbb R^d}K(\theta,x,y)\,\nu_\theta(\mathrm d y).
\end{equation}
\end{lemma}

\begin{proof}
The result follows immediately from the dominated convergence theorem. In particular, the polynomial-growth bounds and the moment bounds in Assumption~\ref{ass:appendix-meanfield-c1} provide an integrable majorant, while the definition of $\nu_\theta$ yields the last term.
\end{proof}

\begin{proof}[Proof of Proposition~\ref{prop:ips-likelihood-n-t-limit-grad}]
Since the invariant distribution $\pi_{\theta_0}$ does not depend on the parameter $\theta$, differentiation under the integral sign gives
\begin{equation}
\partial_\theta \mathcal{J}(\theta)
=
\int_{\mathbb R^d}
\bigl[\partial_\theta B(\theta,x,\pi_\theta)\bigr]
(\sigma\sigma^{\top})^{-1}[B(\theta,x,\pi_\theta)-B(\theta_0,x,\pi_{\theta_0})]\,
\pi_{\theta_0}(\mathrm d x).
\end{equation}
By Lemma~\ref{lem:appendix-diff-under-integral}, we have
\begin{equation}
\partial_\theta B(\theta,x,\pi_\theta)
=
\int_{\mathbb R^d}\partial_\theta b(\theta,x,y)\,\pi_\theta(\mathrm d y)
+
\int_{\mathbb R^d}b(\theta,x,y)\,\nu_\theta(\mathrm d y)
=
G(\theta,x,\pi_\theta,\nu_\theta).
\end{equation}
Substituting this identity yields the claimed result.
\end{proof}

\begin{proof}[Proof of Proposition~\ref{prop:ips-likelihood-n-t-limit-grad-explicit}]
By Proposition~\ref{prop:ips-likelihood-n-t-limit-grad}, 
\begin{equation}
\partial_\theta \mathcal{J}(\theta)
=
\int_{\mathbb R^d}
G(\theta,x,\pi_\theta,\nu_\theta)(\sigma\sigma^{\top})^{-1}
\bigl(B(\theta,x,\pi_\theta)-B(\theta_0,x,\pi_{\theta_0})\bigr)\,
\pi_{\theta_0}(\mathrm d x).
\end{equation}
By definition, $G(\theta,x,\pi_\theta,\nu_\theta)
=
\int_{\mathbb R^d} g(\theta,x,y,\nu_\theta)\,\pi_\theta(\mathrm d y)$ and $
B(\theta,x,\pi_\theta)-B(\theta_0,x,\pi_{\theta_0})
=
\int_{\mathbb R^d}
\bigl(b(\theta,x,z)-B(\theta_0,x,\pi_{\theta_0})\bigr)\,
\pi_\theta(\mathrm d z)$. By Fubini's theorem, we thus have that
\begin{equation}
\partial_\theta \mathcal{J}(\theta)
=
\int_{(\mathbb R^d)^3}
h(\theta,x,y,\nu_\theta,z,\pi_{\theta_0})\,
\pi_{\theta_0}(\mathrm d x)\pi_\theta(\mathrm d y)\pi_\theta(\mathrm d z),
\end{equation}
which is exactly the claimed formula.
\end{proof}

\section{Proofs for Section~\ref{sec:prelim-results}}
\label{appendix:proofs-section43}

\begin{proof}[Proof of Proposition~\ref{prop:asymptotic-partial-log-lik-grad}]
We begin with the observation that the two finite-particle surrogate objectives can be rewritten as
\begin{align*}
\mathcal J^{i,N,M}(\theta)
&=
\int_{(\mathbb R^d)^N}
J(\theta,x^{i,N},\pi_\theta^{1,M},\pi_\theta^{1,M},\mu^N)\,
\pi_{\theta_0}^N(\mathrm d x^N),\\
\mathcal J^{i,j,k,N,M}(\theta)
&=
\int_{(\mathbb R^d)^N}\int_{\mathbb R^d}\int_{\mathbb R^d}
j(\theta,x^{i,N},y,z,\mu^N)\,
\pi_\theta^{1,M}(\mathrm d y)\pi_\theta^{1,M}(\mathrm d z)\,
\pi_{\theta_0}^N(\mathrm d x^N).
\end{align*}
Arguing as in the proof of Proposition~\ref{prop:ips-likelihood-n-t-limit-grad} and Proposition~\ref{prop:ips-likelihood-n-t-limit-grad-explicit}, but now with $\pi_\theta$ and $\nu_\theta$ replaced by $\pi_\theta^{1,M}$ and $\nu_\theta^{1,M}$, we then have that
\begin{align}
\partial_\theta \mathcal J^{i,N,M}(\theta)
&=
\int_{(\mathbb R^d)^N}
H(\theta,x^{i,N},\pi_\theta^{1,M},\nu_\theta^{1,M},\pi_\theta^{1,M},\mu^N)\,
\pi_{\theta_0}^N(\mathrm d x^N) \\
\partial_\theta \mathcal J^{i,j,k,N,M}(\theta)
&=
\int_{(\mathbb R^d)^N}\int_{\mathbb R^d}\int_{\mathbb R^d}
h(\theta,x^{i,N},y,\nu_\theta^{1,M},z,\mu^N)\,
\pi_\theta^{1,M}(\mathrm d y)\pi_\theta^{1,M}(\mathrm d z)\,
\pi_{\theta_0}^N(\mathrm d x^N).
\end{align}
\end{proof}

\begin{proof}[Proof of Proposition~\ref{prop:j-vp}]
Under $\Pi_\theta^{N,M}$ the two virtual particle systems are independent, each is exchangeable, and both have first marginal $\pi_\theta^{1,M}$. By the definition of $j$ and bilinearity of the inner product,
\begin{equation}
\mathcal J^{i,N,M}(\theta)
=
\frac1{M^2}\sum_{a=1}^M\sum_{b=1}^M
\mathbb E_{\Pi_\theta^{N,M}}
\Bigl[
j(\theta,x^{i,N},\hat x^{a,M},\tilde x^{b,M},\mu^N)
\Bigr].
\end{equation}
By exchangeability, every summand on the right-hand side is the same, hence equal to $\mathcal J^{i,j,k,N,M}(\theta)$. This establishes that $\mathcal J^{i,N,M}(\theta)=\mathcal J^{i,j,k,N,M}(\theta)$. Meanwhile, integrating out the virtual blocks gives
\begin{align*}
\mathcal J^{i,j,k,N,M}(\theta)
&=
\int_{(\mathbb R^d)^N}\int_{\mathbb R^d}\int_{\mathbb R^d}
j(\theta,x^{i,N},y,z,\mu^N)\,
\pi_\theta^{1,M}(\mathrm d y)\pi_\theta^{1,M}(\mathrm d z)\,
\pi_{\theta_0}^N(\mathrm d x^N)\\
&=
\int_{(\mathbb R^d)^N}
J(\theta,x^{i,N},\pi_\theta^{1,M},\pi_\theta^{1,M},\mu^N)\,
\pi_{\theta_0}^N(\mathrm d x^N)=:\mathcal J_{\mathrm{vp}}^{i,N,M}(\theta).
\end{align*}
Finally, if in addition the assumptions of Proposition~\ref{prop:asymptotic-partial-log-lik-grad} hold, then the gradient identities in Proposition~\ref{prop:asymptotic-partial-log-lik-grad} imply
\begin{equation}
\partial_\theta \mathcal J^{i,N,M}(\theta)
=
\partial_\theta \mathcal J^{i,j,k,N,M}(\theta)
=
\int_{(\mathbb R^d)^N}
H(\theta,x^{i,N},\pi_\theta^{1,M},\nu_\theta^{1,M},\pi_\theta^{1,M},\mu^N)\,
\pi_{\theta_0}^N(\mathrm d x^N).
\end{equation}
\end{proof}

\begin{lemma}
\label{lem:appendix-H-lipschitz}
Suppose that Assumption~\ref{assumption:drift} holds. Then, writing $\smash{M_m(\gamma):=\int_{\mathbb R^d}(1+\|u\|^m)\,\gamma(\mathrm d u)}$ for non-negative measures $\gamma$, there exist integers $m,r\ge 1$ and a constant $C<\infty$ such that, for every $\theta\in\Theta$,
\begin{align*}
&\Big\|H(\theta,x,\mu,\eta,\bar\mu,\lambda)-H(\theta,x',\mu',\eta',\bar\mu',\lambda')\Big\| \\
&\le
C\Big(
1+\!\!\sum_{z\in\{x,x'\}}\|z\|^m+\!\!\!\!\!\!\!\!\sum_{\gamma\in\{\mu,\mu',\bar{\mu},\bar{\mu}',\lambda,\lambda'\}}\!\!\!\!\!\!M_m(\gamma)+\!\!\sum_{\gamma\in\{\eta,\eta'\}}\!\!M_m(|\gamma|)
\Big)
\Big(
\|x-x'\|
+\!\!\!\!\!\sum_{\gamma\in\{\mu,\eta,\bar{\mu}\}}\!\!\|\gamma-\gamma'\|_{\mathrm{TV},r} 
+ 
\mathsf{W}_{1,r}(\lambda,\lambda')
\Big).
\end{align*}
\end{lemma}

\begin{proof}
By Assumption~\ref{assumption:drift}, the maps $b$, $\partial_\theta b$, and $\partial_y b$ are locally Lipschitz in the spatial variables with polynomially growing Lipschitz constants. Consequently, there exist integers $m,r\ge1$ and a constant $C<\infty$ such that, uniformly in $\theta\in\Theta$,
\begin{align*}
\|B(\theta,x,\mu)-B(\theta,x',\mu')\|
&\le C\big(1+\!\!\sum_{z\in\{x,x'\}}\|z\|^m +\!\!\sum_{\gamma\in\{\mu,\mu'\}}M_m(\gamma)\big)
\big(\|x-x'\|+\|\mu-\mu'\|_{\mathrm{TV},r}\big),
\\
\|G(\theta,x,\mu,\eta)-G(\theta,x',\mu',\eta')\|
&\le C\big(1+\!\!\sum_{z\in\{x,x'\}}\|z\|^m +\!\!\sum_{\gamma\in\{\mu,\mu'\}}M_m(\gamma)+\!\!\sum_{\gamma\in\{\eta,\eta'\}}M_m(|\gamma|)\big)
\\[-1mm]
&\quad\times\big(\|x-x'\|+\!\!\sum_{\gamma\in\{\mu,\eta\}}\|\gamma-\gamma'\|_{\mathrm{TV},r}\big).
\end{align*}
and, similarly, 
\begin{align*}
\|B(\theta_0,x,\lambda)-B(\theta_0,x',\lambda')\|
&\le C\big(1+\!\!\sum_{z\in\{x,x'\}}\|z\|^m +\!\!\sum_{\gamma\in\{\lambda,\lambda'\}}M_m(\gamma)\big)
\big(\|x-x'\|+\mathsf{W}_{1,r}(\lambda,\lambda')\big),
\end{align*}
By definition, $H(\theta,x,\mu,\eta,\bar\mu,\lambda)
=
G(\theta,x,\mu,\eta)(\sigma\sigma^{\top})^{-1}
(B(\theta,x,\bar\mu)-B(\theta_0,x,\lambda))$. 
Applying the product rule, the previous three bounds, and the polynomial-growth bounds implied by Assumption~\ref{assumption:drift}, yields the stated estimate.
\end{proof}

\begin{proof}[Proof of Proposition~\ref{prop:inf-n-convergence-1}]
Due to the uniform-in-time propagation-of-chaos imposed in Assumption~\ref{assumption:model}, there exists a deterministic sequence $\varepsilon_N\downarrow 0$ such that
\begin{equation}
\mathbb E\Bigl[\|x^{i,N}-\bar x\|+\mathsf{W}_{1,r}(\mu^N,\pi_{\theta_0})\Bigr]\le \varepsilon_N.
\end{equation}
where the expectation is with respect to a suitable coupling of $x^{i,N}\sim \pi_{\theta_0}^{1,N}$, $\bar x\sim \pi_{\theta_0}$, and $\mu^{N}=\mu^N(\boldsymbol{x}^N)$ with $\boldsymbol{x}^N\sim \pi_{\theta_0}^N$. Fix $\theta\in\Theta$. By Proposition~\ref{prop:ips-likelihood-n-t-limit-grad} and Proposition~\ref{prop:j-vp},
\begin{align*}
&\|\partial_\theta \mathcal J_{\mathrm{vp}}^{i,N,M}(\theta)-\partial_\theta \mathcal{J}(\theta)\|\le
\mathbb E\Bigl[
\bigl\|
H(\theta,x^{i,N},\pi_\theta^{1,M},\nu_\theta^{1,M},\pi_\theta^{1,M},\mu^N)
-
H(\theta,\bar x,\pi_\theta,\nu_\theta,\pi_\theta,\pi_{\theta_0})
\bigr\|
\Bigr].
\end{align*}
Apply Lemma~\ref{lem:appendix-H-lipschitz}. The moment factor is uniformly bounded by Assumption~\ref{assumption:model} and Assumption~\ref{ass:appendix-meanfield-c1}. Meanwhile, the metric factor is bounded by
\begin{equation}
\mathbb E\Bigl[\|x^{i,N}-\bar x\|+\mathsf{W}_{1,r}(\mu^N,\pi_{\theta_0})\Bigr]
+
\sup_{\theta\in\Theta}
\Bigl(
\|\pi_\theta^{1,M}-\pi_\theta\|_{\mathrm{TV},r}
+
\|\nu_\theta^{1,M}-\nu_\theta\|_{\mathrm{TV},r}
\Bigr).
\end{equation}
By construction and Assumption~\ref{ass:appendix-derivative-poc}, this is at most $\varepsilon_N+\delta_M$. Taking the supremum over $\theta\in\Theta$ proves the claim.
\end{proof}

\begin{proof}[Proof of Proposition~\ref{prop:add:markovian-projection}]
Let $x^N\sim\pi_{\theta_0}^N$, and define $U^{i,N}:=B(\theta_0,x^{i,N},\mu^N)$. By definition of $\mathcal J_N^i$, it follows that
\begin{equation}
\mathcal J_N^i(\theta)
=
\frac12
\mathbb E\Bigl[
\|B(\theta,x^{i,N},\pi_\theta)-U^{i,N}\|_{\sigma\sigma^{\top}}^2
\Bigr].
\end{equation}
Suppose we now add and subtract $B_{\theta_0}^{i,N}(x^{i,N})$ inside the norm. Expanding the resulting quadratic, we arrive at the decomposition
\begin{align*}
\mathcal J_N^i(\theta)
&=
\frac{1}{2}\mathbb E\Bigl[
\|B(\theta,x^{i,N},\pi_\theta)-B_{\theta_0}^{i,N}(x^{i,N})\|_{\sigma\sigma^{\top}}^2
\Bigr] +
\frac{1}{2}\mathbb E\Bigl[
\|U^{i,N}-B_{\theta_0}^{i,N}(x^{i,N})\|_{\sigma\sigma^{\top}}^2
\Bigr]\\
&\quad+
\mathbb E\Bigl[
\bigl\langle
B(\theta,x^{i,N},\pi_\theta)-B_{\theta_0}^{i,N}(x^{i,N}),
B_{\theta_0}^{i,N}(x^{i,N})-U^{i,N}
\bigr\rangle_{\sigma\sigma^{\top}}
\Bigr].
\end{align*}
The last term vanishes when conditioning on $x^{i,N}$, since $\mathbb E[B_{\theta_0}^{i,N}(x^{i,N})-U^{i,N}\mid x^{i,N}]=0$ by the definition of $B_{\theta_0}^{i,N}$. We thus have, as claimed, that 
\begin{equation}
\mathcal J_N^i(\theta)
=
C_N^i
+
\frac12
\int_{\mathbb R^d}
\|B(\theta,x,\pi_\theta)-B_{\theta_0}^{i,N}(x)\|_{\sigma\sigma^{\top}}^2\,
\pi_{\theta_0}^{i,N}(\mathrm d x),
\qquad C_N^i
:=
\frac12
\mathbb E\Bigl[
\|U^{i,N}-B_{\theta_0}^{i,N}(x^{i,N})\|_{\sigma\sigma^{\top}}^2
\Bigr].
\end{equation}
\end{proof}

\begin{proof}[Proof of Proposition~\ref{prop:add:separate-N-M}]
Arguing as in the proof of Proposition~\ref{prop:ips-likelihood-n-t-limit-grad}, now with $\mu^N$ fixed and only the mean-field law depending on $\theta$, we have that
\begin{equation}
\partial_\theta \mathcal J_N^i(\theta)
=
\int_{(\mathbb R^d)^N}
H(\theta,x^{i,N},\pi_\theta,\nu_\theta,\pi_\theta,\mu^N)\,
\pi_{\theta_0}^N(\mathrm d x^N).
\end{equation}
For the quantitative estimates, fix $\theta\in\Theta$. Then, using Proposition~\ref{prop:ips-likelihood-n-t-limit-grad} and the displayed gradient formula above, we have
\begin{align*}
\|\partial_\theta \mathcal J_{\mathrm{vp}}^{i,N,M}(\theta)-\partial_\theta \mathcal J_N^i(\theta)\| \le
\mathbb E\Bigl[
\bigl\|
H(\theta,x^{i,N},\pi_\theta^{1,M},\nu_\theta^{1,M},\pi_\theta^{1,M},\mu^N)
-
H(\theta,x^{i,N},\pi_\theta,\nu_\theta,\pi_\theta,\mu^N)
\bigr\|
\Bigr],
\end{align*}
where $(x^{i,N},\mu^N)$ is distributed under $\pi_{\theta_0}^N$. Apply Lemma~\ref{lem:appendix-H-lipschitz}. The spatial moment factor is uniformly bounded by Assumption~\ref{assumption:model} and Assumption~\ref{ass:appendix-meanfield-c1}, while the metric factor is bounded by
\begin{equation}
\sup_{\theta\in\Theta}
\Bigl(
\|\pi_\theta^{1,M}-\pi_\theta\|_{\mathrm{TV},r}
+
\|\nu_\theta^{1,M}-\nu_\theta\|_{\mathrm{TV},r}
\Bigr)
\le \delta_M.
\end{equation}
This proves the first estimate. For the second estimate, let $(x^{i,N},\mu^N,\bar x)$ be the coupling used in the proof of Proposition~\ref{prop:inf-n-convergence-1}, so that
\begin{equation}
\mathbb E\Bigl[\|x^{i,N}-\bar x\|+\mathsf{W}_{1,r}(\mu^N,\pi_{\theta_0})\Bigr]\le \varepsilon_N.
\end{equation}
Then, using the displayed formula for $\partial_\theta \mathcal J_N^i(\theta)$ and Proposition~\ref{prop:ips-likelihood-n-t-limit-grad}, it follows that
\begin{align}
\|\partial_\theta \mathcal J_N^i(\theta)-\partial_\theta \mathcal{J}(\theta)\| \le
\mathbb E\Bigl[
\bigl\|
H(\theta,x^{i,N},\pi_\theta,\nu_\theta,\pi_\theta,\mu^N)
-
H(\theta,\bar x,\pi_\theta,\nu_\theta,\pi_\theta,\pi_{\theta_0})
\bigr\|
\Bigr].
\end{align}
Applying Lemma~\ref{lem:appendix-H-lipschitz}, together with the same moment bounds, shows that the right-hand side is at most $C\varepsilon_N$. Finally, taking the supremum over $\theta\in\Theta$ proves the claim.
\end{proof}

\section{Proofs for Section~\ref{sec:main-results}}
\label{appendix:proofs-section44}
Let $N,M\in\mathbb{N}$ be fixed. Let $\boldsymbol z_t^{N,M}$ denote the concatenated state appearing in the update equation in \eqref{eq:appendix-alg-H}, and $\widetilde{\boldsymbol z}_t^{N,M}$ the concatenated state appearing in the update equation in \eqref{eq:appendix-alg-h}. Due to Propositions~\ref{prop:asymptotic-partial-log-lik-grad} and~\ref{prop:j-vp}, the two algorithms can then be rewritten in the form
\begin{align}
\mathrm d\theta_t
&=
-\gamma_t\partial_\theta \mathcal J_{\mathrm{vp}}^{i,N,M}(\theta_t)\,\mathrm d t
-\gamma_t F_H^{i,N,M}(\theta_t,\boldsymbol{z}_t^{N,M})\,\mathrm d t
+\gamma_t G^{i,N,M}(\theta_t,\boldsymbol{z}_t^{N,M})\sigma^{-\top}\,\mathrm d w_t^{i,N},
\label{eq:appendix-alg-H}\\
\mathrm d\vartheta_t
&=
-\gamma_t\partial_\theta \mathcal J_{\mathrm{vp}}^{i,N,M}(\vartheta_t)\,\mathrm d t
-\gamma_t F_h^{i,N,M}(\vartheta_t,\widetilde{\boldsymbol{z}}_t^{N,M})\,\mathrm d t
+\gamma_t g^{i,j,N,M}(\vartheta_t,\widetilde{\boldsymbol{z}}_t^{N,M})\sigma^{-\top}\,\mathrm d w_t^{i,N},
\label{eq:appendix-alg-h}
\end{align}
where we have introduced the notation
\begin{equation}
F_H^{i,N,M}(\theta,z):=H^{i,N,M}(\theta,z)-\partial_\theta \mathcal J_{\mathrm{vp}}^{i,N,M}(\theta),
\qquad
F_h^{i,N,M}(\theta,z):=h^{i,j,k,N,M}(\theta,z)-\partial_\theta \mathcal J_{\mathrm{vp}}^{i,N,M}(\theta).
\end{equation}
In addition, due to Proposition~\ref{prop:asymptotic-partial-log-lik-grad}, the functions $F_H^{i,N,M}$ and $F_h^{i,N,M}$ are both centred with respect to the invariant distribution $\Pi_\theta^{N,M}$.

\begin{lemma}
\label{lemma:fluctuations}
Suppose that Assumptions~\ref{assumption:learning-rate} - \ref{assumption:pgp} hold. Let $N,M\in\mathbb{N}$ be fixed. In addition, fix $i\in[N]$, and $j,k\in[M]$. Then
\begin{equation}
\int_0^t \gamma_s F^{i,N,M}_H(\theta_s,\boldsymbol{z}_s^{N,M})\,ds
\qquad\text{and}\qquad
\int_0^t \gamma_s F^{i,N,M}_h(\vartheta_s,\widetilde{\boldsymbol{z}}_s^{N,M})\,ds
\end{equation}
converge almost surely as $t\to\infty$. In particular, for any stopping times
$\tau_n\le \sigma_n\to\infty$ such that $
\sup_n \int_{\tau_n}^{\sigma_n}\gamma_s\,ds < \infty$, it holds almost surely that
\begin{equation}
\int_{\tau_n}^{\sigma_n}\gamma_s F^{i,N,M}_H(\theta_s,\boldsymbol{z}_s^{N,M})\,ds \to 0,
\qquad
\int_{\tau_n}^{\sigma_n}\gamma_s F^{i,N,M}_h(\vartheta_s,\widetilde{\boldsymbol{z}}_s^{N,M})\,ds \to 0.
\end{equation}
\end{lemma}

\begin{proof}
We prove the statement for $F^{i,N,M}_H$; the proof for $F^{i,N,M}_h$ is identical. The proof follows closely the proof of \citet[][Lemma 3.1]{sirignano2017stochastic}; see also \cite[][Lemma 1]{surace2019online}. In the interest of brevity, we just outline the main details.

By Proposition~\ref{prop:asymptotic-partial-log-lik-grad}, Proposition~\ref{prop:j-vp}, Assumption~\ref{assumption:poisson}(ii)--(iii), and
Assumption~\ref{assumption:pgp}, the map $\smash{\theta\mapsto \partial_\theta \mathcal J^{i,N,M}_{\mathrm{vp}}(\theta)}$ is $C^1$, with first derivative uniformly bounded on $\Theta$. Indeed, writing
$\smash{\Sigma_{\theta,\ell}:=\partial_{\theta_\ell}\Pi_\theta^{N,M}}$, we have componentwise that 
\begin{equation}
\partial_{\theta_\ell}\partial_{\theta_m}\mathcal J_{i,N,M}^{\mathrm{vp}}(\theta)
=
\int_{\mathbb R^K}\partial_{\theta_\ell}H_{i,N,M,m}(\theta,z)\,\Pi_\theta^{N,M}(dz)
+
\int_{\mathbb R^K}H_{i,N,M,m}(\theta,z)\,\Sigma_{\theta,\ell}(dz),
\end{equation}
with both terms uniformly bounded by the IPS-PGP bounds from Assumption~\ref{assumption:pgp}, together with the uniform moment bounds in Assumption~\ref{assumption:poisson}(ii)--(iii). Consequently, the function 
\begin{equation}
F^{i,N,M}_H(\theta,\boldsymbol z)=H^{i,N,M}(\theta,\boldsymbol z)-\partial_\theta \mathcal J^{i,N,M}_{\mathrm{vp}}(\theta)
\end{equation}
belongs componentwise to the centred class covered by Assumption~\ref{assumption:poisson}(iv). Thus, the Poisson equation
\begin{equation}
\mathcal A_\theta^{N,M}v(\theta,\cdot)=F_H^{i,N,M}(\theta,\cdot)
\end{equation}
has a unique solution with the regularity and polynomial growth bounds stated in Assumption~\ref{assumption:poisson}(iv). Applying It\^o's formula to $\gamma_t v(\theta_t,\boldsymbol{z}_t^{N,M})$, we obtain the decomposition
\begin{equation}
\int_0^t \gamma_s F^{i,N,M}_H(\theta_s,\boldsymbol{z}_s^{N,M})\,ds
=
\text{boundary term}
+
\dot\gamma\text{-term}
+
\gamma^2\text{-drift terms}
+
\text{martingale terms}.
\end{equation}
The polynomial growth bounds from Assumption~\ref{assumption:poisson}(iv), the moment bounds in
Assumption~\ref{assumption:poisson}(v), and the summability conditions in Assumption~\ref{assumption:learning-rate} imply that every drift term
on the right-hand side is absolutely integrable on $[0,\infty)$ and every martingale term has
finite quadratic variation. Hence each term converges almost surely, which proves the first claim.
The interval version follows by taking differences of the convergent process.
\end{proof}

\begin{proof}[Proof of Proposition~\ref{prop:finite-n-m-convergence}]
We establish the result for $(\theta_t)_{t\ge0}$. The proof for $(\vartheta_t)_{t\ge0}$ is identical after replacing $F^{i,N,M}_H$ by $F^{i,N,M}_h$. 

The proof follows the argument introduced by \citet{sirignano2017stochastic}; see also \citet{surace2019online,sharrock2023online}. Similar to above, we here just outline the main points. First, by definition, we have that
\begin{equation}
\mathcal J^{i,N,M}_{\mathrm{vp}}(\theta)\ge 0 \qquad \text{for all $\theta\in\Theta$}. 
\end{equation}
Second, arguing as in the proof of Lemma~\ref{lemma:fluctuations}, the norm of the gradient and the Hessian are both bounded above, viz
\begin{equation}
\sup_{\theta\in\Theta}\|\partial_\theta \mathcal J^{i,N,M}_{\mathrm{vp}}(\theta)\| < \infty,
\qquad
\sup_{\theta\in\Theta}\|\nabla_\theta^2 \mathcal J^{i,N,M}_{\mathrm{vp}}(\theta)\| < \infty.
\end{equation}
Third, due to Lemma~\ref{lemma:fluctuations}, the fluctuation term in \eqref{eq:appendix-alg-H} is asymptotically negligible over stopping intervals whose
$\int \gamma_s\,ds$-length is uniformly bounded. Finally, for the martingale term, Assumptions~\ref{assumption:learning-rate}, \ref{assumption:poisson}(v), and~\ref{assumption:pgp} imply
\begin{equation}
\int_0^\infty \gamma_s^2 \|G^{i,N,M}(\theta_s,\boldsymbol z_s^{N,M})\|^2\,ds<\infty
\qquad\text{almost surely.}
\end{equation}
Thus, in particular, the martingale term in the update equation \eqref{eq:appendix-alg-H} has finite quadratic variation on $[0,\infty)$.

We can now apply the standard stopping-time argument in \citet{sirignano2017stochastic,surace2019online}. For any fixed $\kappa>0$, we begin by defining the usual cycle of random times corresponding to periods of time for which $\|\partial_\theta \mathcal J^{i,N,M}_{\mathrm{vp}}(\theta_t)\|\in[0,\kappa)$ or otherwise. Each
large-gradient cycle results in a deterministic decrease in $\mathcal J^{i,N,M}_{\mathrm{vp}}$, up to error terms which vanish almost surely by Lemma~\ref{lemma:fluctuations} and the martingale estimate above. Since $\mathcal J_{i,N,M}^{\mathrm{vp}}\ge0$, infinitely many such cycles are impossible. Therefore, almost
surely, only finitely many excursions above level $\kappa$ can occur. Since $\kappa>0$ was chosen
arbitrarily, this establishes that
\begin{equation}
\lim_{t\to\infty}\|\partial_\theta \mathcal J^{i,N,M}_{\mathrm{vp}}(\theta_t)\|=0.
\end{equation}
\end{proof}

\begin{proof}[Proof of Proposition~\ref{prop:inf-n-m-convergence}]
We once again prove the result for $(\theta_t^{N,M})_{t\geq 0}$; the argument for $(\vartheta_t^{N,M})_{t\ge0}$ is identical. By the triangle inequality,
\begin{equation}
\limsup_{t\to\infty}\|\partial_\theta \mathcal J(\theta_t^{N,M})\|
\le
\sup_{\theta\in\Theta}\|\partial_\theta \mathcal J_{\mathrm{vp}}^{i,N,M}(\theta)-\partial_\theta \mathcal{J}(\theta)\|
+
\limsup_{t\to\infty}\|\partial_\theta \mathcal J_{\mathrm{vp}}^{i,N,M}(\theta_t^{N,M})\|.
\end{equation}
The second term is zero almost surely by Proposition~\ref{prop:finite-n-m-convergence}, while the first is bounded by $C(\varepsilon_N+\delta_M)$ by Proposition~\ref{prop:inf-n-convergence-1}. Hence
\begin{equation}
\limsup_{t\to\infty}\|\partial_\theta \mathcal J(\theta_t^{N,M})\|
\le C(\varepsilon_N+\delta_M)
\qquad\text{almost surely.}
\end{equation}
The claim now follows by sending $N,M\to\infty$.
\end{proof}

\begin{proof}[Proof of Corollary~\ref{cor:optional-iterated-parameter-consistency}]
Similar to the previous results, we establish the result for $(\theta_t^{N,M})_{t\geq 0}$. Fix $\varepsilon>0$ and define
\begin{equation}
c_\varepsilon
:=
\inf\Bigl\{
\|\partial_\theta \mathcal{J}(\theta)\|
:
\theta\in\Theta,\ \|\theta-\theta_0\|\ge \varepsilon
\Bigr\}
>0.
\end{equation}
By Proposition~\ref{prop:inf-n-m-convergence}, almost surely there exist $N_0,M_0$ such that for all $N\ge N_0$ and $M\ge M_0$,
\begin{equation}
\limsup_{t\to\infty}\|\partial_\theta \mathcal J(\theta_t^{N,M})\|<c_\varepsilon,
\qquad
\limsup_{t\to\infty}\|\partial_\theta \mathcal J(\vartheta_t^{N,M})\|<c_\varepsilon.
\end{equation}
Suppose, for contradiction, that for some such $N,M$,
\begin{equation}
\limsup_{t\to\infty}\|\theta_t^{N,M}-\theta_0\|\ge \varepsilon.
\end{equation}
Then there exists a sequence $t_n\to\infty$ such that $\|\theta_{t_n}^{N,M}-\theta_0\|\ge \varepsilon$ for all $n$. By the definition of $c_\varepsilon$, it follows that $\smash{ \|\partial_\theta \mathcal J(\theta_{t_n}^{N,M})\|\ge c_\varepsilon}$ for all $n$. But this contradicts $\limsup_{t\to\infty}\|\partial_\theta \mathcal J(\theta_t^{N,M})\|<c_\varepsilon$. Hence
\begin{equation}
\limsup_{t\to\infty}\|\theta_t^{N,M}-\theta_0\|<\varepsilon.
\end{equation}
Since $\varepsilon>0$ was arbitrary, the claim follows.
\end{proof}

\section{Proofs for Section~\ref{sec:numerics-quadratic}}
\label{appendix:quadratic-model-proofs}

In this appendix, we prove the closed-form formulas for the one-dimensional quadratic confinement and quadratic interaction model (cf. Section~\ref{sec:numerics-quadratic}). Let $\alpha:=\theta_1+\theta_2$ and $\alpha_0:=\theta_{0,1}+\theta_{0,2}$. Throughout, we assume $\theta_{0,1}>0$ and $\alpha_0>0$, so that the true finite-$N$ system and the mean-field limit admit the invariant Gaussian laws used below. In addition, whenever formulas are evaluated at a generic parameter $\theta$, we assume \(\theta_1>0\) and \(\alpha>0\).

\begin{proposition}
\label{prop:appendix-quadratic-meanfield}
Suppose that \(\theta_1>0\) and \(\alpha:=\theta_1+\theta_2>0\).  
The invariant mean-field law for the model with quadratic confinement and quadratic interaction is $\pi_\theta = \mathcal N\!\left(0,\frac{\sigma^2}{2\alpha}\right)$. In addition, the mean-field objective is given exactly by
\begin{equation}
\mathcal{J}(\theta)=\frac{(\alpha-\alpha_0)^2}{4\alpha_0}.
\end{equation}
\end{proposition}

\begin{proof}
The mean-field model is given by
\begin{equation}
\mathrm d x_t^\theta
=
\bigl[-\theta_1 x_t^\theta-\theta_2(x_t^\theta-m_t)\bigr]\mathrm d t
+
\sigma\,\mathrm d w_t,
\qquad
m_t:=\mathbb E[x_t^\theta].
\end{equation}
At stationarity, the mean $m_t$ is constant and therefore satisfies $0=-\theta_1 m_t$, so $m_t=0$. The stationary mean-field dynamics are thus given by
\begin{equation}
\mathrm d x_t^\theta=-\alpha x_t^\theta\,\mathrm d t+\sigma\,\mathrm d w_t,
\end{equation}
This is just an Ornstein-Uhlenbeck (OU) process, whose unique invariant law is $\mathcal N(0,\frac{\sigma^2}{2\alpha})$. It follows that $B(\theta,x,\pi_\theta) = -\theta_1 x-\theta_2\bigl(x-\pi_\theta[\mathrm{id}]\bigr) = -\alpha x$ and thus
\begin{equation}
\mathcal{J}(\theta)
=
\frac1{2\sigma^2}
\int_{\mathbb R}
\bigl((\alpha-\alpha_0)x\bigr)^2\,\pi_{\theta_0}(\mathrm d x)=
\frac1{2\sigma^2}
(\alpha-\alpha_0)^2
\frac{\sigma^2}{2\alpha_0}
=
\frac{(\alpha-\alpha_0)^2}{4\alpha_0}.
\end{equation}
\end{proof}

\begin{lemma}
\label{lem:appendix-quadratic-covariances}
Let $x^N=(x^{1,N},\dots,x^{N,N})\sim\pi_{\theta_0}^N$ be the invariant law of the true finite-$N$ system. Define $\bar x^N:=\frac1N\sum_{a=1}^N x^{a,N}$ and $\Xi^{i,N}:=x^{i,N}-\bar x^N$. Then, for every $i\in[N]$,
\begin{equation}
\mathbb E[\bar x^N]=0,~
\mathrm{Var}(\bar x^N)=\frac{\sigma^2}{2N\theta_{0,1}}
, \qquad 
\mathbb E[\Xi^{i,N}]=0,~
\mathrm{Var}(\Xi^{i,N})=\frac{N-1}{N}\frac{\sigma^2}{2\alpha_0}
,\qquad
\mathrm{Cov}(\bar x^N,\Xi^{i,N})=0.
\end{equation}
Thus, in particular, 
\begin{equation}
V_N:=\mathrm{Var}(x^{i,N})
=
\frac{\sigma^2}{2N\theta_{0,1}}
+
\frac{N-1}{N}\frac{\sigma^2}{2\alpha_0},
\qquad
C_N:=\mathrm{Cov}(x^{i,N},\bar x^N)
=
\frac{\sigma^2}{2N\theta_{0,1}}.
\end{equation}
\end{lemma}

\begin{proof}
By summing the SDEs for each particle in the IPS, and then dividing by $N$, we see that the empirical mean evolves according to
\begin{align}
\mathrm d\bar x_t^N=
-\theta_{0,1}\bar x_t^N\,\mathrm d t
+
\frac{\sigma}{N}\sum_{a=1}^N \mathrm d w_t^{a,N} =
-\theta_{0,1}\bar x_t^N\,\mathrm d t
+
\frac{\sigma}{\sqrt N}\,\mathrm d b_t,
\end{align}
where $b=(b_t)_{t\geq 0}$ denotes a standard Brownian motion. Thus, $(\bar x_t^N)_{t\geq 0}$ is an OU process with invariant law $\mathcal N(0,\sigma^2/(2N\theta_{0,1}))$. This proves the first result. Next observe that
\begin{align*}
\mathrm d\Xi_t^{i,N}=
\mathrm d x_t^{i,N}-\mathrm d\bar x_t^N=
-\alpha_0 \Xi_t^{i,N}\,\mathrm d t
+
\sigma\big(
\mathrm d w_t^{i,N}
-
\frac1N\sum_{a=1}^N \mathrm d w_t^{a,N}
\big).
\end{align*}
The driving martingale in the last display has quadratic variation equal to $\smash{(1-\frac1N)\mathrm d t =\frac{N-1}{N}\,\mathrm d t}$. Thus, $(\Xi_t^{i,N})_{t\geq 0}$ is also an OU process, this time with invariant law $\smash{\mathcal{N}(0,\frac{N-1}{N}\frac{\sigma^2}{2\alpha_0})}$.

Finally, the quadratic covariation of the two driving martingales is zero, and the drift equations are decoupled. Applying It\^o's formula to $\bar x_t^N\Xi_t^{i,N}$, we have
\begin{equation}
\frac{\mathrm d}{\mathrm d t}\mathbb E[\bar x_t^N\Xi_t^{i,N}]
=
-(\theta_{0,1}+\alpha_0)\mathbb E[\bar x_t^N\Xi_t^{i,N}].
\end{equation}
Thus, at stationarity, we must have $\mathbb E[\bar x^N\Xi^{i,N}]=0$. Finally, the formulas for $V_N$ and $C_N$ follow from the fact that $x^{i,N}=\bar x^N+\Xi^{i,N}$.
\end{proof}

\begin{proposition}
\label{prop:appendix-quadratic-finiteN}
For every finite $N$ and every $M\in\mathbb N$, the finite-particle surrogate objective is exactly independent of $M$, and given by
\begin{equation}
\mathcal J_{\mathrm{vp}}^{i,N,M}(\theta)=\mathcal J_N^i(\theta)=
\frac{(N-1)\theta_{0,2}^2}{4N\bigl(N\theta_{0,1}+\theta_{0,2}\bigr)}
+
\frac{N\theta_{0,1}+\theta_{0,2}}{4N\alpha_0\theta_{0,1}}
\bigl(\alpha-\alpha_N^\star\bigr)^2,
\end{equation}
where the pseudo-minimiser $\alpha_N^{\star}$ is given by
\begin{equation}
\alpha_N^\star
=
\frac{N\theta_{0,1}\alpha_0}{N\theta_{0,1}+\theta_{0,2}}
=
\alpha_0-\frac{\alpha_0\theta_{0,2}}{N\theta_{0,1}+\theta_{0,2}}.
\end{equation}
Meanwhile, if $\theta_2$ is known and only $\theta_1$ is estimated, then the finite-$N$ pseudo-minimiser is
\begin{equation}
\theta_{1,N}^\star
=
\alpha_N^\star-\theta_{0,2}
=
\frac{N\theta_{0,1}^2-\theta_{0,2}^2}{N\theta_{0,1}+\theta_{0,2}}
=
\theta_{0,1}-\frac{\theta_{0,2}(\theta_{0,1}+\theta_{0,2})}{N\theta_{0,1}+\theta_{0,2}};
\end{equation}
Finally, if $\theta_1$ is known and only $\theta_2$ is estimated, then the finite-$N$ pseudo-minimiser is
\begin{equation}
\theta_{2,N}^\star
=
\alpha_N^\star-\theta_{0,1}
=
\frac{(N-1)\theta_{0,1}\theta_{0,2}}{N\theta_{0,1}+\theta_{0,2}}
=
\theta_{0,2}-\frac{\theta_{0,2}(\theta_{0,1}+\theta_{0,2})}{N\theta_{0,1}+\theta_{0,2}}.
\end{equation}
\end{proposition}

\begin{proof}
Fix $M\in\mathbb N$. For the frozen virtual $M$-particle system evaluated at $\theta$, exchangeability implies that each stationary marginal has the same mean, say $m_\theta^M$. Thus, averaging the $M$ drift equations and taking expectations at stationarity gives $0=-\theta_1 m_\theta^M$, and so $m_\theta^M=0$. This implies, in particular, that
\begin{equation}
B(\theta,x,\pi_\theta^{1,M})
=
-\theta_1 x-\theta_2(x-m_\theta^M)
=
-\alpha x
=
B(\theta,x,\pi_\theta).
\end{equation}
Hence $\mathcal J_{\mathrm{vp}}^{i,N,M}(\theta)=\mathcal J_N^i(\theta)$ for all $\theta\in\Theta$. Suppose now that $x^N\sim\pi_{\theta_0}^N$. The observed finite-$N$ drift under the true parameter is given by
\begin{equation}
B(\theta_0,x^{i,N},\mu^N)
=
-\theta_{0,1}x^{i,N}-\theta_{0,2}(x^{i,N}-\bar x^N)
=
-\alpha_0 x^{i,N}+\theta_{0,2}\bar x^N,
\end{equation}
Meanwhile, the mean-field drift at $\pi_{\theta}$ is given by $B(\theta,x^{i,N},\pi_\theta)=-\alpha x^{i,N}$. Thus, using the definition of $\mathcal J_N^i$, we have that
\begin{equation}
\mathcal J_N^i(\theta)
=
\frac1{2\sigma^2}
\mathbb E\Bigl[
\bigl((\alpha_0-\alpha)x^{i,N}-\theta_{0,2}\bar x^N\bigr)^2
\Bigr].
\end{equation}
Expanding the square and using the results obtained in  Lemma~\ref{lem:appendix-quadratic-covariances}, it follows that
\begin{equation}
\mathcal J_N^i(\theta)
=
\frac1{2\sigma^2}
\Bigl(
(\alpha-\alpha_0)^2V_N
+
2(\alpha-\alpha_0)\theta_{0,2}C_N
+
\theta_{0,2}^2C_N
\Bigr).
\end{equation}
Since the coefficient of $(\alpha-\alpha_0)^2$ is positive, $\mathcal J_N^i$ is a strictly convex quadratic function of $\alpha$, and its unique minimising value of $\alpha$ is obtained by setting the derivative with respect to $\alpha$ equal to zero: $2(\alpha-\alpha_0)V_N+2\theta_{0,2}C_N=0$. Thus
\begin{equation}
\alpha_N^\star=\alpha_0-\theta_{0,2}\frac{C_N}{V_N}.
\end{equation}
Using the explicit formulae from Lemma~\ref{lem:appendix-quadratic-covariances}, we can calculate the ratio of the covariance and the variance as
\begin{equation}
\frac{C_N}{V_N}
=
\frac{\frac1{\theta_{0,1}}}
{\frac1{\theta_{0,1}}+\frac{N-1}{\alpha_0}}
=
\frac{\alpha_0}{N\theta_{0,1}+\theta_{0,2}},
\end{equation}
Thus, substituting into the previous expression, the pseudo-minimiser $\alpha_N^{*}$ is given by
\begin{equation}
\alpha_N^\star
=
\alpha_0\left(1-\frac{\theta_{0,2}}{N\theta_{0,1}+\theta_{0,2}}\right)
=
\frac{N\theta_{0,1}\alpha_0}{N\theta_{0,1}+\theta_{0,2}}.
\end{equation}
Finally, substituting this value back into the quadratic polynomial in $\alpha$ gives the completed-square form, viz
\begin{equation}
\mathcal J_N^i(\theta)
=
\frac{(N-1)\theta_{0,2}^2}{4N\bigl(N\theta_{0,1}+\theta_{0,2}\bigr)}
+
\frac{N\theta_{0,1}+\theta_{0,2}}{4N\alpha_0\theta_{0,1}}
\bigl(\alpha-\alpha_N^\star\bigr)^2.
\end{equation}
The last two statements in the proposition follow straightforwardly, restricting the previous result to the affine lines $\theta_2=\theta_{0,2}$ and $\theta_1=\theta_{0,1}$ respectively.
\end{proof}

\end{document}